\documentclass[12pt]{article}
\usepackage{amsmath}
\usepackage{graphicx}
\usepackage{enumerate}
\usepackage{natbib}
\usepackage{url} 

\newcommand{\blind}{1}

\usepackage{chngcntr}
\usepackage{apptools}
\usepackage{xr}
\makeatletter

\newcommand*{\addFileDependency}[1]{
\typeout{(#1)}
\@addtofilelist{#1}
\IfFileExists{#1}{}{\typeout{No file #1.}}
}\makeatother

\usepackage{color, amssymb, bbm, bm, times, subfig}
\usepackage{booktabs}
\usepackage{latexsym}
\usepackage{comment}
\usepackage{enumitem}
\usepackage{mathtools}
\usepackage{multirow}
\usepackage{algorithm}
\usepackage{algpseudocode}
\usepackage[normalem]{ulem}
\usepackage[english]{babel}
\usepackage{amsthm}
\usepackage{threeparttable}

\def\T{{ \mathrm{\scriptscriptstyle T} }}

\def \hat{\widehat}
\def \tilde{\widetilde}
\def \bA{ {A} }

\def \bD{ {D} }

\def \bE{ {E} }
\def \bF{ {F} }
\def \be{ {e} }
\def \bare{ \bar{{e}} }
\def \bxi{ {\xi} }
\def \bI{ {I} }
\def \bG{ {G} }

\def \bQ{ {Q} }
\def \bR{ {R} }
\def \bS{ {S} }

\def \bV{ {V} }

\def \bX{ {X} }
\def \bx{ {x} }
\def \bY{ {Y} }

\def \balpha{{\alpha}}

\def \bbeta{{\beta}}

\def \bgamma{{\gamma}}
\def \bmu{{\mu}}
\def \bnu{{\nu}}
\def \bPsi{{\Psi}}

\def \btheta{{\theta}}
\def \bomega{{\omega}}
\def \baromega{\bar{{\omega}}}

\def \argmin{\mathrm{argmin}}
\def \mR{\mathbb{R}}

\def \hbeta{\widehat{{\beta}}}
\def \tbeta{\widetilde{{\beta}}}
\def \omegahat{\widehat{{\omega}}}
\def \homega{\widehat{{\omega}}}

\def \bDelta{{\Delta}}
\def \bSigma{{\Sigma}}
\def \barR{\bar{{R}}}
\def \hatR{\widehat{{R}}}
\def \barS{\bar{{S}}}
\def \barPsi{\bar{\bPsi}}
\def \O{O}
\def \OP{{O}_p}
\def \oP{{o}_p}
\def \E{{\mathbb{E}}}
\def \PP{pr}
\def \calS{\mathcal{S}}
\def \calD{\mathcal{D}}
\def \beps{{\varepsilon}}
\def \Var{\mathrm{var}}
\def \diag{\mathrm{diag}}
\def \trace{\mathrm{tr}}
\def \rmF{\mathrm{F}}
\def \lambdamin{\lambda_{\mathrm{min}}}

\theoremstyle{plain}
\newtheorem{assumption}{Assumption}
\newtheorem{theorem}{Theorem}

\newtheorem{lemma}[theorem]{Lemma}
\newtheorem{proposition}[theorem]{Proposition}

\theoremstyle{remark}
\newtheorem{remark}{Remark}

\newcommand{\lu}[1]{\textcolor{black}{#1}}

\newcommand{\add}[1]{\textcolor{black}{#1}}

\begin{document}


\def\spacingset#1{\renewcommand{\baselinestretch}%
{#1}\small\normalsize} \spacingset{1}
\renewcommand{\arraystretch}{0.6}


\if1\blind
{
  \title{Inference for linear functionals of high-dimensional longitudinal proteomics data using generalized estimating equations}
  \author{Lu Xia \\[0.2em]
    Department of Statistics and Probability, Michigan State University \\[0.5em]
    Ali Shojaie
    \\[0.2em]
    Department of Biostatistics, University of Washington}
  \maketitle
} \fi

\if0\blind
{
  \bigskip
  \bigskip
  \bigskip
  \begin{center}
    {\LARGE\bf Inference for linear functionals of high-dimensional longitudinal proteomics data using generalized estimating equations}
\end{center}
  \medskip
} \fi

\thispagestyle{empty}

\begin{abstract}
Regression analysis of correlated data, where multiple correlated responses are recorded on the same unit, is ubiquitous in many scientific areas.  
With the advent of new technologies, in particular high-throughput omics profiling assays, such correlated data increasingly consist of large number of variables compared with the available sample size. 
Motivated by recent longitudinal proteomics studies of COVID-19, we propose a novel inference procedure for linear functionals
of high-dimensional regression coefficients in  generalized estimating equations, which are widely used to analyze correlated data. Our estimator for this more general inferential target, obtained via constructing projected estimating equations, is shown to be asymptotically normally distributed under mild regularity conditions. We also introduce a data-driven cross-validation procedure to select the tuning parameter for estimating the projection direction, which is not addressed in the existing procedures. 
We illustrate the utility of the proposed procedure in providing confidence intervals for associations of individual proteins and severe COVID risk scores obtained based on high-dimensional proteomics data, and demonstrate its robust finite-sample performance, especially in estimation bias and confidence interval coverage, via extensive simulations.
\end{abstract}


\noindent%
\vfill

\newpage
\spacingset{1.9} 

\setcounter{page}{1}
\section{Introduction}
\label{sec:intro}

Longitudinal \emph{omics} data, including longitudinal proteomic profiling, are increasingly used to improve the understanding of molecular mechanisms of various chronic and infectious diseases and their prognosis. An important case study is the Coronavirus Disease 2019 (COVID-19), caused by infection with the newly discovered severe acute respiratory syndrome coronavirus 2 (SARS-Cov-2). A deeper understanding of the molecular mechanisms of COVID-19 would reveal biological insights into the immunopathology of severe infection, and facilitate timely development of effective therapies. Moreover, developing molecular-level prognostic tools would fulfill the urgent need of accurately identifying patients with poor prognosis early in the course of the disease to prevent severe outcomes such as intubation and death.
 
Recent studies have adopted the longitudinal proteomic profiling approach for COVID-19 \citep{filbin2021longitudinal,haljasmagi2020longitudinal} and post-acute sequelae of COVID-19 \citep{yin2023long}. For example, \citet{filbin2021longitudinal} collected longitudinal data on 305 patients with COVID-19 infection and examined the associations of COVID-19 infection and severity with over a thousand plasma proteins measured using the Olink assay 
(Olink Proteomics, Uppsala, Sweden). A cross-sectional study was later conducted to investigate COVID-19 severity \citep{feyaerts2022integrated}. These studies were used to identify molecular biomarkers and pathways and develop prognostic scores for disease severity. Principal component analysis (PCA) based on 853 measured proteins from 64 COVID-19  patients in the cross-sectional study \citep{feyaerts2022integrated} shows some separation of severe COVID-19 cases from mild and moderate cases (Figure~\ref{fig:intro}a). 

As an initial investigation, we used a na\"ive Lasso logistic regression, under working independence, to identify proteins associated with severe COVID-19 and develop risk scores for patients at high risk of infection using the longitudinal data in \citet{filbin2021longitudinal}. 
We then used the proteomic data on 64 (new) subjects in \citet{feyaerts2022integrated} to obtain risk scores for COVID-19 severity for these patients based on the previously-obtained Lasso estimates.
While these risk scores offer insight into patients' risk, they would be of limited use in clinical settings without measures of uncertainty. The quantification of uncertainty is particularly important when determining whether a patient is at risk of more severe conditions or may have poor prognosis, which is in turn crucial for making clinical decision in resource-constrained settings, or when treatments can have adverse side effects.

\begin{figure}[t!] 
\centering
\includegraphics[width=0.8\textwidth]{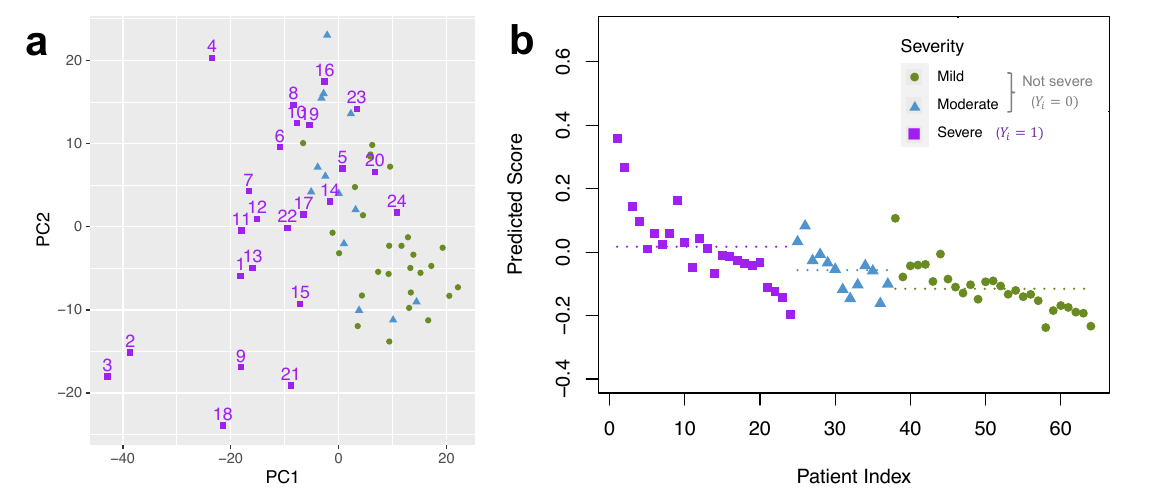}
\caption{\small (a) The first two principal components from principal component analysis (PCA) based on 853 measured proteins from 64 COVID-19 positive patients in \citet{feyaerts2022integrated}, including 24 severe (numbered), 13 moderate and 27 mild cases. (b) Predicted scores for severity versus moderate or mild for 64 COVID-19 positive patients in \citet{feyaerts2022integrated}, based on the regression coefficients fitted in a naive Lasso logistic regression under working independence using the longitudinal data in \citet{filbin2021longitudinal}. Mean scores in each group are highlighted in dotted line segments.}
\label{fig:intro}
\end{figure}

Another key limitation of the na\"ive Lasso-based scores in Figure~\ref{fig:intro}b is that they do not account for the correlation among longitudinal observations. A more suitable alternative is to use methods that 
account for the correlation between responses, including the well-received generalized estimating equations \citep{liang1986longitudinal}, quadratic inference functions \citep{qu2000improving} and mixed effects models \citep{gelman2007data}. Recognizing this need, \citet{filbin2021longitudinal} applied linear mixed models with individual proteins as outcome to identify proteins that are associated with COVID-19 infection and severity. 
However, high-throughput longitudinal proteomic data provide a unique opportunity to simultaneously assess the associations of hundreds to thousands of protein markers and develop more accurate prognostic tools, which can help identify directly relevant biomarkers and reduce potential confounding \citep{wang2011gee}. 

A common feature of longitudinal proteomic studies, including the ones discussed above, is that the number of covariates $p$ can be large compared to the sample size $n$, or far exceed $n$. 
In such settings, accurate biomarker identification requires statistical methods that can (a) reliably estimate the associations of features with outcomes of interest without incurring large biases, and (b) quantify their uncertainty in the presence of a large number of variables. 
Moreover, in many clinical applications, data from new patients are used to obtain predictive risk scores for prognosis based on previously-developed models. Uncertainty quantification in such settings is crucial in clinical decision making and requires inference---e.g., confidence intervals---for linear functionals of high-dimensional parameters, which is a challenging and understudied problem beyond linear models. 
Achieving these two goals simultaneously requires development of a unified and reliable statistical estimation and inference approach.

In this article, we consider marginal approaches for correlated data, which, compared with conditional approaches, such as linear mixed-effects models, provide inference under fewer modeling assumptions. 
In particular, we focus on generalized estimating equations (GEE) with potentially misspecified working correlation structure for responses observed within the same subject. When the number of variables is small, the GEE can achieve consistency and improved efficiency for estimating the regression coefficients compared with ignoring within-unit correlations \citep{liang1986longitudinal}.
As an alternative to GEE, the quadratic inference functions (QIF) approach \citep{qu2000improving} has gained popularity due to improved efficiency. 

Despite considerable work on regularized estimation for high-dimensional correlated data (see the ``Related literature" subsection), inference has received less attention. A notable exception, closely related to our work, is the  quadratic decorrelated inference functions (QDIF) approach for low-dimensional parameters in high-dimensional longitudinal data \citep{fang2020test}. 
The QDIF offers a flexible inference framework for correlated data. However, QIF-based approaches may be less reliable in certain settings, such as imbalanced cluster sizes \citep{khajeh2011comparison, westgate2012effect}. 
Moreover, QDIF does not provide inference for linear functionals, which is needed for inference on risk scores in our application. To overcome these limitations, in this paper we develop an inference framework for GEE with high-dimensional covariates, which we refer to as HDIGEE. In addition to its wide adoption in biomedical research, GEE offers flexibility in handling irregular observations and can thus be a viable alternative to QDIF in such settings. We also develop a statistically rigorous approach for inference on flexible linear functionals of regression parameters, which encompasses individual coefficients and predictive scoring using a set of proteins as special cases. Moreoever, the proposed framework is not only limited to longitudinal data, but provides an example solution for more general estimating equation problems.  


To develop our inference framework, termed HDIGEE, for regression parameters in GEE, we construct projected estimating equations for any linear functional of regression parameters (Section~\ref{sec:method}). The projection direction vector is estimated through a linear programming system resembling the column-wise sparse precision matrix estimation in the constrained $\ell_1$ minimization method \citep{cai2011constrained} and the direction estimation in \citet{neykov2018unified}. 
In Section~\ref{sec:prac}, we discuss implementation issues, in particular, the selection of tuning parameters, which, despite practical importance, has not been carefully studied. We design a new cross-validation procedure to satisfy the needs of tuning parameter selection while reasonably estimating the projection direction for inference purposes.  
Unlike the existing high-dimensional inference procedures for correlated data, by utilizing the least favourable sub-model approach in semi-parametric inference theory \citep{zhang2011statistical}, our theoretical analysis in Section~\ref{sec:theory} enables reliable estimation and inference, particularly valid confidence intervals, on a broader range of linear functionals of the regression parameters. 
In Section~\ref{sec:app}, we illustrate the utility of the proposed framework in our longitudinal proteomic application, where we identify biomarker of COVID-19 severity and develop inference on patient risk scores.
Extensive simulation studies in Section~\ref{sec:sim} show that the proposed method performs at least as well as the quadratic de-correlated inference function approach \citep{fang2020test}. Moreover, in settings with binary outcomes, or irregularly measured time points mimicking the riboflavin production data, HDIGEE can offer more reliable coefficient, variance estimates and confidence interval coverage than QDIF. 



\subsection*{{Related literature}}\label{sec:lit}

 {In addition to QDIF \citep{fang2020test}, our work is related to two classes of existing approaches: penalized estimation methods for correlated data and high-dimensional inference procedures.}

Penalized estimation methods for correlated data have received considerable attention over the years. \citet{pan2001akaike} and \citet{wang2009consistent} devised information criteria based on quasi-likelihood and quadratic inference functions, respectively. \citet{dziak2009overview} considered penalized QIF with the smoothly clipped absolute deviation penalty. These methods were developed for the setting with fixed numbers of covariates. \citet{wang2012penalized} proposed a penalized GEE procedure with a non-convex penalty in the setting where $p$ is at most in the same order of $n$. \citet{fan2012variable} discussed selection and estimation of both fixed and random effects in linear mixed effects models.

Among various high-dimensional inference procedures, lasso de-biasing has gained popularity in linear regression \citep{javanmard2014confidence,van2014asymptotically,zhang2014confidence} and generalized linear models \citep{van2014asymptotically}. 
\citet{ning2017general} developed an inferential approach for generic penalized M-estimators based on de-correlated score functions. \citet{neykov2018unified} proposed a unified inferential framework for high-dimensional estimating equations, which can be applied to problems such as linear regression, undirected graphical models and instrumental variable regression. In earlier work on correlated data,  \citet{wang2012penalized} established the oracle property and asymptotic distribution of penalized GEE. However, this post-selection inference procedures is conditional on the correct model being selected, which does not take into account the variability associated with model selection nor the possibility of selecting an incorrect model \citep{zhao2021defense}. 
The recent work on high-dimensional linear mixed effect models \citep{lin2020statistical, bradic2020fixed, li2021inference} provides inference on high-dimensional fixed effects through conditional models, which involve additional modeling assumptions. Moreover, these procedures only address inference on individual effects and do not provide (explicit) inference for linear functionals, which is essential for the risk scores in our application. On the other hand, earlier efforts for inference on linear functionals \citep{zhu2018linear,zhao2023estimation} have only focused on linear models, not generalized linear models or models for correlated data.

\section{Inference for Linear Functionals in High-Dimensional Generalized Estimating Equations}
\label{sec:method}
	
\subsection{Preliminaries}
	
Let $| \mathcal{S} |$ and denotes the cardinality and complement of a set $\mathcal{S}$, respectively. For a vector $\balpha = (\alpha_1, \ldots, \alpha_p)^\T \in \mR^p$ and a subset $\calS \subseteq \{ 1, 2, \ldots, p \}$, let $\balpha_{\calS}$ denote the $|\calS|$-dimensional subvector of $\balpha$ corresponding to the indices in $\calS$; moreover, let $\| \balpha \|_q = (\sum_{j=1}^p | \alpha_j |^q)^{1/q}$, $q \ge 1$, $\| \balpha \|_{\infty} = \max_{1 \le j \le p} | \alpha_j |$, and $\| \balpha \|_0 = | \{ j: \alpha_j \ne 0, j = 1, \ldots, p \} |$. For a symmetric matrix $\bA = (A_{ij}) \in \mR^{p \times p}$, let $\lambda_{\min}(\bA)$ and $\lambda_{\max}(\bA)$ denote the smallest and largest eigenvalues, respectively, $\| \bA \| = \lambda_{\max}^{1/2}(\bA^\T \bA)$ the spectral norm, $\| \bA \|_{\infty} = \max_{i, j} |A_{ij}|$ the element-wise max norm, and $\| \bA \|_{op,1} = \max_j \sum_{i} |A_{ij}|$ the induced $\ell_1$ matrix norm. 
A random variable $Z$ is sub-Gaussian if its $\psi_2$-norm satisfies $\| Z \|_{\psi_2} = \sup_{k \ge 1} k^{-1/2} (\E |Z|^k)^{1/k} ={\O}(1)$, and is sub-exponential if its $\psi_1$-norm satisfies $\| Z \|_{\psi_1} = \sup_{k \ge 1} k^{-1} (\E |Z|^k)^{1/k} ={\O}(1)$.
We denote the first- and the second-order derivatives of a function $h$ with respect to its argument by $\dot{h}$ and $\ddot{h}$. 
For two non-negative sequences $\{a_n\}$ and $\{b_n\}$, we write $a_n \lesssim  b_n$ ($a_n \gtrsim b_n$) if there is a constant $C>0$ such that $a_n \le C b_n$ ($a_n \ge C b_n$), and $a_n \asymp b_n$ if  $a_n \lesssim b_n$ and $a_n \gtrsim b_n$.
	
For the $j$th observation of the $i$th subject,  we consider an outcome $Y_{ij}$ and a $p$-dimensional vector of covariates $\bx_{ij} \in \mR^p$, $i = 1, \ldots, n$ and $j = 1, \ldots, m_i$. Though the proposed method is applicable in cases where $m_i$'s differ, we assume, without loss of generality, that $m_i = m$, $i = 1, \ldots, n$. The number of observations per subject, $m$, is assumed to be relatively small and can be viewed as a fixed integer. Observations from different subjects are assumed independent and identically distributed, and within-subject outcomes are correlated. 
	
\subsection{Generalized estimating equations {and target of inference}}\label{sec:gee}

Suppose the expectation of $Y_{ij}$ given the observed covariates $\bx_{ij}$ satisfies
$
g\{ \E(Y_{ij} | \bx_{ij}) \} = \bx_{ij}^\T \bbeta,
$
for a canonical link function $g(\cdot)$, where $\bbeta = (\beta_1, \ldots, \beta_p)^\T$ is the $p$-dimensional vector of regression coefficients. Let $\bbeta^0$ be the true regression coefficient vector. The target of inference here is a general linear functional of $\bbeta^0$, i.e. $\theta^0 = \bxi^\T \bbeta^0$, where $\bxi \in \mR^p$ is a target loading vector;  without loss of generality, we assume $\| \bxi \|_2 = 1$. 
	
Let $\mu_{ij}(\bbeta) = \mu(\bx_{ij}^\T \bbeta) = \E (Y_{ij} | \bx_{ij})$ be the expectation of $Y_{ij}$ given the covariates $\bx_{ij}$ in the posited model, where $\mu(\cdot)$ is the mean function in linear predictors $\bx_{ij}^\T \bbeta$ for responses, and $\bmu_i(\bbeta) = (\mu_{i1}(\bbeta), \ldots, \mu_{i m}(\bbeta))^\T$. The variance of $Y_{ij}$ conditional on the covariates $\bx_{ij}$ is written as $\Var (Y_{ij} | \bx_{ij}) = \sigma^2 v (\mu_{ij}(\bbeta)) = \sigma^2 v (\mu(\bx_{ij}^\T \bbeta))$, where $v(\cdot)$ is the variance function. Without loss of generality, we assume that the scale parameter $\sigma^2 = 1$ is known.  For instance, under the commonly used logit link for a binary response $Y_{ij}$, $\mu_{ij}(\bbeta) = \exp(\bx_{ij}^\T \bbeta) / \{ \exp(\bx_{ij}^\T \bbeta) + 1 \}$, and $v(\mu_{ij}(\bbeta)) = \mu_{ij}(\bbeta) \{ 1 - \mu_{ij}(\bbeta)\}$.
	
We write $\bY_i = (Y_{i1}, \ldots, Y_{i m})^\T$ as the vector of $m$ responses from the $i$th subject, and $\bX_i = (\bx_{i1}, \ldots, \bx_{i m})^\T$ as the corresponding $m \times p$ matrix of covariates. For $\bX_i \in \mR^{m \times p}$, we also write $\bX_i = (\bx_{i[1]}, \ldots, \bx_{i[p]})$, where $\bx_{i[j]}$ is the $j$th column in $\bX_i$. 
Denoting the $(m \times p)$-dimensional Jacobian matrix $\bD_i = \partial \bmu_i / \partial \bbeta^\T$, the optimal estimating equation for $\bbeta$ would be 
\[
	\displaystyle \frac{1}{n} \sum_{i=1}^n \bD_i^\T \Var(\bY_i | \bX_i)^{-1} \left\{ \bY_i - \bmu_i(\bbeta) \right\} = 0,
\]
where $\Var (\bY_i | \bX_i)$ is the $m \times m$ variance-covariance matrix of $\bY_i$. 
However, the true covariance matrix $\Var (\bY_i | \bX_i)$ is often unknown, making the optimal estimating equation intractable. We therefore follow the generalized estimating equation framework and use a working correlation matrix $\bR(\bgamma)$ for $\bY_i$, where $\bgamma$ is a finite-dimensional nuisance parameter for the working correlation matrix. The working covariance matrix for $\bY_i$ is thus 
$	
\Sigma_i = \bG_i^{1/2} (\bbeta) \bR(\bgamma) \bG_i^{1/2} (\bbeta), 
$
where $\bG_i (\bbeta) = \diag \{ G_{i1}(\bbeta), \ldots, G_{im}(\bbeta) \} = \diag \{ \Var(Y_{i1} | \bx_{i1}), \ldots, \Var(Y_{i m} | \bx_{im}) \}$.
	
Let $\hat{\bgamma}$ be the estimated nuisance parameter in the working correlation matrix. Methods for obtaining $\hat\bgamma$ under various working correlation structures are described in \citet{liang1986longitudinal}. We write $\hatR = \bR(\hat{\bgamma})$ as the estimated working correlation matrix. The generalized estimating equation for $\bbeta$, given the working covariance matrix $\bG_i^{1/2} (\bbeta) \hatR \bG_i^{1/2} (\bbeta)$, is then 
\begin{equation}\label{eq:GEEorig}
	\bPsi(\bbeta) = \displaystyle n^{-1} \sum_{i=1}^n \bX_i^\T \bG_i^{1/2}  (\bbeta) \hatR^{-1} \bG_i^{-1/2}(\bbeta) \left\{ \bY_i - \bmu_i(\bbeta) \right\}.
\end{equation}
The sensitivity matrix associated with $\bPsi(\bbeta)$ is
$
	\bS(\bbeta) = n^{-1} \sum_{i=1}^n \bX_i^\T \bG_i^{1/2}(\bbeta) \hatR^{-1} \bG_i^{1/2}(\bbeta) \bX_i,
$
which is viewed as the expectation of the Jacobian matrix of $\bPsi(\bbeta)$ conditional on $\{\bX_i\}_{i=1}^n$ and some independently estimated $\hatR$. 

\subsection{Initial estimator}

With high-dimensional covariates, it may become infeasible to directly solve the estimating equation $\bPsi(\bbeta) = 0$ for $\bbeta$. Similar to \citet{van2012quasi} and \citet{fang2020test}, we obtain an initial estimator for $\bbeta$ by minimizing an $\ell_1$ penalized quasi log-likelihood, i.e.,
\begin{equation} \label{eq:init_beta}
	\hbeta = \argmin_{\bbeta \in \mR^p} \{ \ell_n(\bbeta) + \lambda \| \bbeta \|_1 \},
\end{equation}
where
$
	\ell_n(\bbeta) = - n^{-1} \sum_{i=1}^n \sum_{j=1}^m \int_{Y_{ij}}^{\mu_{ij}(\bbeta)} (Y_{ij} - u) / v(u) du
$
is the negative quasi log-likelihood under the working independence structure. For example, with binary outcomes and the logit link $g(u) = \log \{ u/(1-u) \}$, $\ell_n(\bbeta)$ is identical to the negative log-likelihood function in logistic regression as if all outcomes $Y_{ij}$'s were independent.

\subsection{One-step updated estimator via projected estimating equations \label{subsec:one_step_estimator}}
	
Similar to \citet{zhang2011statistical}, we consider the class of one-dimensional least favourable sub-models $\{ \bbeta + \bomega \phi: |\phi| < \epsilon \}$ for $\bPsi(\bbeta)$ in \eqref{eq:GEEorig} with some sufficiently small $\epsilon > 0$, which leads to inference via projected estimating equations. In the sub-models $\{ \bbeta + \bomega \phi: |\phi| < \epsilon \}$, the direction $\bomega$ is obtained as follows. First, we solve the constrained $\ell_1$ minimization problem
\begin{equation} \label{eq:direction}
	\tilde{\bomega} = \argmin_{\bomega \in \mR^p} \left\{ \| \bomega \|_1: ~ \| \bS(\hbeta) \bomega - \bxi \|_{\infty} \le \lambda^{\prime} \right\},
\end{equation}
with a tuning parameter $\lambda^{\prime}>0$. The projection direction $\homega$ is then obtained by normalizing $\tilde{\bomega}$, i.e., $\hat{\bomega} = \tilde{\bomega} / \{ \tilde{\bomega}^\T \bS(\hbeta) \tilde{\bomega} \}$. 
	
Next, we define the projected estimating function corresponding to \eqref{eq:GEEorig} as
\[
    \Psi^P(\theta) = \hat{\bomega}^\T \bPsi \left( \hbeta + \hat{\bomega} \left(\theta - \bxi^\T \hbeta \right) \right), ~ \theta \in \mR. 
\]
One natural strategy is to solve the projected estimating equation $\Psi^P(\theta) = 0$ to obtain an estimator of $\theta^0$. However, this would theoretically require additional assumptions to guarantee the existence and uniqueness of the roots of $\Psi^P(\theta) = 0$. Meanwhile, finding the roots of $\Psi^P(\theta) = 0$ may sometimes be computationally unstable and generate solutions far away from the truth; for instance, with binary outcomes and the logit link. Hence, instead of Z-estimation, we consider the first-order Taylor expansion of $\bPsi^P (\theta)$ around $\hat{\theta}$ and resort to a one-step updated estimator that is the approximate solution to $\Psi^P(\hat{\theta}) + \dot{\Psi}^P(\hat{\theta}) (\theta - \hat{\theta}) = 0$, where $\hat{\theta} = \bxi^\T \hbeta$ and $\dot{\Psi}^P(\theta) = d \Psi^P(\theta) / d \theta$. 
The exact form of $\dot{\Psi}^P(\hat{\theta})$ is complicated and is given in the proof in the Supplementary Material \ref{supp-supp:sec:proof}.
Note that $\dot{\Psi}^P(\hat{\theta}) = - \homega^\T \bS(\hbeta) \homega + \oP(1)$; therefore, the proposed one-step updated estimator is defined as 
\begin{equation} \label{eq:def_onestep}
	\tilde{\theta} = \hat{\theta} + \left\{ \homega^\T \bS(\hbeta) \homega \right\}^{-1} \Psi^P(\hat{\theta}).
\end{equation}
In \eqref{eq:def_onestep},  $\{ \homega^\T \bS(\hbeta) \homega \}^{-1} \Psi^P(\hat{\theta})$ can be viewed as a bias correction term for the initial estimator $\hat{\theta}$.

	
\subsection{Tuning parameter selection and implementation}
\label{sec:prac}
	

Choosing appropriate tuning parameters is crucial for reliable estimation and inference using regularized estimation procedures. In our case, the tuning parameters for the working independence model in \eqref{eq:init_beta}, $\lambda$, can be readily chosen via $K$-fold cross-validation. However, selecting the second tuning parameter, $\lambda^{\prime}$, is more nuanced and has remained a challenging problem in high-dimensional inference involving estimating equations. 
For instance, \citet{neykov2018unified} did not provide  a data-driven procedure for this tuning parameter. To overcome this limitation, we propose a practical $K$-fold cross-validation procedure for $\lambda^{\prime}$.


\begin{algorithm} 
	\caption{$K$-fold cross-validation for data-driven tuning parameter selection on $\lambda^{\prime}$}	
        \label{algo:cv_lambda_prime}
        \begin{algorithmic}
			\State \textbf{Input:} A pre-specified grid of $L$ trial points, $\lambda^{\prime}_1, \ldots, \lambda^{\prime}_L$ 
			\State \textbf{Input:} $\{ \bY_i, \bX_i \}_{i=1}^n$ randomly split into $K$ equally sized folds 
			\State For  $l = 1, \ldots, L$ 
			\State \qquad Set  $\textsc{cv}_l=0$  
			\State \qquad For $k = 1, \ldots, K$ 
			\State \qquad\qquad Use the training data $\calD^{(k,\mathrm{train})}$ to estimate $\hbeta^{(k)}$,  $\homega^{(k)}$, and $\tilde{\theta}^{(k)}$ 
			\State \qquad\qquad Set $\tilde{\bbeta}^{(k)} = \hbeta^{(k)} + \homega^{(k)} (\tilde{\theta}^{(k)} - \bxi^\T \hbeta^{(k)} )$ 
			\State \qquad\qquad Let $\textsc{cv}_l \leftarrow \textsc{cv}_l + [ \{ \homega^{(k)} \}^\T \bPsi^{(k,\mathrm{test})}(\tilde{\bbeta}^{(k)}) ]^2$
			\State \textbf{Output:} The smallest $\lambda^{\prime}_{l}$ which gives  $\textsc{cv}_l$  within three standard errors above $\min_l \textsc{cv}_l$
        \end{algorithmic}
\end{algorithm}
	
As shown in Algorithm~\ref{algo:cv_lambda_prime}, the $k$th training set $\calD^{(k,\mathrm{train})}$ is utilized to compute the penalized estimator $\hbeta^{(k)}$ and the projection direction $\homega^{(k)}$, followed by the one-step updated estimator $\tilde{\theta}^{(k)}$. Unlike existing cross-validation procedures suited for loss functions, such as negative log (pseudo) likelihood, our proposed criterion relies on the squared projected estimating equation 
$ [ \{ \homega^{(k)} \}^\T \bPsi^{(k,\mathrm{test})}(\tilde{\bbeta}^{(k)}) ]^2$, 
where the observed data in $\bPsi^{(k,\mathrm{test})}$ are taken from the test set $\calD^{(k,\mathrm{test})}$. 
In particular, we choose the smallest tuning parameter $\lambda^{\prime}_l$ that results in a cross-validated criterion value which is within three standard errors above the minimum criterion value. 
In the Supplementary Material \ref{supp-supp:sec:add_sim_res}, we show that this choice performs better than an alternative based on the smallest cross-validated value. 
	
The projection direction $\tilde{\bomega}$ in \eqref{eq:direction} can be estimated by formulating the original problem into a linear programming problem, which we solve with the open source software GNU Linear Programming Kit \texttt{GLPK} and its R interface \texttt{Rglpk}.

	
\section{Theoretical Justification} 
\label{sec:theory}
	
In this section, we show that the proposed one-step updated estimator $\tilde{\btheta}$ is asymptotically normally distributed, which lays the foundation for inference using $\tilde{\btheta}$. Similar to \citet{wang2012penalized}, we suppose the estimated working correlation matrix $\hatR$ converges to some non-random but potentially unknown matrix $\barR \in \mR^{m \times m}$, recognizing that $\barR$ may be different from the true correlation matrix of $\bY_i$, $\bR_0$. Replacing $\hatR$ with $\barR$ in $\bPsi(\bbeta)$, we introduce a surrogate generalized estimating equation for theoretical purposes:
\[
	\barPsi(\bbeta) = \displaystyle \frac{1}{n} \sum_{i=1}^n  \bX_i^\T \bG_i^{1/2} (\bbeta) \barR^{-1} \bG_i^{-1/2}(\bbeta) \left\{ Y_i - \bmu_i(\bbeta) \right\}.
\]
The sensitivity and variability matrices corresponding to $\barPsi(\bbeta)$ are then given by
\[
	\begin{array}{rcl}
		\barS(\bbeta) & = & \displaystyle \frac{1}{n} \sum_{i=1}^n \bX_i^\T \bG_i^{1/2}(\bbeta) \barR^{-1} \bG_i^{1/2}(\bbeta) \bX_i, \\
		\bar{\bV}(\bbeta) & = & \displaystyle \frac{1}{n} \sum_{i=1}^n \bX_i^\T \bG_i^{1/2}(\bbeta) \barR^{-1} \bR_0 \barR^{-1} \bG_i^{1/2}(\bbeta) \bX_i, 
	\end{array}
\]
respectively. Let $\bS^0 = \E \bar{\bS}(\bbeta^0)$ and $\bV^0 = \E \bar{\bV}(\bbeta^0)$ be the corresponding population-level matrices.
For the target of interest $\theta^0 = \bxi^\T \bbeta^0$, we define the direction vector $\baromega^0 = (\bS^0)^{-1} \bxi$, and the rescaled direction $\bomega^0 = \baromega^0 / \{ (\baromega^0)^\T \bS^0 \baromega^0 \}$, which are the population-level counterparts for $\tilde{\bomega}$ and $\hat{\bomega}$, respectively. Let $\calS^{\prime} = \{ 1 \le j \le p: \omega^0_j \ne 0 \} = \{ 1 \le j \le p: \bar{\omega}^0_j \ne 0 \}$, and denote by $s^{\prime}  = | \calS^{\prime} |$ the sparsity of $\bomega^0$ and $\baromega^0$. The model sparsity is denoted $s_0 = | \{  1 \le j \le p: \beta_j^0 \ne 0 \} |$.

We next investigate the asymptotic properties of the proposed estimator by first establishing the convergence of $\hbeta$ and the direction $\homega$. To this end, we need to establish the restricted eigenvalue condition, which is shown in Supplementary Material \ref{supp-supp:sec:cond}. The results stated below are obtained under relatively mild regulatory assumptions discussed and justified in Appendix~\ref{appendix:assumptions}.
	
\begin{lemma} \label{lemma:lasso}
Under Assumptions \ref{assump:bound_covs}, \ref{assump:var_fun}, \ref{assump:bdd_eigen_covs} and \ref{assump:cor_mat}, if the tuning parameter $\lambda$ in \eqref{eq:init_beta} satisfies $\lambda \asymp \{ \log(p)/n \}^{1/2}$, then
\[
	\| \hbeta - \bbeta^0 \|_q  = \OP ( s_0^{1/q} \lambda ),  ~ q = 1, 2 ~; \quad \quad
	\displaystyle \frac{1}{n} \sum_{i=1}^n \sum_{j=1}^m \left\{ \bx_{ij}^\T ( \hbeta - \bbeta^0 ) \right\}^2   = \OP (s_0 \lambda^2).
\]
\end{lemma}

\begin{lemma} \label{lemma:hat_omega_rate}
Suppose Assumptions \ref{assump:bound_covs}--\ref{assump:cor_mat} hold and $s^{\prime} \lambda^{\prime} \asymp s^{\prime} \| \baromega^0 \|_1 (s_0 \lambda + r_n) = o(1)$.  Then 
$
	\| \tilde{\bomega} - \baromega^0 \|_1 = \OP(s^{\prime} \lambda^{\prime}).
$
Since $\{ \tilde{\bomega}^\T \bS(\hbeta) \tilde{\bomega} \}^{-1} = \OP(1)$, it is valid to define $\hat{\bomega} = \tilde{\bomega} / \{ \tilde{\bomega}^\T \bS(\hbeta) \tilde{\bomega} \}$ with probability going to one, and
$
	\| \hat{\bomega} - \bomega^0 \|_1 = \OP(s^{\prime} \lambda^{\prime} \| \baromega^0 \|_1 ).
$
\end{lemma}
	
	
Our main result establishes the convergence in distribution of $\tilde{\theta}$ in \eqref{eq:def_onestep} to $\theta^0 = \xi^\T \beta$. 

\begin{theorem} \label{thm:OS_est_xi}
Suppose the tuning parameters satisfy $\lambda \asymp \{ \log(p)/n \}^{1/2}$ and $\lambda^{\prime} \asymp \| \baromega^0 \|_1 (s_0 \lambda + r_n)$, and that 
$\max(s_0, s^{\prime} \| \baromega^0 \|_1) \lambda^{\prime} \{ \log(p) \}^{1/2} = o(1)$. 
Then, under Assumptions \ref{assump:bound_covs}--\ref{assump:cor_mat}, we have that
$$
	{n}^{1/2}  \left. \left( \tilde{\theta} - \theta^0 \right)   \middle/  \left[ \left\{ (\bomega^0)^\T \bV^0 \bomega^0 \right\}^{1/2} \middle/ \left\{ (\bomega^0)^\T \bS^0 \bomega^0 \right\} \right] \right.
$$
converges in distribution to the standard normal distribution.
\end{theorem}
	
In practice, we use $\homega^\T \bV(\hbeta) \homega / \{ {n}^{1/2} \homega^\T \bS(\hbeta) \homega \}^2$ to approximate the variance of $\tilde{\theta}$. Then, for $0< \alpha < 1$, the $(1-\alpha)\times 100\%$ confidence interval for $\theta^0$ can be constructed as
\[
	\left[ \tilde{\theta} - z_{\alpha/2} \left\{ \homega^\T \bV(\hbeta) \homega \right\}^{1/2} \middle/ \left\{ {n}^{1/2} \homega^\T \bS(\hbeta) \homega \right\}, \tilde{\theta} + z_{\alpha/2} \left\{ \homega^\T \bV(\hbeta) \homega \right\}^{1/2} \middle/ \left\{ {n}^{1/2} \homega^\T \bS(\hbeta) \homega \right\}  \right],
\]
where $z_{\alpha/2}$ is the upper $(\alpha/2)$th quantile of the standard normal distribution.
	
\begin{remark}
Theorem~\ref{thm:OS_est_xi} is proved in the Supplementary Material \ref{supp-supp:sec:proof}. For technical treatment, we assume that the initial estimator $\hbeta$ and the working correlation matrix $\hat{\bR}$ are obtained independently from the proposed de-biased estimator $\tilde{\theta}$, using, e.g., data splitting \citep{ma2021global}; that is, suppose there are $2n$ observations, where the samples $\{\bY_i, \bX_i\}_{i=1}^n$ are used to estimate $\hbeta$ and $\hat{\bR}$ and the samples $\{\bY_i, \bX_i\}_{i=n+1}^{2n}$ to construct the projected estimating function and conduct inference on $\theta^0$ as discussed above. This data splitting is solely needed for technical analysis, and our numerical results in Section~\ref{sec:sim} demonstrate that the proposed method performs well without data splitting.
\end{remark}

	
\section{{Longitudinal proteomic profiling for COVID-19 severity}}
\label{sec:app}

Proteomics has proven a useful tool to reveal molecular level signatures associated with COVID-19 infection, severity, mortality and post-acute sequelae of COVID-19 \citep{filbin2021longitudinal,haljasmagi2020longitudinal,shen2020proteomic,yin2023long,feyaerts2022integrated}. Complementary to the conventional univariate approaches that examine one protein at a time without taking into account the correlations among proteins, we apply HDIGEE to the longitudinal proteomic study of COVID-19 patient plasma samples in \citet{filbin2021longitudinal} to uncover protein signatures associated with COVID-19 severity in a high-dimensional regression setting. An independent cross-sectional study by \citet{feyaerts2022integrated} measured protein abundance using the same proximity extension assay. 
The matching proteomic platforms facilitates the development and independent evaluation of risk predictors: we use the longitudinal  binary responses and covariates from \citet{filbin2021longitudinal} to obtain de-biased estimates and 95\% confidence intervals for linear functionals of the form $\bbeta^{\T} x$. 
Baseline demographics and protein measures in \citet{feyaerts2022integrated} are then input as new covariates to obtain risk scores of the form $\bbeta^{\T} x_{\mathrm{new}}$ and corresponding confidence intervals for these 64 new subjects.

The first dataset \citep{filbin2021longitudinal} includes 659 measurements from $n=305$ COVID-19 positive patients on Day 0 after enrollment, and on Day 3 and Day 7 for those who were still hospitalized. 
Severity is defined as a binary outcome based on the World Health Organization ordinal outcomes scale. Following \citet{filbin2021longitudinal}, we classify patients with outcomes A1 (died) and A2 (intubated, survived) as severe cases ($Y_{ij}=1$), and those with outcomes A3 (hospitalized on oxygen), A4 (hospitalized
without oxygen) and A5 (discharged from Emergency Department) as non-severe cases ($Y_{ij}=0$). There are 80 severe and 225 non-severe cases on Day 0, 80 severe and 135 non-severe cases on Day 3, and 72 severe and 67 non-severe cases on Day 7. 
The second dataset \citep{feyaerts2022integrated}, which is cross-sectional, consists of 64 COVID-19 positive patients (24 severe, and 40 non-severe, including 13 moderate and 27 mild).
Both studies contain $1,420$ plasma proteins with UniProt identifiers. Since highly correlated proteins are likely redundant and can cause difficulties in high-dimensional estimation \citep{buhlmann2013correlated}, we include $p=853$ proteins for subsequent analysis after applying a filter using R function ``findCorrelation" with cutoff 0.75. Protein levels are represented as normalized protein expression value (NPX) in log2 scale, and standardized to z-scores. Our analysis is adjusted for age and obesity.

We use the default 10-fold cross-validation for the initial lasso estimator with logit link  and 5-fold cross-validation for the proposed projection direction estimation as described in Section \ref{sec:prac}. The working correlation structure for within-patient outcomes is AR(1). 
Under working independence, Lasso selects 131 proteins, whereas 14 proteins are deemed significant by HDIGEE with q-value $< 0.05$. Figure~\ref{fig:app_new}a shows the estimated regression coefficients after de-biasing and the negative log$_{10}$ q-values of all proteins;  significant proteins are marked with gene names. 
A full list of gene names, Olink IDs, UniProt IDs, estimated coefficients by lasso and HDIGEE, and q-values for the 14 significant proteins can be found in the Supplementary Material \ref{supp-supp:sec:add_covid_res}.


\begin{figure}[t!] 
\centering
\includegraphics[width=0.9\textwidth]{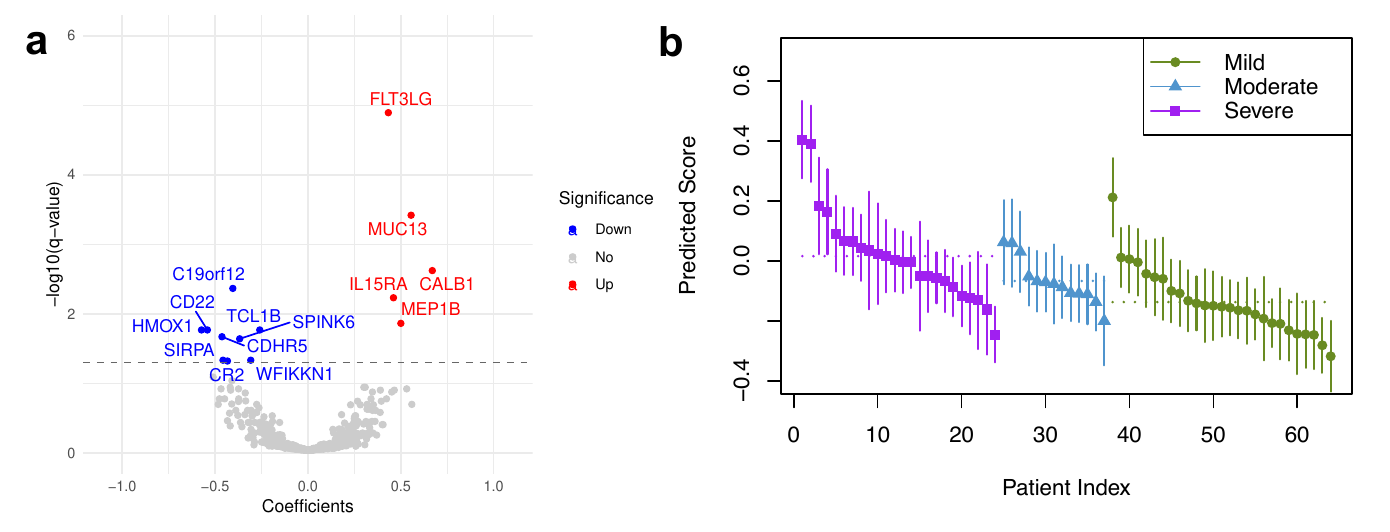}
\caption{\small (a) Estimated regression coefficients (x-axis) and negative log10 of q-values (y-axis) for all proteins by the proposed inference method. Fourteen significant proteins (q-value less than 0.05) are labeled with corresponding gene names, in red for positive associations with severity and in blue for negative associations. (b) Estimated linear predictors $\bbeta^\T x_{\mathrm{new}}$ via de-biasing with 95\% confidence intervals for each patient in the second dataset. Mean scores in each group are highlighted in dotted line segments.}
\label{fig:app_new} 
\end{figure}

By taking into account the correlations between proteins in a multivariable fashion instead of exploring marginal associations, HDIGEE identifies biologically meaningful proteins in Olink Inflammation, Neurology, Oncology and Cardiometabolic panels. For instance, Fms-related tyrosine kinase 3 ligand (FLT3LG), the most significant protein and up-regulated among the severe cases, activates hematopoietic progenitors and functions as a cytokine and growth factor that increases immune cells, and it was also previously found to be of increased levels in severe long covid patients \citep{espin2023cellular}. The significant association between COVID-19 severity and {IL15RA} (Interleukin-15 receptor subunit alpha) indicates the importance of IL15, a critical immunoregulatory cytokine that aids in T cell responses, activates natural killer cells and may be the target of novel immunotherapy \citep{kandikattu202015}. Signal-regulatory protein alpha (SIRPA) is a well-known inhibitor for phagocytosis, which interacts with transmembrane protein CD47 expressed on target cells, and previous studies suggest that increasing SIRPA's activity could be an anti-viral therapeutic target \citep{sarute2021signal}.

Next, we apply HDIGEE to the first dataset \citep{filbin2021longitudinal} again, but now with $\bxi = x_{\mathrm{new}}$ for each of the 64 patients in the second dataset \citep{feyaerts2022integrated}. Here, $x_{\mathrm{new}}$ contains information on age, obesity condition and abundance of 853 proteins, scaled by $\sqrt{p}$; we conduct inference on prognostic scores based on the linear predictors $x_{\mathrm{new}}^{\T} \bbeta$. Figure \ref{fig:app_new}b shows the estimated linear predictors via de-biasing and their 95\% confidence internals for the 64 patients with mild, moderate or severe conditions. While severe and non-severe patients cannot be perfectly separated (Figure~\ref{fig:intro}a), we observe a clear increasing trend of the predictive scores based on the estimated linear predictors from mild to severe infection. 
More importantly, the corresponding confidence intervals provide data-based evidence for clinical decisions in resource-constrained settings, or when balancing  between treatment benefits and adverse side effects.
Hence, HDIGEE is demonstrated as an effective unified method for inference from individual coefficients to predictive scores using all covariate information.

	
\section{Empirical evaluations}
\label{sec:sim}
	
\subsection{Simulation settings}
	
In this section, we evaluate the empirical performance of the proposed method via Monte Carlo simulations. The first scenario considers a fixed number of observations $m$ for all subjects with commonly used identity and logit link functions. For both models, we simulate data for $n=100$ subjects, each having $m=5$ correlated observations. The true correlation matrix $\bR_0$ is either set based on an autoregressive process of order 1, AR(1), with correlation 0.3, or unstructured:
\[  
    \footnotesize
	\bR^{\mathrm{(AR1)}}_0 = \left( \begin{array}{ccccc}
		1 & 0.3 & 0.3^2 & 0.3^3 & 0.3^4  \\
		0.3 & 1 & 0.3 & 0.3^2 & 0.3^3  \\
		0.3^2 & 0.3 & 1 & 0.3 & 0.3^2 \\
		0.3^3 & 0.3^2 & 0.3 & 1 &  0.3  \\
		0.3^4 & 0.3^3 & 0.3^2 & 0.3 & 1    \\
	\end{array} \right) 
	\mathrm{~ or ~}
	\bR^{\mathrm{(UN)}}_0 = \left( \begin{array}{ccccc}
		1 & 0.4 & 0.3 & 0.2 & 0.1  \\
		0.4 & 1 & 0.4 & 0.3 & 0.2  \\
		0.3 & 0.4 & 1 & 0.4 & 0.3 \\
		0.2 & 0.3 & 0.4 & 1 &  0.4  \\
		0.1 & 0.2 & 0.3  & 0.4 & 1    \\
	\end{array} \right).
\]
The covariates in $\bX_i$'s are simulated from a multivariate normal distribution with mean zero and an AR(1) covariance matrix $\bSigma_x$, where the $(k,l)$th element of $\bSigma_x$ equals $0.5^{|k-l|}$. We set $p=100$ or $500$, and randomly select $s_0 = 3$ or $10$ covariates to  taking values of 1 and 0.5 for continuous and binary outcomes, respectively. These coefficients are fixed throughout.  For the continuous case with $s_0=3$, $p=100$ or $500$, we  also evaluate the performance of HDIGEE when inferring different linear functionals of regression coefficients (Table \ref{tab:linear_combn_setup} in the Appendix), which has been rarely studied before for correlated data or estimating equations in general. 

In the second scenario, we mimic the scale and complexity of the longitudinal COVID-19 proteomic data of \citet{filbin2021longitudinal}, from Section~\ref{sec:app} to assess the performance of our model in that setting. 
More specifically, we use the $p = 853$ protein measures from $n=305$ subjects with 1--3 timepoints per subject. 
Binary outcomes mimicking the severity of illness are simulated using either $s_0 = 3$ or $10$ non-zero regression coefficients generated from Uniform(0.3, 0.5). The coefficients are fixed throughout the simulations. For each replication, the correlated outcomes for each patient are simulated based on the upper left corner submatrix as $R^0$ below and with corresponding dimensions to be the true correlation matrix
 \begin{equation*}
 \footnotesize
    R_0 = \left(  \begin{array}{ccc} 1.00 & 0.45 & 0.30 \\
    0.45 & 1.00 & 0.45 \\
    0.30 & 0.45 & 1.00 
        \end{array}
    \right).
 \end{equation*}

In the third scenario, we conduct additional simulations with continuous outcomes that mimic the scale and complexity of the riboflavin production data analyzed by \citet{buhlmann2014high}, which involves varying number of observations, $m_i$, from 2 to 6 per subject. The data contains $N = 111$ observations from $n = 28$ strains of \textit{B. subtilis} (a type of bacteria for producing riboflavin), with $p=267$ genes after filtering. Additional details on this simulation are given in Supplementary Material \ref{supp-supp:sec:add_sim_res}. 

The empirical performance of the proposed method (HDIGEE) is compared against a modified version of the high-dimensional inference procedure for estimating equations by \citet{neykov2018unified} (HDEE; see the Supplementary Material \ref{supp-supp:sec:hdee} for details), which was originally designed for independent observations, and the quadratic decorrelated inference function  \citep[QDIF,][]{fang2020test}. In preliminary studies, we found that solving the projected estimating equation, as used in \citet{neykov2018unified}, might cause additional numerical instability for binary outcomes. We thus adopt the one-step updated estimator for HDEE, similar to Section \ref{subsec:one_step_estimator}.
 The initial estimator $\hat\bbeta$ is obtained using the R package ``glmnet" with 10-fold cross-validation. For {HDIGEE and HDEE},  $\lambda^{\prime}$ is selected via 5-fold cross-validation following Algorithm \ref{algo:cv_lambda_prime}. For QDIF, the basis matrices are chosen according to \citet{fang2020test}. 

	\subsection{Simulation results}


For the first scenario, we summarize results of 200 simulations on the empirical bias (Bias), coverage probability of 95\% confidence interval (Cov),  model-based standard error (SE) and empirical standard error (EmpSE), averaged for all  nonzero coefficients and three arbitrarily chosen zero coefficients (see Table~\ref{tab:sim_logit_n100m5_3se}). 
For HDIGEE and HDEE, cross-validation with the proposed criterion  in Algorithm~\ref{algo:cv_lambda_prime} is preferred, as it outperforms the selection of $\lambda^{\prime}$ as the value that gives the smallest cross-validated quantity (see the Supplementary Material \ref{supp-supp:sec:add_sim_res}). 
Based on Table~\ref{tab:sim_logit_n100m5_3se}, HDIGEE and QDIF show better performance than HDEE, and both well maintain the type 1 error rates for the examined noise variables. Although QDIF exhibits slightly smaller variability, HDIGEE presents smaller estimation bias and better confidence interval coverage probability for true signals compared to QDIF. HDIGEE is particularly more promising for inference on true signals when $s_0=10$, and hence, potentially more desirable as the true model sparsity is unknown.

	
	\begin{table}[ht!]
   \small
		\caption{\small Binary outcome setting in the first scenario with $(n, m) = (100, 5)$: estimation bias (Bias), coverage of the 95\% confidence interval (Cov), model-based standard error (SE) and empirical standard error (EmpSE) for $s_0$ signal and three randomly chosen noise variables, averaged over 200 replications.}
			\begin{tabular}{llrrrrrrrr} 
   \hline
				&  & \multicolumn{4}{c}{Signal $(s_0)$} & \multicolumn{4}{c}{Noise $(3)$} \\
				Method & Item & \multicolumn{2}{c}{$p=100$} & \multicolumn{2}{c}{$p=500$} & \multicolumn{2}{c}{$p=100$} & \multicolumn{2}{c}{$p=500$} \\
				&  & $s_0=3$ & $s_0=10$ & $s_0=3$ & $s_0=10$ & $s_0=3$ & $s_0=10$ & $s_0=3$ & $s_0=10$ \\[5pt]
                \hline
				\multicolumn{10}{c}{Autoregressive} \\
				\multirow{4}{*}{HDIGEE}  &   Bias & -0.015 & -0.021 & -0.030 & -0.073 & -0.001 & 0.004 & -0.001 & 0.007 \\ 
  & Cov & 0.895 & 0.881 & 0.897 & 0.828 & 0.947 & 0.927 & 0.950 & 0.948 \\ 
  & SE & 0.116 & 0.121 & 0.115 & 0.112 & 0.116 & 0.120 & 0.115 & 0.111 \\ 
  & EmpSE & 0.133 & 0.150 & 0.135 & 0.138 & 0.119 & 0.129 & 0.117 & 0.111 \\[5pt] 
  \multirow{4}{*}{HDEE} & Bias & -0.032 & -0.061 & -0.023 & -0.064 & -0.004 & 0.017 & -0.002 & 0.007 \\ 
  & Cov & 0.847 & 0.792 & 0.892 & 0.831 & 0.895 & 0.895 & 0.933 & 0.938 \\ 
  & SE & 0.085 & 0.090 & 0.111 & 0.108 & 0.086 & 0.091 & 0.111 & 0.108 \\ 
  & EmpSE & 0.107 & 0.124 & 0.134 & 0.138 & 0.104 & 0.109 & 0.119 & 0.114 \\[5pt] 
  \multirow{4}{*}{QDIF} & Bias & -0.021 & -0.059 & -0.034 & -0.100 & -0.001 & 0.011 & 0.000 & 0.007 \\ 
  & Cov & 0.895 & 0.849 & 0.882 & 0.766 & 0.940 & 0.932 & 0.945 & 0.948 \\ 
  & SE & 0.104 & 0.108 & 0.100 & 0.097 & 0.103 & 0.107 & 0.098 & 0.097 \\ 
  & EmpSE & 0.117 & 0.131 & 0.119 & 0.115 & 0.107 & 0.112 & 0.098 & 0.093 \\[5pt] 
                \hline
				\multicolumn{10}{c}{Unstructured} \\
				\multirow{4}{*}{HDIGEE}   & Bias & -0.014 & -0.020 & -0.031 & -0.071 & -0.005 & 0.006 & 0.001 & 0.001 \\ 
  & Cov & 0.895 & 0.875 & 0.902 & 0.816 & 0.943 & 0.935 & 0.952 & 0.948 \\ 
  & SE & 0.115 & 0.120 & 0.114 & 0.111 & 0.115 & 0.119 & 0.113 & 0.110 \\ 
  & EmpSE & 0.131 & 0.150 & 0.129 & 0.141 & 0.118 & 0.131 & 0.111 & 0.111 \\[5pt] 
  \multirow{4}{*}{HDEE} & Bias & -0.028 & -0.060 & -0.019 & -0.062 & -0.006 & 0.017 & 0.001 & 0.001 \\ 
  & Cov & 0.845 & 0.786 & 0.898 & 0.818 & 0.878 & 0.878 & 0.933 & 0.933 \\ 
  & SE & 0.085 & 0.090 & 0.111 & 0.108 & 0.085 & 0.090 & 0.110 & 0.107 \\ 
  & EmpSE & 0.109 & 0.124 & 0.127 & 0.141 & 0.104 & 0.108 & 0.116 & 0.116 \\[5pt] 
  \multirow{4}{*}{QDIF} & Bias & -0.020 & -0.059 & -0.034 & -0.098 & -0.005 & 0.012 & 0.001 & 0.004 \\ 
  & Cov & 0.888 & 0.857 & 0.882 & 0.767 & 0.928 & 0.933 & 0.963 & 0.965 \\ 
  & SE & 0.102 & 0.106 & 0.099 & 0.097 & 0.101 & 0.105 & 0.097 & 0.096 \\ 
  & EmpSE & 0.117 & 0.129 & 0.113 & 0.116 & 0.108 & 0.111 & 0.095 & 0.092 \\ 
    \hline
		\end{tabular}
		\label{tab:sim_logit_n100m5_3se}
	\end{table}

\begin{table}[ht]
\small
\centering
\caption{\small Continuous outcome setting in the first scenario with $(n, m) = (100, 5)$ and $s_0 = 3$: estimation bias (Bias), coverage probability (Cov), model-based standard error (SE), and empirical standard error (EmpSE) for inference on different  linear functionals described in Table~\ref{tab:linear_combn_setup}.}
\begin{tabular}{l|rrrr|rrrr}
  \hline
  & \multicolumn{4}{c}{$p=100$} &  \multicolumn{4}{c}{$p=500$} \\
  \hline
No. & 1 & 2 & 3 & 4 & 9 & 10 & 11 & 12 \\ 
  \hline
Bias & 0.002 & -0.003 & 0.002 & -0.001 & -0.003 & -0.002 & 0.004 & 0.003 \\ 
  Cov & 0.955 & 0.940 & 0.935 & 0.920 & 0.940 & 0.975 & 0.940 & 0.945 \\ 
  SE & 0.040 & 0.036 & 0.050 & 0.052 & 0.039 & 0.031 & 0.044 & 0.048 \\ 
  EmpSE & 0.039 & 0.037 & 0.051 & 0.057 & 0.039 & 0.029 & 0.047 & 0.051 \\ 
  \hline
  No. & 5 & 6 & 7 & 8 & 13 & 14 & 15 & 16 \\ 
  \hline
  Bias & 0.008 & 0.005 & -0.007 & -0.005 & 0.004 & 0.003 & -0.004 & -0.003 \\ 
  Cov & 0.960 & 0.975 & 0.925 & 0.935 & 0.980 & 0.940 & 0.925 & 0.985 \\ 
  SE & 0.039 & 0.049 & 0.051 & 0.050 & 0.042 & 0.057 & 0.049 & 0.052 \\ 
  EmpSE & 0.038 & 0.043 & 0.049 & 0.051 & 0.038 & 0.058 & 0.050 & 0.048 \\ 
    \hline
\end{tabular}
\label{tab:res_linear_combn_3se}
\end{table}

Detailed results for the continuous outcome setting, presented in Supplementary Material~\ref{supp-supp:sec:add_sim_res}, show  similar phenomenon but smaller differences between HDIGEE and QDIF. Additionally, HDIGEE shows satisfactory performance for inference on various linear functionals of regression coefficients (Table \ref{tab:res_linear_combn_3se}). Even for cases where the true target $\theta^0 \ne 0$ involves both signals and noise variables (Combinations 3, 7, 11, 15) and where the true target $\theta^0 = 0$ involves noise variables and cancellation of signals (Combinations 4, 8, 12, 16), HDIGEE presents negligible bias and can achieve coverage probability around the nominal level.

The simulation results for the second scenario mimicking the longitudinal COVID-19 proteomic data with $s_0=3$ true signals are displayed in Figure~\ref{fig:res_sim_covidprot_s3_3se} and also include mean squared error (MSE). Again, for the three randomly chosen noise variables, both HDIGEE and QDIF perform well in estimating the coefficients and maintaining the coverage. For the more difficult task of inference on the true signals, although HDIGEE produces slightly larger biases than QDIF, its coverage probabilities are slightly better than QDIF, with smaller variability in estimating these coefficients and reduced MSE. The simulation results with $s_0=10$ true signals are shown in the Supplementary Material \ref{supp-supp:sec:add_sim_res}.

\begin{figure}[ht]
    \centering
    \includegraphics[width=0.8\textwidth]{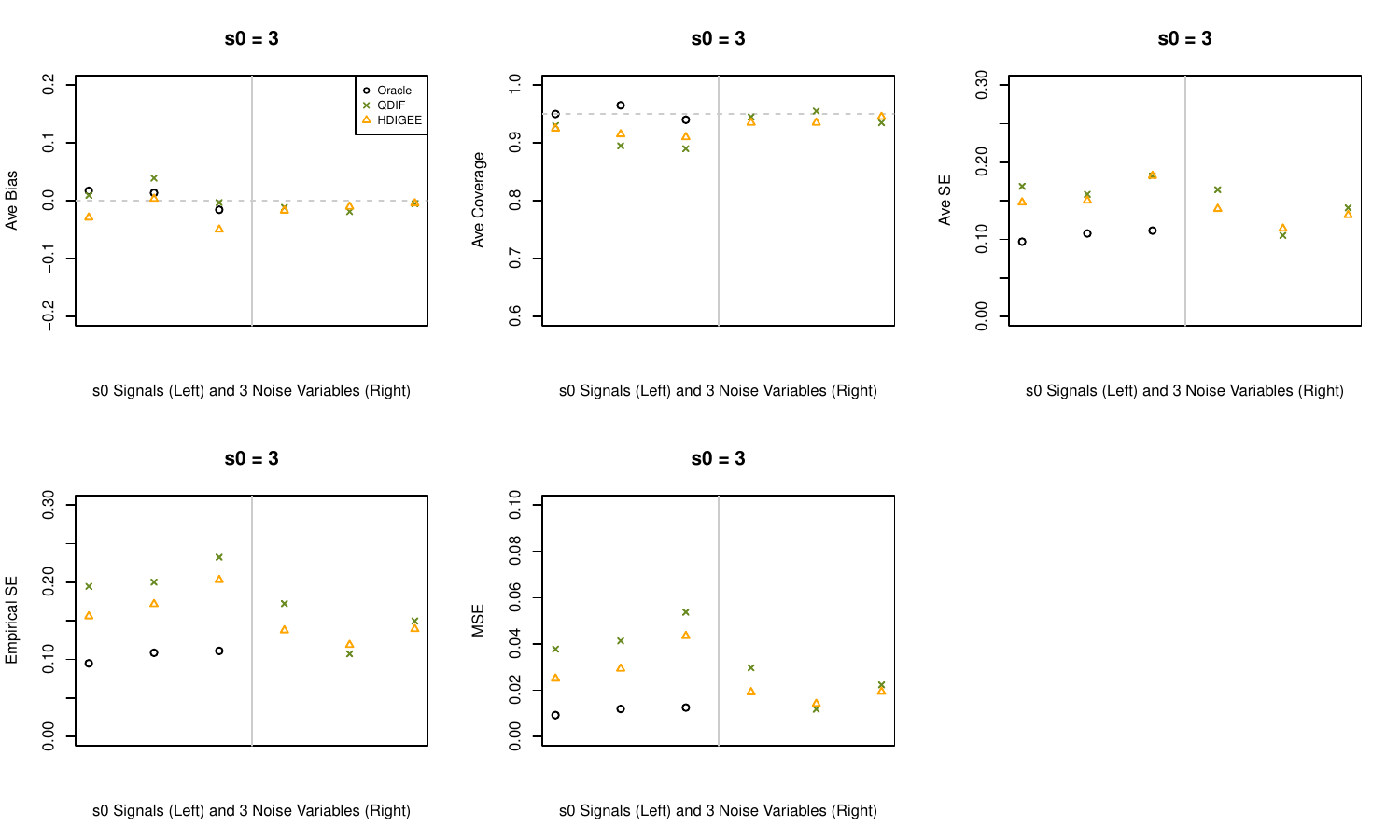}
    \caption{The second scenario mimicking  the real longitudinal COVID-19 proteomic data, with 659 measurements from $n=305$ subjects and $s_0=3$: estimation bias (Bias), coverage probability (Cov), model-based standard error (SE), empirical standard error (EmpSE) and mean squared error (MSE) for $s_0$ signals and 3 randomly chosen noise variables on the left and right of the vertical lines, respectively, average over 200 simulations}
    \label{fig:res_sim_covidprot_s3_3se}
\end{figure}

Figure~\ref{fig:sim_ribo_s3} shows the simulation results in the third setup with $s_0=3$. Comparing the best two methods from the previous simulations, HDIGEE and QDIF, we find that in this setting, QDIF results in much larger estimation bias for true signals than HDIGEE and the oracle estimator. Since QDIF produces unreasonably low standard error estimates for the truly nonzero coefficients, its coverage probabilities  are below 20\%. 
The instability of QDIF in this setting, compared to the previous two simulation scenarios, is likely due to the small sample size and unbalanced cluster size. In contrast, HDIGEE has more stable performance. The results with $s_0=10$ are shown in Supplementary Material~\ref{supp-supp:sec:add_sim_res}.

 	\begin{figure}[ht]
		\centering
  \includegraphics[width=0.8\textwidth]{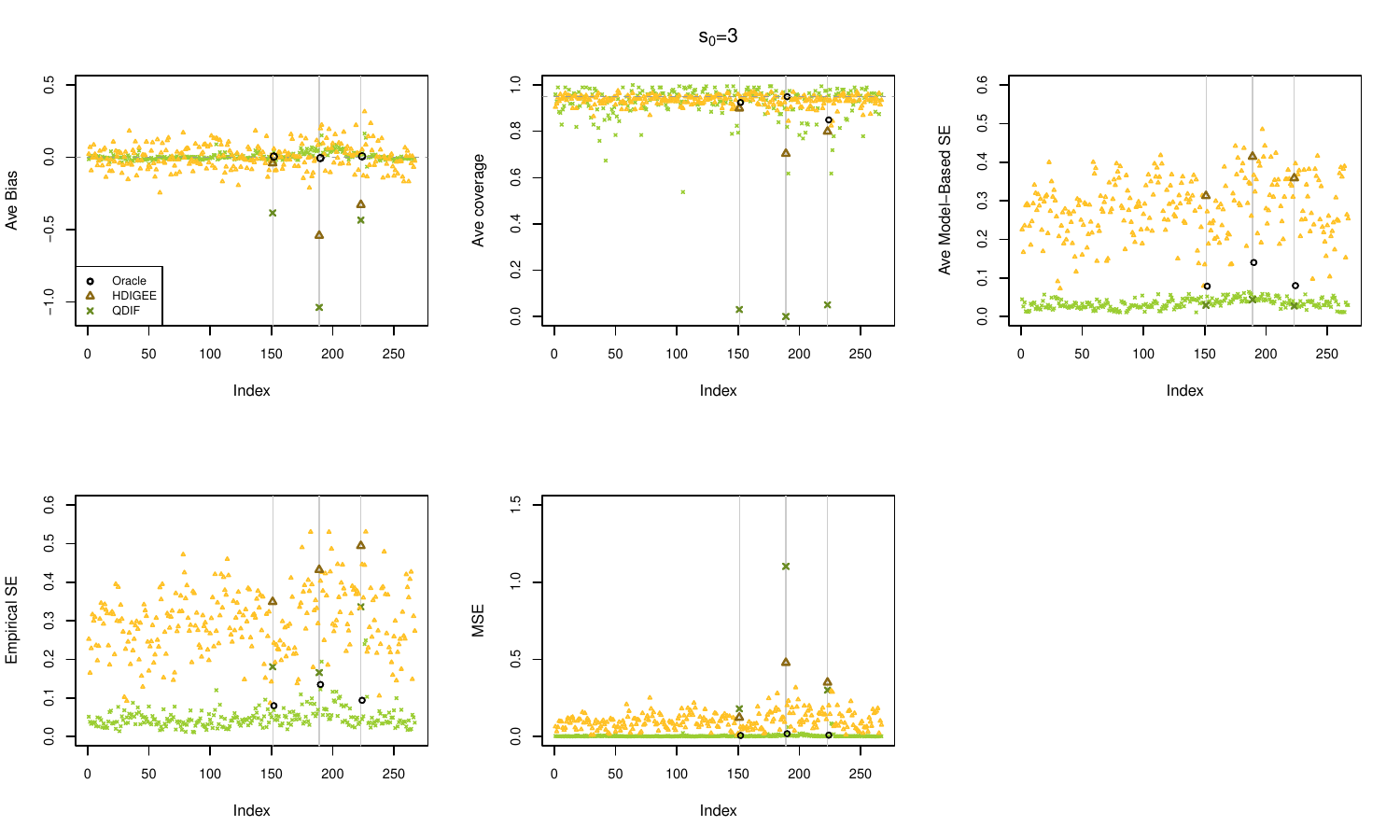}
		\caption{Regression coefficient estimates, empirical coverage probabilities and lengths for 95\% confidence intervals averaged over 200 replications, in the simulation setting derived from the real riboflavin data with $s_0=3$. Points with enlarged symbols and darker shades of colors represent the estimates for truly nonzero coefficients. 
        }
		\label{fig:sim_ribo_s3}
	\end{figure}

	
	\section{Concluding remarks}
 \label{sec:disc}
	
Motivated by the applications involving risk scores based on longitudinal \emph{omics} data, we proposed a projected estimating equation approach for inference on generalized estimating equations with high-dimensional covariates. The target of inference is a linear functional of the underlying regression coefficients $\bxi^T \bbeta^0$, which enables us to draw inference not only on the individual coefficients but also on prognostic scores. 
This feature is a key distinction of our approach from many existing approaches targeting low-dimensional parameters in high-dimensional inference and can be readily extended to other estimating equation problems.
	
Effective data-driven tuning parameter selection is essential in high-dimensional inference but is not fully addressed in the literature. While the performance of our method is not very sensitive to the choice of the tuning parameter for $\hbeta$, $\lambda$,  the second tuning parameter, $\lambda'$ in \eqref{eq:direction}, more directly governs the final de-biasing step. 
Our specialized cross-validation-based procedure in Algorithm~\ref{algo:cv_lambda_prime} provides a data-driven approach for selecting $\lambda'$. 
To further improve upon this procedure, future research can consider procedures that are more adaptive to signal strength.

In the existing high-dimensional inference literature, assumptions on the $\ell_0$ sparsity of inverse information matrix beyond linear regression models (e.g. generalized linear models and Cox proportional hazards models) have been scrutinized \citep{xia2021debiased,xia2021statistical}, as they may not hold in general settings.  The imposed condition on the sparsity of $(\bS^0)^{-1} \bxi$ in this article has certain similarity to those on the inverse information matrix. Estimating the projection direction when when $p \gg n$ is difficult without any structural conditions. Exploring relaxations of this sparsity assumption could be an important direction of future research.

\appendix
\setcounter{table}{0}
\renewcommand{\thetable}{A.\arabic{table}}
\bigskip
\begin{center}
{\large\bf APPENDIX}
\end{center}



\section{Assumptions}\label{appendix:assumptions}

Our proposed method is theoretically justified upon the following assumptions.
	
	\begin{assumption} \label{assump:bound_covs}
		The covariates $\bX_i$, as well as $\bX_i \bbeta^0$, are almost surely uniformly bounded; that is, there exists constants $K, K^{\prime} > 0$ such that $\| \bX_i \|_{\infty} \le K$, and  $\| \bX_i \bbeta^0 \|_{\infty} \le K^{\prime}$, $i = 1, \ldots, n$. Also, the standardized error terms, $\varepsilon_{ij} (\bbeta^0) = G_{ij}^{-1/2}(\bbeta^0) \left\{ Y_{ij} - \mu_{ij}(\bbeta^0)\right\}$, are sub-expotential.
	\end{assumption}
	
	\begin{assumption}  \label{assump:bdd_xomega}
		$\| \bX_i \bomega^0 \|_{\infty}$ is uniformly bounded for all $i = 1, \ldots, n$, almost surely.
	\end{assumption}
	
	\begin{assumption} \label{assump:var_fun}
		The mean and variance functions, $\mu(\zeta)$ and $v(\mu)$, are differentiable, and their derivatives are denoted by $\dot{\mu}(\zeta) = d \mu / d \zeta$ and $\dot{v}(\mu) = d v / d \mu$, respectively. There exist constants $\delta^{\prime}, \delta^{\prime\prime} > 0$ and $K_1,~ K_2, ~ c_{L} > 0$ such that 
		\[
		\begin{array}{c}
			\max_{\zeta_0 \in \{\bx_{ij}^\T \bbeta^0 \} } \sup_{ \{ \zeta: |\zeta - \zeta_0| \le \delta^{\prime}\} } \max \left\{  | \mu(\zeta) |, | \dot{\mu}(\zeta) |, 1/ | {\mu}(\zeta) | \right\} \le K_1,\\
			\max_{\mu_0 \in \{\mu(\bx_{ij}^\T \bbeta^0) \} } \sup_{ \{ \mu: |\mu - \mu_0| \le \delta^{\prime\prime} \} } \max \left\{ v(\mu), | \dot{v}(\mu) |, 1/ |{v}(\mu)| \right\} \le K_2, \\
			\max_{\mu_0 \in \{\mu(\bx_{ij}^\T \bbeta^0) \} } \sup_{  \{ ( \mu, \bar{\mu} ): |\mu - \mu_0| \le \delta^{\prime\prime}, |\bar{\mu} - \mu_0| \le \delta^{\prime\prime} \} } | v(\mu) - v(\bar{\mu}) | \le c_{L} | \mu - \bar{\mu} |.
		\end{array}
		\]
	\end{assumption}
	
	\begin{assumption} \label{assump:bdd_eigen_covs}
		There exists a constant $c > 0$ such that the minimum eigenvalue of the covariance matrix $\lambda_{\min} \left( \E(\bX_1^\T \bX_1) \right) \ge c$.
	\end{assumption}
	


\begin{assumption}\label{assump:cor_mat}
Let $r_n = \| \hatR - \barR \|$ denote the rate of convergence from the estimated correlation matrix $\hatR$ to its theoretical limit $\barR$. Then  
$
	r_n = \oP \left[ \left\{ \max \left(s_0, s^{\prime}\| \baromega^0 \|_1 \right) \left\{ \log(p) \right\}^{1/2} \| \baromega^0 \|_1 \right\}^{-1} \right].
$
Both $\barR$ and $\bR_0$, the true correlation matrix for $\bY_i$,  have their eigenvalues bounded from above and away from zero.
\end{assumption}
	
Assumption~\ref{assump:bound_covs} requires bounded covariate entries and mean responses, which is common in high-dimensional inference literature \citep{van2014asymptotically,javanmard2014confidence,zhang2014confidence,ning2017general,fang2020test}. The sub-exponential error condition in Assumption \ref{assump:bound_covs} is used to derive the  concentration inequalities in the proofs and holds for many commonly used models \citep{fang2020test}, such as linear regression with sub-Gaussian or sub-exponential errors or logistic regression. Assumption \ref{assump:bdd_xomega} is analogous to Assumption (iv) of Theorem 3.3 in \citet{van2014asymptotically}, and yet generalizes to a general loading vector $\bxi$ rather than targeting a single coefficient at a time. The mean and variance functions, $\mu(\cdot)$ and $v(\cdot)$,  as well as their derivatives, are uniformly bounded in the neighbourhood around the true linear predictors $\{ \bx_{ij}^\T \bbeta^0 \}$, as described in Assumption \ref{assump:var_fun}. In addition, as in \citet{van2014asymptotically}, the variance function $v(\cdot)$ is Lipschitz continuous. Assumption \ref{assump:bdd_eigen_covs} imposes bounded eigenvalues on the covariance matrix for the covariates $\bX_i$. Finally, Assumption \ref{assump:cor_mat} describes a sufficient rate of convergence from $\hat{\bR}$ to a fixed, but unknown and potentially misspecified $\barR$. The assumption extends the $n^{1/2}$-rate in the fixed dimension setting \citep{liang1986longitudinal} and involves both sparsities $s_0$ and $s^{\prime}$. 
In certain cases, it is possible to verify that the convergence of $\hatR$ in Assumption~\ref{assump:cor_mat} holds. One example is using a moment estimator, 
$
	\hatR =  n^{-1} \sum_{i=1}^{n} \bG_i^{-1/2} (\hbeta) \{ \bY_i - \bmu_i(\hbeta) \} \{ \bY_i - \bmu_i(\hbeta) \}^\T \bG_i^{-1/2} (\hbeta), 
$
when assuming unstructured working correlation matrix, similar to \citet{balan2005asymptotic}. 
In this case, $\hatR$ converges to $\barR = \bR_0$.  
See Supplementary Material \ref{supp-supp:sec:cond} for the detailed proof. 


\section{Additional simulation details}\label{appendix:sims}

Table~\ref{tab:linear_combn_setup} below lists 16 different cases of linear functionals to evaluate the performance of the proposed method, HDIGEE, in a continous outcome setting outlined in Section \ref{sec:sim}.1. 

\begin{table}[ht!]
   \footnotesize
    \caption{\small Setup of different  linear functionals to evaluate the performance of HDIGEE for the continuous outcome case with $(n, m) = (100, 5)$, $s_0=3$, $p=100$ or $500$}
    \begin{tabular}{llll} 
    \hline
    No. & $p=100 ~ (\beta_{6} = \beta_{45} = \beta_{82} = 1)$ & No. & $p=500 ~ (\beta_{134} = \beta_{338} = \beta_{429} = 1)$ \\
    \hline
    1 & $\xi_{3} = \xi_{4} = 1/\sqrt{2}$ & 9 & $\xi_{131} = \xi_{132} = 1/\sqrt{2}$ \\
    2 & $\xi_{100} = \xi_{101} = 1/\sqrt{2}$ & 10 & $\xi_{2} = \xi_{4} = 1/\sqrt{2}$ \\
    3 & $\xi_{2} = \xi_{6} = 1/\sqrt{2}$ & 11 & $\xi_{130} = \xi_{134} = 1/\sqrt{2}$ \\
    4 & $\xi_{6} = 1/\sqrt{2}, \xi_{45} = - 1/\sqrt{2}$ & 12 & $\xi_{134} = 1/\sqrt{2}, \xi_{338} = - 1/\sqrt{2}$ \\
    5 & $\xi_{3} = \xi_{4} = \xi_{8} = \xi_{9} = \xi_{10} =    1/\sqrt{5}$ & 13 & $\xi_{131} = \xi_{132} = \xi_{136} = \xi_{137} = \xi_{138} = 1/\sqrt{5}$ \\
    6 & $\xi_{97} = \xi_{98} = \xi_{99} = \xi_{100} = \xi_{101} =    1/\sqrt{5}$ & 14 & $\xi_{2} = \xi_{4} = \xi_{6} = \xi_{8} = \xi_{10} = 1/\sqrt{5}$ \\
    7 & $\xi_{2} = \xi_{3} = \xi_{6} = \xi_{9} = \xi_{10} =    1/\sqrt{5}$ & 15 & $\xi_{130} = \xi_{134} = \xi_{136} = \xi_{137} = \xi_{138} = 1/\sqrt{5}$ \\
    8 & $\xi_{2} = \xi_{6} = 1/\sqrt{5}, \xi_{45} = \xi_{100} = \xi_{101} = - 1/\sqrt{5}$ & 16 & $\xi_{130} = \xi_{134} = 1/\sqrt{5}, \xi_{338} = \xi_{399} = \xi_{400} = - 1/\sqrt{5}$ \\
    \hline
    \end{tabular}
    \label{tab:linear_combn_setup}
\end{table}

\bigskip
\begin{center}
{\large\bf SUPPLEMENTARY MATERIAL}
\end{center}

The supplementary material contains additional simulation results and additional details of the application to the longitudinal proteomic profiling data, the proofs for the theoretical justifications, and further discussion on the assumptions. R code for the implementation of the proposed method, HDIGEE, is made available on GitHub [GitHub link will be provided upon acceptance].






\bibliographystyle{apalike}

\bibliography{references}


\newpage

\begin{center}
	\textbf{\Large Supplementary Material for ``Inference for linear functionals of high-dimensional longitudinal proteomics data using generalized estimating equations"}
\end{center}

\setcounter{section}{0}
\renewcommand{\thesection}{S.\arabic{section}}
\renewcommand{\thetheorem}{S.\arabic{theorem}}
\renewcommand{\theequation}{S.\arabic{equation}}
\renewcommand{\thetable}{S.\arabic{table}}
\renewcommand{\thefigure}{S.\arabic{figure}}



In this supplementary material, we present additional simulation results in Section \ref{supp-supp:sec:add_sim_res}, and additional results from the longitudinal proteomic profiling for COVID-19 severity in Section \ref{supp-supp:sec:add_covid_res}. Section \ref{supp-supp:sec:proof} presents all the technical proofs of lemmas and the main theorem. Some technical conditions related to the theoretical justification in the main text are further discussed in Section \ref{supp-supp:sec:cond}. Finally, we include a modified version of the high-dimensional inference method for estimating equations in the main text, with more details outlined in Section \ref{supp-supp:sec:hdee}.


\section{Additional simulation results}
\label{supp-supp:sec:add_sim_res}

\subsection{Continuous outcome setting in the first simulation scenario}

In the main text, we only present the simulation results from the binary outcome case. Here, we include the simulation results from the continuous outcome case in the first simulation scenario in Table~\ref{tab:sim_linear_n100m5_3se}.


\begin{table}[h!]
	\small
	\caption{\lu{Continuous outcome scenario with $(n, m) = (100, 5)$: estimation bias, coverage probability of the 95\% confidence interval, model-based standard error and empirical standard error for all signal and three randomly chosen noise variables, averaged over 200 replications}}
	\begin{tabular}{llrrrrrrrr} 
		\hline
		&  & \multicolumn{4}{c}{Signal $(s_0)$} & \multicolumn{4}{c}{Noise $(3)$} \\
		Method & Item & \multicolumn{2}{c}{$p=100$} & \multicolumn{2}{c}{$p=500$} & \multicolumn{2}{c}{$p=100$} & \multicolumn{2}{c}{$p=500$} \\
		&  & $s_0=3$ & $s_0=10$ & $s_0=3$ & $s_0=10$ & $s_0=3$ & $s_0=10$ & $s_0=3$ & $s_0=10$ \\[5pt]
		\hline
		\multicolumn{10}{c}{Autoregressive} \\
		\multirow{4}{*}{HDIGEE}  &   Bias & 0.000 & -0.003 & -0.009 & -0.009 & -0.003 & 0.000 & -0.002 & -0.001 \\ 
		& Cov & 0.933 & 0.921 & 0.927 & 0.920 & 0.947 & 0.928 & 0.945 & 0.948 \\ 
		& SE & 0.055 & 0.054 & 0.056 & 0.055 & 0.054 & 0.054 & 0.056 & 0.054 \\ 
		& EmpSE & 0.057 & 0.058 & 0.060 & 0.061 & 0.053 & 0.059 & 0.057 & 0.055 \\[5pt]
		\multirow{4}{*}{HDEE} & Bias & -0.002 & -0.005 & -0.003 & -0.005 & -0.002 & 0.006 & -0.002 & -0.001 \\ 
		& Cov & 0.882 & 0.866 & 0.925 & 0.909 & 0.917 & 0.858 & 0.932 & 0.937 \\ 
		& SE & 0.041 & 0.040 & 0.055 & 0.053 & 0.041 & 0.040 & 0.055 & 0.053 \\ 
		& EmpSE & 0.048 & 0.051 & 0.060 & 0.060 & 0.046 & 0.050 & 0.059 & 0.057 \\[5pt] 
		\multirow{4}{*}{QDIF} & Bias & -0.005 & -0.007 & -0.009 & -0.013 & -0.002 & 0.004 & -0.002 & 0.000 \\ 
		& Cov & 0.940 & 0.915 & 0.928 & 0.893 & 0.948 & 0.925 & 0.957 & 0.945 \\ 
		& SE & 0.049 & 0.048 & 0.048 & 0.047 & 0.048 & 0.047 & 0.047 & 0.046 \\ 
		& EmpSE & 0.051 & 0.053 & 0.054 & 0.053 & 0.047 & 0.051 & 0.048 & 0.046 \\ 
		\hline
		\multicolumn{10}{c}{Unstructured} \\
		\multirow{4}{*}{HDIGEE}   & Bias & -0.001 & -0.003 & -0.009 & -0.009 & -0.003 & 0.001 & -0.002 & 0.000 \\ 
		& Cov & 0.935 & 0.926 & 0.925 & 0.916 & 0.948 & 0.922 & 0.942 & 0.947 \\ 
		& SE & 0.053 & 0.052 & 0.054 & 0.053 & 0.051 & 0.051 & 0.054 & 0.052 \\ 
		& EmpSE & 0.055 & 0.055 & 0.057 & 0.058 & 0.051 & 0.056 & 0.055 & 0.053 \\[5pt]
		\multirow{4}{*}{HDEE} & Bias & -0.002 & -0.005 & -0.003 & -0.005 & -0.002 & 0.006 & -0.002 & -0.001 \\ 
		& Cov & 0.892 & 0.870 & 0.925 & 0.908 & 0.932 & 0.863 & 0.920 & 0.933 \\ 
		& SE & 0.040 & 0.039 & 0.053 & 0.052 & 0.039 & 0.039 & 0.053 & 0.051 \\ 
		& EmpSE & 0.047 & 0.049 & 0.058 & 0.059 & 0.044 & 0.049 & 0.057 & 0.055 \\[5pt] 
		\multirow{4}{*}{QDIF} & Bias & -0.005 & -0.007 & -0.010 & -0.013 & -0.003 & 0.004 & -0.002 & -0.001 \\ 
		& Cov & 0.937 & 0.917 & 0.933 & 0.896 & 0.950 & 0.923 & 0.960 & 0.947 \\ 
		& SE & 0.047 & 0.046 & 0.047 & 0.045 & 0.046 & 0.045 & 0.046 & 0.045 \\ 
		& EmpSE & 0.048 & 0.051 & 0.052 & 0.052 & 0.045 & 0.048 & 0.046 & 0.044 \\ 
		\hline
	\end{tabular}
	\label{tab:sim_linear_n100m5_3se}
\end{table}

\subsection{Cross-validation criteria}

We present additional simulation results under the same setting as in Section 5 of the main text with an equal number of observations $m_i = m = 5$ for each subject. The cross-validation procedure described in the main text chooses the smallest tuning parameter $\lambda^{\prime}$ value  that results in a cross-validated criterion value $\textsc{cv}_l$ within three standard errors above the minimum criterion value (as obtained over the pre-specified grid of trial points for $\lambda^{\prime}$). Another natural choice is the  tuning parameter $\lambda^{\prime}$ value that gives the smallest cross-validated criterion value, denoted by ``HDIGEE (Min)" or ``HDEE (Min)". In this subsection, we show that HDIGEE (Min) and HDEE (Min) have inferior performance compared to the proposed cross-validation, using the first simulation scenario in Section 5 as an example. 

All simulation results are displayed in Tables~\ref{tab:add_sim_linear_n100m5_ar1}-\ref{tab:add_sim_logit_n100m5_un}, corresponding to the four combinations between regression models (identity and logit links) and true correlation structures [AR(1) and unstructured] as described in the main text. It is obvious that using the minimum cross-validation criterion generates larger estimation biases and poorer coverage probabilities. 


\begin{table}[h!]
	\centering
	\caption{Continuous outcome scenario with $(n, m) = (100, 5)$ and autoregressive correlation structure: estimation bias, coverage probability and length of the 95\% confidence interval, and empirical standard error for all signal and three randomly chosen noise variables, averaged over 200 replications}
	\begin{threeparttable}
		\begin{tabular}{llrrrrrrrr}
			\hline
			&  & \multicolumn{4}{c}{Signal $(s_0)$} & \multicolumn{4}{c}{Noise $(3)$} \\
			Method & Item & \multicolumn{2}{c}{$p=100$} & \multicolumn{2}{c}{$p=500$} & \multicolumn{2}{c}{$p=100$} & \multicolumn{2}{c}{$p=500$} \\
			&  & $s_0=3$ & $s_0=10$ & $s_0=3$ & $s_0=10$ & $s_0=3$ & $s_0=10$ & $s_0=3$ & $s_0=10$ \\[5pt] \hline
			\multirow{4}{*}{HDIGEE (Min)}  & Bias & -0.005 & -0.007 & -0.016 & -0.018 & -0.003 & 0.003 & -0.003 & 0.000 \\ 
			& Cov & 0.910 & 0.895 & 0.885 & 0.872 & 0.952 & 0.920 & 0.950 & 0.948 \\ 
			& Len & 0.194 & 0.189 & 0.180 & 0.174 & 0.190 & 0.189 & 0.177 & 0.173 \\ 
			& EmpSE & 0.055 & 0.055 & 0.053 & 0.053 & 0.048 & 0.053 & 0.046 & 0.044 \\[5pt]  
			\multirow{4}{*}{HDEE (Min)}  &   Bias & -0.005 & -0.008 & -0.009 & -0.014 & 0.000 & 0.017 & -0.001 & 0.004 \\ 
			& Cov & 0.740 & 0.683 & 0.627 & 0.610 & 0.772 & 0.680 & 0.725 & 0.700 \\ 
			& Len & 0.099 & 0.095 & 0.085 & 0.082 & 0.096 & 0.094 & 0.083 & 0.081 \\ 
			& EmpSE & 0.044 & 0.047 & 0.045 & 0.046 & 0.039 & 0.042 & 0.040 & 0.038 \\[5pt]  
			\hline
		\end{tabular}
		\begin{tablenotes}
			\item HDIGEE, the proposed high-dimensional inference for generalized estimating equations; HDEE, the modified version of high-dimensional inference for estimating equations  by \citet{neykov2018unified}; Bias, estimation bias; Cov, coverage probability of 95\% confidence interval; Len, length of confidence interval; EmpSE, empirical standard error. 
		\end{tablenotes}
	\end{threeparttable}
	\label{tab:add_sim_linear_n100m5_ar1}
\end{table}


\begin{table}[h!]
	\centering
	\caption{Continuous outcome scenario with $(n, m) = (100, 5)$ and unstructured correlation: estimation bias, coverage probability and length of the 95\% confidence interval, and empirical standard error for all signal and three randomly chosen noise variables, averaged over 200 replications}
	\begin{threeparttable}
		\begin{tabular}{llrrrrrrrr}
			\hline
			&  & \multicolumn{4}{c}{Signal $(s_0)$} & \multicolumn{4}{c}{Noise $(3)$} \\
			Method & Item & \multicolumn{2}{c}{$p=100$} & \multicolumn{2}{c}{$p=500$} & \multicolumn{2}{c}{$p=100$} & \multicolumn{2}{c}{$p=500$} \\
			&  & $s_0=3$ & $s_0=10$ & $s_0=3$ & $s_0=10$ & $s_0=3$ & $s_0=10$ & $s_0=3$ & $s_0=10$ \\[5pt] \hline
			\multirow{4}{*}{HDIGEE (Min)} & Bias & -0.005 & -0.007 & -0.017 & -0.018 & -0.003 & 0.003 & -0.002 & 0.001 \\ 
			& Cov & 0.915 & 0.899 & 0.875 & 0.861 & 0.952 & 0.917 & 0.955 & 0.950 \\ 
			& Len & 0.186 & 0.182 & 0.169 & 0.167 & 0.181 & 0.181 & 0.169 & 0.165 \\ 
			& EmpSE & 0.052 & 0.053 & 0.051 & 0.051 & 0.046 & 0.051 & 0.043 & 0.042 \\[5pt]
			\multirow{4}{*}{HDEE (Min)}   & Bias & -0.005 & -0.008 & -0.009 & -0.014 & 0.000 & 0.016 & -0.001 & 0.004 \\ 
			& Cov & 0.742 & 0.677 & 0.622 & 0.604 & 0.765 & 0.665 & 0.722 & 0.698 \\ 
			& Len & 0.094 & 0.091 & 0.082 & 0.079 & 0.091 & 0.091 & 0.079 & 0.077 \\ 
			& EmpSE & 0.042 & 0.045 & 0.044 & 0.045 & 0.038 & 0.041 & 0.039 & 0.036 \\[5pt]
			\hline
		\end{tabular}
		\begin{tablenotes}
			\item HDIGEE, the proposed high-dimensional inference for generalized estimating equations; HDEE, the modified version of high-dimensional inference for estimating equations  by \citet{neykov2018unified}; Bias, estimation bias; Cov, coverage probability of 95\% confidence interval; Len, length of confidence interval; EmpSE, empirical standard error. 
		\end{tablenotes}
	\end{threeparttable}
	\label{tab:add_sim_linear_n100m5_un}
\end{table}


\begin{table}[h!]
	\centering
	\caption{Binary outcome scenario with $(n, m) = (100, 5)$ and autoregressive correlation: estimation bias, coverage probability and length of the 95\% confidence interval, and empirical standard error for all signal and three randomly chosen noise variables, averaged over 200 replications}
	\begin{threeparttable}
		\begin{tabular}{llrrrrrrrr}
			\hline
			&  & \multicolumn{4}{c}{Signal $(s_0)$} & \multicolumn{4}{c}{Noise $(3)$} \\
			Method & Item & \multicolumn{2}{c}{$p=100$} & \multicolumn{2}{c}{$p=500$} & \multicolumn{2}{c}{$p=100$} & \multicolumn{2}{c}{$p=500$} \\
			&  & $s_0=3$ & $s_0=10$ & $s_0=3$ & $s_0=10$ & $s_0=3$ & $s_0=10$ & $s_0=3$ & $s_0=10$ \\[5pt] \hline
			\multirow{4}{*}{HDIGEE (Min)} & Bias & -0.029 & -0.044 & -0.055 & -0.107 & -0.001 & 0.006 & -0.001 & 0.007 \\ 
			& Cov & 0.873 & 0.845 & 0.843 & 0.700 & 0.943 & 0.922 & 0.960 & 0.955 \\ 
			&  Len & 0.413 & 0.430 & 0.371 & 0.361 & 0.412 & 0.422 & 0.372 & 0.361 \\ 
			&  EmpSE & 0.125 & 0.140 & 0.118 & 0.118 & 0.109 & 0.118 & 0.094 & 0.091 \\[5pt]
			\multirow{4}{*}{HDEE (Min)} &   Bias & -0.040 & -0.073 & -0.044 & -0.106 & -0.002 & 0.036 & 0.001 & 0.017 \\ 
			&  Cov & 0.677 & 0.558 & 0.543 & 0.378 & 0.763 & 0.723 & 0.673 & 0.690 \\ 
			&  Len & 0.207 & 0.215 & 0.175 & 0.172 & 0.206 & 0.214 & 0.173 & 0.172 \\ 
			&  EmpSE & 0.098 & 0.111 & 0.101 & 0.101 & 0.089 & 0.094 & 0.086 & 0.078 \\[5pt] \hline
		\end{tabular}
		\begin{tablenotes}
			\item HDIGEE, the proposed high-dimensional inference for generalized estimating equations; HDEE, the modified version of high-dimensional inference for estimating equations  by \citet{neykov2018unified}; Bias, estimation bias; Cov, coverage probability of 95\% confidence interval; Len, length of confidence interval; EmpSE, empirical standard error. 
		\end{tablenotes}
	\end{threeparttable}
	\label{tab:add_sim_logit_n100m5_ar1}
\end{table}


\begin{table}[h!]
	\centering
	\caption{Binary outcome scenario with $(n, m) = (100, 5)$ and unstructured correlation: estimation bias, coverage probability and length of the 95\% confidence interval, and empirical standard error for all signal and three randomly chosen noise variables, averaged over 200 replications}
	\begin{threeparttable}
		\begin{tabular}{llrrrrrrrr}
			\hline
			&  & \multicolumn{4}{c}{Signal $(s_0)$} & \multicolumn{4}{c}{Noise $(3)$} \\
			Method & Item & \multicolumn{2}{c}{$p=100$} & \multicolumn{2}{c}{$p=500$} & \multicolumn{2}{c}{$p=100$} & \multicolumn{2}{c}{$p=500$} \\
			&  & $s_0=3$ & $s_0=10$ & $s_0=3$ & $s_0=10$ & $s_0=3$ & $s_0=10$ & $s_0=3$ & $s_0=10$ \\[5pt] \hline
			\multirow{4}{*}{HDIGEE (Min)} & Bias & -0.028 & -0.043 & -0.053 & -0.104 & -0.005 & 0.007 & 0.002 & 0.004 \\ 
			&  Cov & 0.868 & 0.837 & 0.840 & 0.695 & 0.935 & 0.925 & 0.958 & 0.947 \\ 
			&  Len & 0.406 & 0.428 & 0.369 & 0.357 & 0.406 & 0.421 & 0.366 & 0.360 \\ 
			&  EmpSE & 0.125 & 0.140 & 0.114 & 0.120 & 0.109 & 0.119 & 0.090 & 0.091 \\[5pt]
			\multirow{4}{*}{HDEE (Min)} &   Bias & -0.036 & -0.072 & -0.042 & -0.104 & -0.003 & 0.035 & 0.001 & 0.014 \\ 
			&  Cov & 0.642 & 0.560 & 0.577 & 0.387 & 0.753 & 0.713 & 0.697 & 0.692 \\ 
			&  Len & 0.206 & 0.214 & 0.175 & 0.170 & 0.203 & 0.212 & 0.171 & 0.169 \\ 
			&  EmpSE & 0.099 & 0.112 & 0.097 & 0.103 & 0.090 & 0.094 & 0.084 & 0.078 \\[5pt] \hline
		\end{tabular}
		\begin{tablenotes}
			\item HDIGEE, the proposed high-dimensional inference for generalized estimating equations; HDEE, the modified version of high-dimensional inference for estimating equations  by \citet{neykov2018unified}; Bias, estimation bias; Cov, coverage probability of 95\% confidence interval; Len, length of confidence interval; EmpSE, empirical standard error. 
		\end{tablenotes}
	\end{threeparttable}
	\label{tab:add_sim_logit_n100m5_un}
\end{table}

\subsection{The second simulation scenario mimicking the longitudinal proteomic data for COVID-19 severity}

For the second scenario of Section 5.1, we have presented the simulation results for binary outcomes with the true model size $s_0=3$. Table~\ref{fig:res_sim_covidprot_s10_3se} shows the simulation results for a denser model with $s_0 = 10$.

\begin{figure}[ht!]
	\centering
	\includegraphics[width=0.9\textwidth]{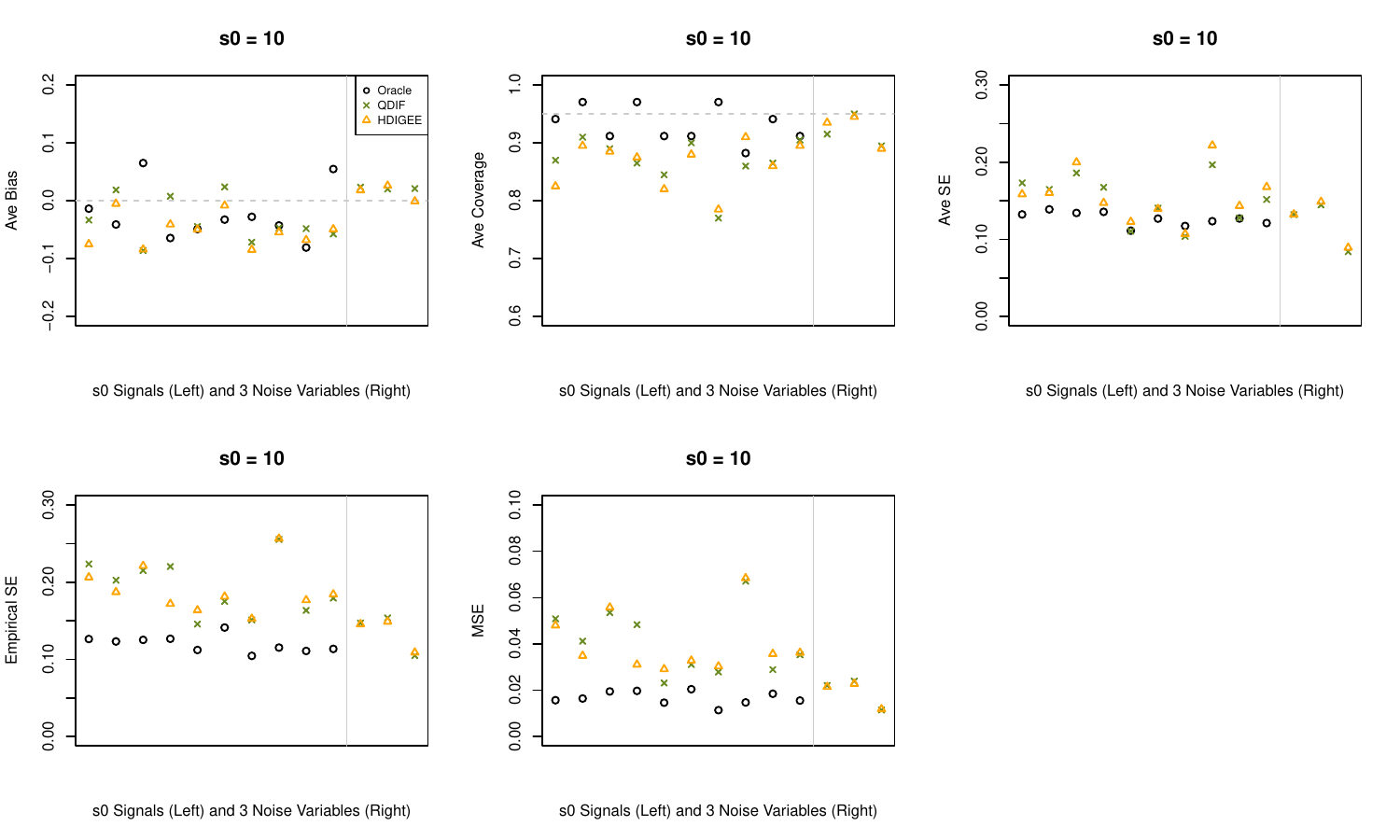}
	\caption{The second scenario mimicking  the real longitudinal COVID-19 proteomic data, with 659 measurements from $n=305$ subjects and $s_0=10$: estimation bias (Bias), coverage probability (Cov), model-based standard error (SE), empirical standard error (EmpSE) and mean squared error (MSE) for $s_0$ signals and 3 randomly chosen noise variables on the left and right of the vertical lines, respectively, average over 200 simulations}
	\label{fig:res_sim_covidprot_s10_3se}
\end{figure}

\subsection{The third simulation scenario mimicking the riboflavin production data analyzed by \citet{buhlmann2014high}}

In the third scenario of Section 5.1, we carry out additional simulations with continuous outcomes that mimic the scale and complexity of the riboflavin production data analyzed by \citet{buhlmann2014high}. The number of observations $m_i$ varies among different subjects from 2 to 6, totaling $N = 111$ observations from $n = 28$ strains of \textit{B. subtilis}. After filtering out the genes whose expression levels have coefficient of variation less than 0.1, we focus on the most variable $p = 267$ genes, reduced from the original 4,088 genes.
We set the true $\bbeta^0$ with $s_0 = 3$ or $10$ randomly chosen nonzero coefficients, out of $p=267$, with values realized from $\mathrm{Uniform}(0.5, 1.5)$. With these coefficients, we then simulate $N= 111$ continuous observations from $n=28$ groups. Since each cluster in the real data has 2 to 6 measurements, we take the upper left corner submatrix of $\bR^0$ below with corresponding dimensions to be the true correlation matrix for each cluster: 
\begin{equation*} \footnotesize
	\bR^0 = \left(  
	\begin{array}{cccccc}
		1.00 & 0.50 & 0.45 & 0.40 & 0.35 & 0.30  \\ 
		0.50 & 1.00 & 0.50 & 0.45 & 0.40 & 0.35  \\ 
		0.45 & 0.50 & 1.00 & 0.50 & 0.45 & 0.40  \\ 
		0.40 & 0.45 & 0.50 & 1.00 & 0.50 & 0.45  \\ 
		0.35 & 0.40 & 0.45 & 0.50 & 1.00 & 0.50  \\ 
		0.30 & 0.35 & 0.40 & 0.45 & 0.50 & 1.00 \\ 
	\end{array}  \right).
\end{equation*}
The estimate $\hatR$ is obtained using AR(1) working correlation.

For this simulation scenario, we have presented the results with $s_0 = 3$ in the main text. Figure~\ref{fig:sim_ribo_s10} shows the simulation results with $s_0=10$.

\begin{figure}[ht]
	\centering
	\includegraphics[width=0.9\textwidth]{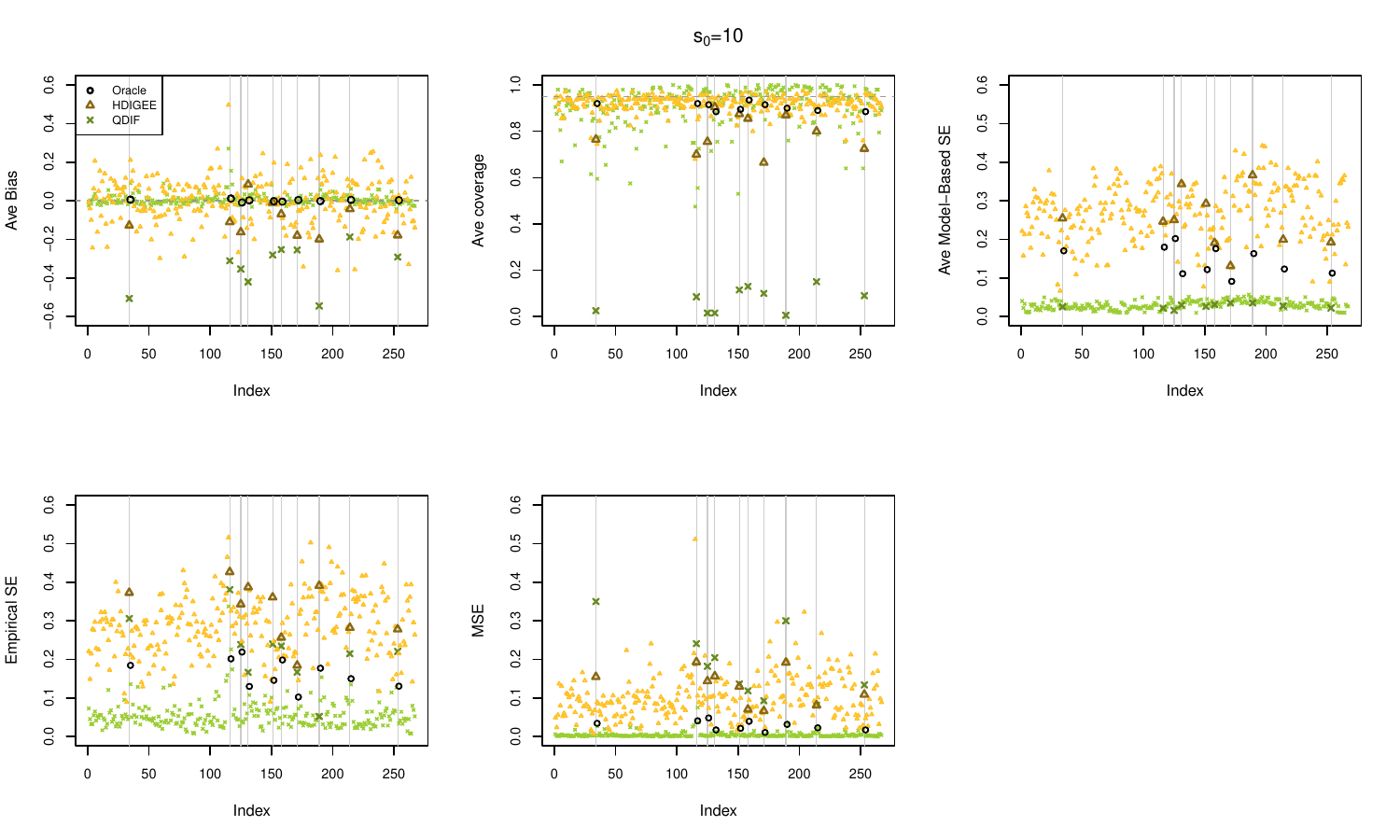}
	\caption{Regression coefficient estimates, empirical coverage probabilities and lengths for 95\% confidence intervals averaged over 200 replications, in the simulation setting derived from the real riboflavin data with $s_0=10$. Points with enlarged symbols and darker shades of colors represent the estimates for truly nonzero coefficients. 
	}
	\label{fig:sim_ribo_s10}
\end{figure}


\section{Additional results from the longitudinal proteomic analysis of COVID-19 patient plasma samples}
\label{supp-supp:sec:add_covid_res}

A full list of gene names, Olink IDs, UniProt IDs, estimated coefficients by lasso and the proposed method, model-based standard errors by the proposed method, and adjusted p-values for the significant proteins is delayed below (Table \ref{tab:covid_prot_add}).

\begin{table}[ht]
	\centering
	\caption{Additional results from the longitudinal proteomic analysis of COVID-19 patient plasma samples}
	\label{tab:covid_prot_add}
	\small
	\begin{threeparttable}
		\begin{tabular}{rllllrrrr}
			\hline
			& Assay & OlinkID & UniProt & Panel & Lasso & Beta & SE & Q-value \\ 
			\hline
			1 & FLT3LG & OID20661 & P49771 & INFLAMMATION & 0.355 & 0.433 & 0.077 & $1.28 \times 10^{-5}$ \\ 
			2 & MUC13 & OID20866 & Q9H3R2 & NEUROLOGY & 0.266 & 0.556 & 0.114 & $3.81\times 10^{-4}$ \\ 
			3 & CALB1 & OID21306 & P05937 & ONCOLOGY & 0.326 & 0.670 & 0.151 & $2.39 \times 10^{-3}$ \\ 
			4 & C19orf12 & OID20804 & Q9NSK7 & NEUROLOGY & -0.175 & -0.404 & 0.095 & $4.30 \times 10^{-3}$ \\ 
			5 & IL15RA & OID20498 & Q13261 & INFLAMMATION & 0.297 & 0.461 & 0.112 & $5.87 \times 10^{-3}$ \\ 
			6 & MEP1B & OID20168 & Q16820 & CARDIOMETABOLIC & 0.271 & 0.500 & 0.129 & $1.37 \times 10^{-2}$ \\ 
			7 & HMOX1 & OID20217 & P09601 & CARDIOMETABOLIC & -0.194 & -0.573 & 0.154 & $1.70 \times 10^{-2}$ \\ 
			8 & TCL1B & OID20058 & O95988 & CARDIOMETABOLIC & -0.179 & -0.259 & 0.070 & $1.70 \times 10^{-2}$ \\ 
			9 & CD22 & OID20637 & P20273 & INFLAMMATION & -0.155 & -0.542 & 0.143 & $1.70 \times 10^{-2}$ \\ 
			10 & CDHR5 & OID20193 & Q9HBB8 & CARDIOMETABOLIC & -0.230 & -0.463 & 0.127 & $2.13 \times 10^{-2}$ \\ 
			11 & SPINK6 & OID21450 & Q6UWN8 & ONCOLOGY & -0.394 & -0.368 & 0.102 & $2.29 \times 10^{-2}$ \\ 
			12 & SIRPA & OID20304 & P78324 & CARDIOMETABOLIC & -0.194 & -0.457 & 0.136 & $4.64 \times 10^{-2}$ \\ 
			13 & WFIKKN1 & OID20939 & Q96NZ8 & NEUROLOGY & -0.179 & -0.308 & 0.092 & $4.64 \times 10^{-2}$ \\ 
			14 & CR2 & OID20393 & P20023 & CARDIOMETABOLIC & -0.194 & -0.433 & 0.130 & $4.77 \times 10^{-2}$ \\ 
			\hline
		\end{tabular}
		\begin{tablenotes}
			\item \footnotesize{Lasso: estimated coefficients by lasso; Beta: estimated coefficients by the proposed method HDIGEE; SE: model-based standard error by the proposed method HDIGEE.}
		\end{tablenotes}
	\end{threeparttable}
\end{table}


\section{Technical proofs}
\label{supp-supp:sec:proof}

For the convenience of the readers, here we repeat the assumptions used in the main text.

\begin{assumption} \label{assump:bound_covs}
	The covariates $\bX_i$, as well as $\bX_i \bbeta^0$, are almost surely uniformly bounded; that is, there exists constants $K, K^{\prime} > 0$ such that $\| \bX_i \|_{\infty} \le K$, and  $\| \bX_i \bbeta^0 \|_{\infty} \le K^{\prime}$, $i = 1, \ldots, n$. Also, the standardized error terms, $\varepsilon_{ij} (\bbeta^0) = G_{ij}^{-1/2}(\bbeta^0) \left\{ Y_{ij} - \mu_{ij}(\bbeta^0)\right\}$, are sub-expotential.
\end{assumption}

\begin{assumption}  \label{assump:bdd_xomega}
	$\| \bX_i \bomega^0 \|_{\infty}$ is uniformly bounded for all $i = 1, \ldots, n$, almost surely.
\end{assumption}

\begin{assumption} \label{assump:var_fun}
	The mean and variance functions, $\mu(\zeta)$ and $v(\mu)$, are differentiable, and their derivatives are denoted by $\dot{\mu}(\zeta) = d \mu / d \zeta$ and $\dot{v}(\mu) = d v / d \mu$, respectively. There exist constants $\delta^{\prime}, \delta^{\prime\prime} > 0$ and $K_1,~ K_2, ~ c_{L} > 0$ such that 
	\[
	\begin{array}{c}
		\max_{\zeta_0 \in \{\bx_{ij}^\T \bbeta^0 \} } \sup_{ \{ \zeta: |\zeta - \zeta_0| \le \delta^{\prime}\} } \max \left\{  | \mu(\zeta) |, | \dot{\mu}(\zeta) |, 1/ | {\mu}(\zeta) | \right\} \le K_1,\\
		\max_{\mu_0 \in \{\mu(\bx_{ij}^\T \bbeta^0) \} } \sup_{ \{ \mu: |\mu - \mu_0| \le \delta^{\prime\prime} \} } \max \left\{ v(\mu), | \dot{v}(\mu) |, 1/ |{v}(\mu)| \right\} \le K_2, \\
		\max_{\mu_0 \in \{\mu(\bx_{ij}^\T \bbeta^0) \} } \sup_{  \{ ( \mu, \bar{\mu} ): |\mu - \mu_0| \le \delta^{\prime\prime}, |\bar{\mu} - \mu_0| \le \delta^{\prime\prime} \} } | v(\mu) - v(\bar{\mu}) | \le c_{L} | \mu - \bar{\mu} |.
	\end{array}
	\]
\end{assumption}

\begin{assumption} \label{assump:bdd_eigen_covs}
	There exists a constant $c > 0$ such that the minimum eigenvalue of the covariance matrix $\lambda_{\min} \left( \E(\bX_1^\T \bX_1) \right) \ge c$.
\end{assumption}

\begin{assumption}
	\label{assump:cor_mat}
	Let $r_n = \| \hatR - \barR \|$ denote the rate of convergence from the estimated correlation matrix $\hatR$ to its theoretical limit $\barR$. Then  
	$
	r_n = \oP \left[ \left\{ \max \left(s_0, s^{\prime}\| \baromega^0 \|_1 \right) \left\{ \log(p) \right\}^{1/2} \| \baromega^0 \|_1 \right\}^{-1} \right].
	$
	Both $\barR$ and $\bR_0$, the true correlation matrix for $\bY_i$,  have their eigenvalues bounded from above and away from zero.
\end{assumption}

\subsection{Proofs of technical lemmas}

Throughout the proofs, we will use $C$, $C^{\prime}$ and $C^{\prime\prime}$ for generic absolute constants, and their specific meanings vary from case to case.


We first provide a useful lemma for the representation of the derivative of $\bPsi(\bbeta)$, denoted as $\partial \bPsi(\bbeta) / \partial \bbeta^\T$. This is essentially the same as Lemma 3.2 of \citet{wang2011gee}, and thus we will omit the derivation.

\begin{lemma} \label{lemma:deriv_psi}
	The derivative of $\bPsi(\bbeta)$ can be rewritten as
	\[
	\displaystyle \frac{\partial \bPsi(\bbeta)}{\partial \bbeta^\T} = - \bS(\bbeta) + \bE_n(\bbeta) + \bF_n(\bbeta),
	\]
	where
	\[
	\begin{array}{rcl}
		\bS(\bbeta) & = & \displaystyle \frac{1}{n} \sum_{i=1}^n \bX_i^\T \bG_i^{1/2}(\bbeta) \hatR^{-1} \bG_i^{1/2}(\bbeta) \bX_i,  \\
		\bE_n(\bbeta) & = & \displaystyle - \frac{1}{2n} \sum_{i=1}^n \sum_{j=1}^m \dot{v}(\mu_{ij}(\bbeta)) G_{ij}^{-1/2}(\bbeta) (Y_{ij} - \mu_{ij}(\bbeta)) \bX_i^\T \bG_i^{1/2}(\bbeta) \hatR^{-1} \bare_j \bare_j^\T\bX_i, \\
		\bF_n(\bbeta) & = & \displaystyle \frac{1}{2n} \sum_{i=1}^n \sum_{j=1}^m \dot{v}(\mu_{ij}(\bbeta)) G_{ij}^{1/2}(\bbeta) \bx_{ij} \bx_{ij}^\T \bare_j^\T \hatR^{-1} \bG_i^{-1/2} (\bbeta) (\bY_i - \bmu_i(\bbeta)).
	\end{array}
	\]
	Here $\bare_j$ is a $m$-dimensional unit vector with the $j$th element being one and the others being zero. 
	
\end{lemma}


Lemma \ref{lemma:feasible_sol_omega0} studies the convergence rate of $	\| \bS(\hbeta) \baromega^0 - \bxi \|_{\infty}$. This result determines the desirable rate for the tuning parameter $\lambda^{\prime}$; with a sufficiently large $\lambda^{\prime}$, $\baromega^0$  is admissible  to the constraint in the constrained $\ell_1$ minimization problem (4).

\begin{lemma} \label{lemma:feasible_sol_omega0}
	Under Assumptions \ref{assump:bound_covs}--\ref{assump:cor_mat}, we have
	\[
	\| \bS(\hbeta) \baromega^0 - \bxi \|_{\infty} = \OP \left\{ \| \baromega^0 \|_1 \left( s_0 \lambda + r_n \right) \right\}.
	\]
	Thus, it suffices to take $\lambda^{\prime} \asymp \| \baromega^0 \|_1 ( s_0 \lambda + r_n \| )$.
\end{lemma}


\begin{proof}[Proof of Lemma \ref{lemma:feasible_sol_omega0}]
	Since $\bS(\hbeta) \baromega^0 - \bxi = \{ \bS(\hbeta) - \bS^0 \} (\bS^0)^{-1} \bxi = \{ \bS(\hbeta) - \bS^0 \}  \baromega^0$,
	\begin{equation*}
		\begin{array}{rl}
			\| \bS(\hbeta) \baromega^0 - \bxi \|_{\infty} & \le \| \bS(\hbeta) - \bS^0 \|_{\infty} \| \baromega^0 \|_1.
		\end{array}
	\end{equation*}
	We then focus on $$\| \bS(\hbeta) - \bS^0 \|_{\infty} \le \| \bS(\hbeta) - \barS(\hbeta) \|_{\infty} + \| \barS(\hbeta) - \barS(\bbeta^0) \|_{\infty} + \| \barS(\bbeta^0) - \bS^0 \|_{\infty}.$$
	
	For the first term in the upper bound of $\| \bS(\hbeta) - \bS^0 \|_{\infty}$, 
	\[
	\begin{array}{rl}
		\| \bS(\hbeta) - \barS(\hbeta) \|_{\infty} & = \max_{1 \le j, k \le p} \left| \displaystyle \frac{1}{n} \sum_{i=1}^n \bx_{i[j]}^\T \bG_i^{1/2}(\hbeta) [\hatR^{-1} - \barR^{-1}] \bG_i^{1/2} (\hbeta) \bx_{i[k]} \right|  \\
		& \le \max_{1 \le j, k \le p} \displaystyle \frac{1}{n} \sum_{i=1}^n \| \bG_i^{1/2}(\hbeta) \bx_{i[j]} \|_2 \| \hatR^{-1} - \barR^{-1} \| \| \bG_i^{1/2}(\hbeta) \bx_{i[k]} \|_2 \\
		& \le \max_{1 \le j, k \le p} \displaystyle \frac{1}{n} \sum_{i=1}^n \| \bG_i(\hbeta) \|  \| \hatR^{-1} - \barR^{-1} \| \| \bx_{i[k]} \|_2 \| \bx_{i[j]} \|_2 \\
		& = \OP(1)  \| \hatR^{-1} - \barR^{-1} \|,
	\end{array}
	\]
	where the last equality holds due to Assumptions \ref{assump:bound_covs} and \ref{assump:var_fun},   that $\bx_{i[j]}$ and $\bx_{i[k]} \in \mR^m$ are fixed dimensional vectors, and that Lemma 1 guarantees $\hbeta$ converges to $\bbeta^0$. Since $\| \hatR^{-1} - \barR^{-1} \| = \|  \hatR^{-1} ( \barR - \hatR )  \barR^{-1} \|  \le \| \hatR^{-1} \| \times \| \hatR - \barR \| \times \| \barR^{-1} \|$, $\| \hatR^{-1} - \barR^{-1} \| \le \OP(r_n)$ and $\| \bS(\hbeta) - \barS(\hbeta) \|_{\infty} = \OP(r_n)$.
	
	For the second term in the upper bound of $\| \bS(\hbeta) - \bS^0 \|_{\infty}$, 
	\begin{equation*}
		\begin{array}{rl}
			\barS(\hbeta) - \barS(\bbeta^0)  & = \displaystyle \frac{1}{n} \sum_{i=1}^n \bX_i^\T \bG_i^{1/2}(\hbeta) \barR^{-1} \bG_i^{1/2}(\hbeta) \bX_i - \displaystyle \frac{1}{n} \sum_{i=1}^n \bX_i^\T \bG_i^{1/2}(\bbeta^0) \barR^{-1} \bG_i^{1/2}(\bbeta^0) \bX_i   \\
			&  = \left\{ \displaystyle \frac{1}{n} \sum_{i=1}^n \bX_i^\T \bG_i^{1/2}(\hbeta) \barR^{-1} \bG_i^{1/2}(\hbeta) \bX_i - \displaystyle \frac{1}{n} \sum_{i=1}^n \bX_i^\T \bG_i^{1/2}(\hbeta) \barR^{-1} \bG_i^{1/2}(\bbeta^0) \bX_i \right\}  \\
			& ~ + \left\{ \displaystyle \frac{1}{n} \sum_{i=1}^n \bX_i^\T \bG_i^{1/2}(\hbeta) \barR^{-1} \bG_i^{1/2}(\bbeta^0) \bX_i - \displaystyle \frac{1}{n} \sum_{i=1}^n \bX_i^\T \bG_i^{1/2}(\bbeta^0) \barR^{-1} \bG_i^{1/2}(\bbeta^0) \bX_i \right\} \\
			& = I_{n1} + I_{n2}.
		\end{array}
	\end{equation*}
	Then,
	\begin{equation*}
		\begin{array}{rl}
			\| I_{n1} \|_{\infty} & = \| n^{-1} \sum_{i=1}^n \bX_i^\T \bG_i^{1/2}(\hbeta) \barR^{-1} [ \bG_i(\hbeta) - \bG_i^{1/2}(\bbeta^0) ] \bX_i \|_{\infty}  \\
			& = \max_{1 \le j, k \le p} \| n^{-1} \sum_{i=1}^n \bx_{i[j]}^\T \bG_i^{1/2}(\hbeta) \barR^{-1} [ \bG_i(\hbeta) - \bG_i^{1/2}(\bbeta^0) ] \bx_{i[k]} \|_{\infty} \\
			& \le \max_{1 \le j, k \le p} n^{-1} \sum_{i=1}^n \| \bG_i^{1/2}(\hbeta) \bx_{i[j]} \|_1 \| \barR^{-1} \|_{\infty} \|   [ \bG_i(\hbeta) - \bG_i^{1/2}(\bbeta^0) ] \bx_{i[k]} \|_1.
		\end{array}
	\end{equation*}
	Due to Assumptions \ref{assump:bound_covs} and \ref{assump:var_fun}, $\| \bG_i^{1/2}(\hbeta) \bx_{i[j]} \|_1 = \OP(mK)$, and $\| \{ \bG_i^{1/2}(\hbeta) - \bG_i^{1/2}(\bbeta^0) \} \bx_{i[k]} \|_1 = \OP(m K^2 \| \hbeta - \bbeta^0 \|_1)$. Since $\| \barR^{-1} \|_{\infty} = {O}(1)$, then by Lemma 1, $\| I_{n1} \|_{\infty} = \OP(s_0 \lambda)$. Similarly, we can show that $\| I_{n2} \|_{\infty} = \OP(s_0 \lambda)$. Hence, the second term $\| \barS(\hbeta) - \barS(\bbeta^0) \|_{\infty} = \OP(s_0 \lambda)$.
	
	We rewrite the third term in the upper bound of $\| \bS(\hbeta) - \bS^0 \|_{\infty}$ as
	\begin{equation*}
		\barS(\bbeta^0) - \bS^0  = \displaystyle \frac{1}{n} \sum_{i=1}^n \left[ \bX_i^\T \bG_i^{1/2}(\bbeta^0) \barR^{-1} \bG_i^{1/2}(\bbeta^0) \bX_i - \E \left\{ \bX_i^\T \bG_i^{1/2}(\bbeta^0) \barR^{-1} \bG_i^{1/2}(\bbeta^0) \bX_i \right\} \right].
	\end{equation*}
	For the $(j,k)$th element in the above matrix, since 
	\begin{align*}
		| \bx_{i[j]}^\T \bG_i^{1/2}(\bbeta^0) \barR^{-1} \bG_i^{1/2}(\bbeta^0) \bx_{i[k]} | & \le \| \bG_i^{1/2}(\bbeta^0) \bx_{i[j]} \|_1 \| \bG_i^{1/2}(\bbeta^0) \bx_{i[k]} \|_1 \| \barR^{-1} \|_{\infty} \\
		& \le (m K_1^{1/2} K)^2 \| \barR^{-1} \|_{\infty},
	\end{align*}
	we have 
	\begin{equation*}
		\begin{array}{c}
			| \bx_{i[j]}^\T \bG_i^{1/2}(\bbeta^0) \barR^{-1} \bG_i^{1/2}(\bbeta^0) \bx_{i[k]} - \E\{ \bx_{i[j]}^\T \bG_i^{1/2}(\bbeta^0) \barR^{-1} \bG_i^{1/2}(\bbeta^0) \bx_{i[k]} \} |   \le 2 m^2 K_1 K^2 \| \barR^{-1} \|_{\infty},  \\
			\E[ \bx_{i[j]}^\T \bG_i^{1/2}(\bbeta^0) \barR^{-1} \bG_i^{1/2}(\bbeta^0) \bx_{i[k]} - \E\{ \bx_{i[j]}^\T \bG_i^{1/2}(\bbeta^0) \barR^{-1} \bG_i^{1/2}(\bbeta^0) \bx_{i[k]} \} ]   = 0.
		\end{array}
	\end{equation*}
	By Hoeffding's inequality (see, e.g., Lemma 14.11 in \citealt{buhlmann2011statistics}), for any $t > 0$, 
	\[
	\begin{array}{rl}
		\PP \left( \displaystyle \left|  [\barS(\bbeta^0) - \bS^0]_{jk} \right| > t \right) &  \le 2 \exp[ - n t^2/ \{ 2 (2 m^2 K_1 K^2 \| \barR^{-1} \|_{\infty})^2 \} ] \\
		& = 2 \exp\{ - n t^2/ ( 8 m^4 K_1^2 K^4 \| \barR^{-1} \|_{\infty}^2 ) \}.
	\end{array}
	\]
	Then, $$\PP \left( \| \barS(\bbeta^0) - \bS^0 \|_{\infty} > t \right) \le \sum_{j,k} \PP \left( \displaystyle \left|  [\barS(\bbeta^0) - \bS^0]_{jk} \right| > t \right) \le 2 p^2 \exp\{ - n t^2/ ( 8 m^4 K_1^2 K^4 \| \barR^{-1} \|_{\infty}^2 ) \}.$$ 
	Since $m$ and $\| \barR^{-1} \|_{\infty}$ are bounded by absolute constants, 
	\[
	\| \barS(\bbeta^0) - \bS^0 \|_{\infty} = \OP[ \{ \log(p) / n \}^{1/2} ].
	\]
	
	Combining the three terms in $\bS(\hbeta) - \bS^0$ and taking $\lambda \asymp \{ \log(p)/n \}^{1/2}$, we have
	\begin{equation} \label{eq:hatS_to_S0_max}
		\| \bS(\hbeta) - \bS^0 \|_{\infty} = \OP(s_0 \lambda + r_n).
	\end{equation}
	Finally,
	$
	\| \bS(\hbeta) \baromega^0 - \bxi \|_{\infty} = \OP \{  \| \baromega^0 \|_1 ( s_0 \lambda + r_n ) \}.
	$
\end{proof}


Before studying the convergence rate of the projection direction $\widehat{\bomega}$ in Lemma \ref{lemma:hat_omega_rate}, we introduce the following lemma that guarantees the compatibility condition holds for $\bS(\hbeta)$ with large probability.

\begin{lemma} \label{lemma:compat_omega}
	Let 
	\[
	\kappa_D (s^{\prime}) = \inf \left\{ 
	\displaystyle \frac{(s^{\prime})^{1/2} \{ \bnu ^\T \bS(\hbeta) \bnu \}^{1/2}}{\| \bnu_{\calS^{\prime}} \|_1}: ~ \bnu \in \mR^p, \bnu \ne 0, \| \bnu_{(\calS^{\prime})^c} \|_1 \le L \| \bnu_{\calS^{\prime}} \|_1 \mathrm{~ for ~ some ~} L > 0
	\right\}.
	\]
	Assume $s^{\prime}  ( s_0 \lambda + r_n ) = o(1)$. Under Assumptions \ref{assump:bound_covs}--\ref{assump:cor_mat}, $\kappa_D(s^{\prime}) \ge \kappa$ for some constant $\kappa > 0$ with probability going to one.
\end{lemma}


\begin{proof}[Proof of Lemma \ref{lemma:compat_omega}]
	
	Since $\| \bnu_{\mathcal{S}^{\prime}} \|_1 \le (s^{\prime})^{1/2} \| \bnu_{\mathcal{S}^{\prime}} \|_2 \le (s^{\prime})^{1/2} \| \bnu \|_2$,
	\[
	\kappa_D^2(s^{\prime}) \ge \inf \displaystyle \left\{ \frac{\bnu ^\T \bS(\hbeta) \bnu}{\| \bnu \|_2^2}: ~ \bnu \in \mR^p, \bnu \ne 0, \| \bnu_{(\calS^{\prime})^c} \|_1 \le L \| \bnu_{\calS^{\prime}} \|_1 \mathrm{~ for ~ some ~} L > 0 \right\},
	\]
	and
	\[
	\begin{array}{rcl}
		| \bnu ^\T \{ \bS(\hbeta) - \bS^0 \} \bnu | & \le & \| \bnu \|_1^2 \| \bS(\hbeta) - \bS^0 \|_{\infty}  \\
		& \le & (L+1)^2 \| \bnu_{\mathcal{S}^{\prime}} \|_1^2 \| \bS(\hbeta) - \bS^0 \|_{\infty} \\
		& \le & (L+1)^2 s^{\prime} \| \bnu \|_2^2 \| \bS(\hbeta) - \bS^0 \|_{\infty}.
	\end{array}
	\]
	Because we have shown $\| \bS(\hbeta) - \bS^0 \|_{\infty} = \OP(s_0 \lambda + r_n)$ in \eqref{eq:hatS_to_S0_max}, then 
	$$| \bnu ^\T \{ \bS(\hbeta) - \bS^0 \} \bnu | / \| \bnu \|_2^2 \le (L+1)^2 s^{\prime} \| \bS(\hbeta) - \bS^0 \|_{\infty} = \OP\{ (L+1)^2 s^{\prime} (s_0 \lambda +r_n) \} = \oP(1).$$ 
	Therefore, 
	\[
	\displaystyle \frac{\bnu ^\T \bS(\hbeta) \bnu}{\| \bnu \|_2^2} = \frac{\bnu ^\T \bS^0 \bnu}{\| \bnu \|_2^2} + \frac{\bnu ^\T \{\bS(\hbeta) - \bS^0 \} \bnu}{\| \bnu \|_2^2} = \frac{\bnu ^\T \bS^0 \bnu}{\| \bnu \|_2^2} + \oP(1), 
	\]
	and with probability going to one, 
	\[
	\displaystyle \frac{\bnu ^\T \bS(\hbeta) \bnu}{\| \bnu \|_2^2} \ge \frac{1}{2} \frac{\bnu ^\T \bS^0 \bnu}{\| \bnu \|_2^2} \ge \frac{1}{2} \lambda_{\min}(\bS^0).
	\]
	It suffices to take $\kappa = \{ \lambda_{\min}(\bS^0) / 2 \}^{1/2}$, and based on Assumptions \ref{assump:bound_covs}, \ref{assump:var_fun}, \ref{assump:bdd_eigen_covs} and \ref{assump:cor_mat}, it is easy to see that the smallest eigenvalue of $\bS^0$ is bounded away from zero.
\end{proof}


Next, we study the convergence of the projection direction $\widehat{\bomega}$ for a general target $\theta^0 = \bxi^\T \bbeta^0$, where $\| \bxi \|_2 = 1$ without loss of generality. In reality, we only need $\| \bxi \|_2$ to be bounded. Recall that we solve the constrained $\ell_1$ minimization problem
\[
\tilde{\bomega} = \argmin_{\bomega \in \mR^p} \{ \| \bomega \|_1: ~ \| \bS(\hbeta) \bomega - \bxi \|_{\infty} \le \lambda^{\prime} \},
\]
and then rescale $\tilde{\bomega}$ to obtain $\hat{\bomega} = \tilde{\bomega} / \{ \tilde{\bomega}^\T \bS(\hbeta) \tilde{\bomega} \}$.

\begin{lemma} \label{lemma:hat_omega_rate}
	Let $\tilde{\bomega} = \argmin \{ \| \bomega \|_1: \| \bS(\hbeta) \bomega - \bxi \|_{\infty} \le \lambda^{\prime} \}$. Suppose Assumptions \ref{assump:bound_covs}--\ref{assump:cor_mat} hold and $s^{\prime} \lambda^{\prime} \asymp s^{\prime} \| \baromega^0 \|_1 (s_0 \lambda + r_n) = o(1)$.  Then 
	\[
	\| \tilde{\bomega} - \baromega^0 \|_1 = \OP(s^{\prime} \lambda^{\prime}).
	\]
	Since $\{ \tilde{\bomega}^\T \bS(\hbeta) \tilde{\bomega} \}^{-1} = \OP(1)$, it is valid to define $\hat{\bomega} = \tilde{\bomega} / \{ \tilde{\bomega}^\T \bS(\hbeta) \tilde{\bomega} \}$ with probability going to one, and
	\[
	\| \hat{\bomega} - \bomega^0 \|_1 = \OP(s^{\prime} \lambda^{\prime} \| \baromega^0 \|_1 ).
	\]
\end{lemma}

\begin{proof}[Proof of Lemma \ref{lemma:hat_omega_rate}]
	
	We first study the convergence rate of the non-normalized projection direction $\tilde{\bomega}$.
	By the construction of $\tilde{\bomega}$,  on the event $\mathcal{A} = \{ \| \bS(\hbeta) \baromega^0 - \bxi \|_{\infty} \le \lambda^{\prime} \}$,
	\begin{equation} \label{eq:omega_l1_ineq1}
		\| \tilde{\bomega} \|_1 \le  \| \baromega^0 \|_1.
	\end{equation}
	Denote the difference $\hat\bDelta = \tilde{\bomega} - \baromega^0$. By triangle inequality, we have $|\bar{\omega}_j^0| - |\tilde{\omega}_j| \le |\widehat{\Delta}_j| = |\tilde{\omega}_j - \bar{\omega}_j^0|$, and then
	\begin{equation} \label{eq:omega_l1_ineq2}
		\| \baromega^0 \|_1 = \sum_{j \in \calS^{\prime}} | \bar{\omega}_j^0 | \ge \| \tilde{\bomega} \|_1 = \sum_{j \in \calS^{\prime}} |\tilde{\omega}_j| + \sum_{j \in (\calS^{\prime})^c} | \tilde{\omega}_j |
		\ge \left( \sum_{j \in \calS^{\prime}} | \bar{\omega}_j^0 | -  \sum_{j \in \calS^{\prime}} | \widehat{\Delta}_j | \right) + \sum_{j \in (\calS^{\prime})^c} | \tilde{\omega}_j |,
	\end{equation}
	where the first inequality holds due to \eqref{eq:omega_l1_ineq1}. Then, \eqref{eq:omega_l1_ineq2} results in the following inequality:
	\[
	\sum_{j \in \calS^{\prime}} | \bar{\omega}_j^0 | \ge \sum_{j \in \calS^{\prime}} | \bar{\omega}_j^0 | -  \sum_{j \in \calS^{\prime}} | \widehat{\Delta}_j | + \sum_{j \in (\calS^{\prime})^c} | \tilde{\omega}_j |,
	\]
	which implies that
	\begin{equation} \label{eq:L_eq_1}
		\| \hat{\bDelta}_{\mathcal{S}^{\prime}} \|_1 =   \sum_{j \in \calS^{\prime}} | \widehat{\Delta}_j | \ge  \sum_{j \in (\calS^{\prime})^c} | \tilde{\omega}_j |  = \sum_{j \in (\calS^{\prime})^c} | \tilde{\omega}_j - \bar{\omega}_j^0 |  = \| \widehat{\bm{\Delta}}_{(\calS^{\prime})^c} \|_1.
	\end{equation}
	
	Next, we will invoke the compatility condition in Lemma \ref{lemma:compat_omega} to show the rate of $\| \hat{\bDelta} \|_1$. Consider $\hat{\bDelta}^\T \bS(\hbeta) \hat{\bDelta}$ that appears in the compatibility condition of Lemma \ref{lemma:compat_omega} with $\bnu = \hat{\bDelta}$. By definition,
	\begin{equation}
		\begin{array}{rl}
			\hat{\bDelta}^\T \bS(\hbeta) \hat{\bDelta}    & =  \hat{\bDelta}^\T \bS(\hbeta) (\tilde{\bomega} - \baromega^0) \\
			& = \hat{\bDelta}^\T ( \bS(\hbeta) \tilde{\bomega} - \bxi + \bxi - \bS(\hbeta) \baromega^0 ). \\
		\end{array}
	\end{equation}
	Then, on the event $\mathcal{A}$, $| \hat{\bDelta}^\T \bS(\hbeta) \hat{\bDelta} | \le \| \hat{\bDelta} \|_1 \{ \| \bS(\hbeta) \tilde{\bomega} - \bxi \|_{\infty} + \| \bS(\hbeta) \baromega^0 - \bxi \|_{\infty} \} \le 2 \| \hat{\bDelta} \|_1 \lambda^{\prime}$.
	
	By Lemma \ref{lemma:compat_omega}, $\PP ( \kappa_D(s^{\prime}) \ge \kappa )$ goes to one as $n \to \infty$. Due to \eqref{eq:L_eq_1}, without loss of generality, we choose $L=1$ in Lemma \ref{lemma:compat_omega}. On the event $\{ \kappa_D(s^{\prime}) \ge \kappa \}$, we have
	$
	\| \hat{\bm{\Delta}}_{\calS^{\prime}} \|_1 \lesssim (s^{\prime})^{1/2} ( \hat{\bDelta}^\T \bS(\hbeta) \hat{\bDelta} )^{1/2} \lesssim (s^{\prime} \lambda^{\prime}  \| \hat{\bDelta} \|_1 )^{1/2}.
	$
	Combined with \eqref{eq:L_eq_1}, this implies that 
	\[
	\| \hat{\bDelta} \|_1 = \| \hat{\bDelta}_{\calS^{\prime}} \|_1 + \| \hat{\bDelta}_{(\calS^{\prime})^c} \|_1 \le 2 \| \hat{\bDelta}_{\calS^{\prime}} \|_1 \lesssim (s^{\prime} \lambda^{\prime}  \| \hat{\bDelta} \|_1 )^{1/2}.
	\]
	Therefore, $\| \hat{\bDelta} \|_1 = \| \tilde{\bomega} - \baromega^0 \|_1 = \OP(s^{\prime} \lambda^{\prime})$.

	Next, we study the rate of $\| \homega - \bomega^0 \|_1$ for the normalized projection direction $\homega$. By definition,
	\begin{equation} \label{eq:diff_homega_omega0}
		\homega - \bomega^0 = \displaystyle \frac{\tilde{\bomega}}{\tilde{\bomega}^\T \bS(\hbeta) \tilde{\bomega}} - \frac{\bar{\bomega}^0}{(\bar{\bomega}^0)^\T \bS^0 \bar{\bomega}^0} = \frac{ \{ (\bar{\bomega}^0)^\T \bS^0 \bar{\bomega}^0 \} \tilde{\bomega}  -  \{ \tilde{\bomega}^\T \bS(\hbeta) \tilde{\bomega} \} \bar{\bomega}^0 }{ \{ \tilde{\bomega}^\T \bS(\hbeta) \tilde{\bomega} \} \{ (\bar{\bomega}^0)^\T \bS^0 \bar{\bomega}^0 \} }.
	\end{equation}
	We will show that the denominator of \eqref{eq:diff_homega_omega0}, i.e., $\{ \tilde{\bomega}^\T \bS(\hbeta) \tilde{\bomega} \} \{ (\bar{\bomega}^0)^\T \bS^0 \bar{\bomega}^0 \}$, is bounded away from zero with large probability, and that $\| \homega - \bomega^0 \|_1$ is in the same order as the numerator of \eqref{eq:diff_homega_omega0}, which is $\OP(s^{\prime} \lambda^{\prime} \| \bar{\bomega} ^0 \|_1)$.
	
	In the denominator of \eqref{eq:diff_homega_omega0}, $| \tilde{\bomega}^\T \bS(\hbeta) \tilde{\bomega} - (\bar{\bomega}^0)^\T \bS^0 \bar{\bomega}^0 | = \OP(s^{\prime} \lambda^{\prime}) = \oP(1)$. To see this, we have
	\[
	\begin{array}{rl}
		& | \tilde{\bomega}^\T \bS(\hbeta) \tilde{\bomega} - (\bar{\bomega}^0)^\T \bS^0 \bar{\bomega}^0 | \\
		\le & | (\tilde{\bomega} - \bar{\bomega}^0)^\T \bS(\hbeta) \tilde{\bomega} | + | (\bar{\bomega}^0)^\T \bS(\hbeta) (\tilde{\bomega} - \bar{\bomega}^0) | + | (\bar{\bomega}^0)^\T \{ \bS(\hbeta) - \bS^0 \} \bar{\bomega}^0 | \\
		\le & \| \tilde{\bomega} - \bar{\bomega}^0 \|_1 \| \bS(\hbeta) \tilde{\bomega} \|_{\infty} + \| \tilde{\bomega} - \bar{\bomega}^0 \|_1 \| \bS(\hbeta) \bar{\bomega}^0 \|_{\infty} + | (\bar{\bomega}^0)^\T \{ \bS(\hbeta) - \bS^0 \} \bar{\bomega}^0 | \\
		\le & \| \tilde{\bomega} - \bar{\bomega}^0 \|_1 (\| \bxi \|_{\infty} + \lambda^{\prime}) + \| \tilde{\bomega} - \bar{\bomega}^0 \|_1 \OP(\| \bxi \|_{\infty} + \lambda^{\prime}) + | (\bar{\bomega}^0)^\T \{ \bS(\hbeta) - \bS^0 \} \bar{\bomega}^0 |   \\
		= & \OP(s^{\prime} \lambda^{\prime}) + | (\bar{\bomega}^0)^\T \{ \bS(\hbeta) - \bS^0 \} \bar{\bomega}^0 |.
	\end{array}
	\]
	Due to Assumption \ref{assump:bdd_xomega}, we have $\| \bX_i \bar{\bomega}^0 \|_{\infty}$ is uniformly bounded for all $i = 1, \ldots, n$, and similar to showing $\| \bS(\hbeta) - \bS^0 \|_{\infty} = \OP(s_0 \lambda + r_n)$ in Lemma \ref{lemma:feasible_sol_omega0}, we see that $| (\bar{\bomega}^0)^\T \{ \bS(\hbeta) - \bS^0 \} \bar{\bomega}^0 | = \OP(s_0 \lambda + r_n) \lesssim \OP(s^{\prime} \lambda^{\prime})$. Hence, $\tilde{\bomega}^\T \bS(\hbeta) \tilde{\bomega} = (\bar{\bomega}^0)^\T \bS^0 \bar{\bomega}^0 + \oP(1)$. Since $0 < \lambda_{\max}^{-1}(\bS^0) \le (\bar{\bomega}^0)^\T \bS^0 \bar{\bomega}^0 = \bxi^\T (\bS^0)^{-1} \bxi \le \lambda_{\min}^{-1}(\bS^0)$, it is easy to see that $\tilde{\bomega}^\T \bS(\hbeta) \tilde{\bomega}$ is strictly positive and bounded away from zero with large probability, and that $\{ \tilde{\bomega}^\T \bS(\hbeta) \tilde{\bomega} \}^{-1} \{ (\bar{\bomega}^0)^\T \bS^0 \bar{\bomega}^0 \}^{-1} = \OP(1)$. We then examine the numerator of \eqref{eq:diff_homega_omega0}, where
	\[
	\begin{array}{rl}
		& \| \{ (\bar{\bomega}^0)^\T \bS^0 \bar{\bomega}^0 \} \tilde{\bomega}  -  \{ \tilde{\bomega}^\T \bS(\hbeta) \tilde{\bomega} \} \bar{\bomega}^0 \|_1 \\
		\le & \| \{ (\bar{\bomega}^0)^\T \bS^0 \bar{\bomega}^0 \} \tilde{\bomega} - \{ (\bar{\bomega}^0)^\T \bS^0 \bar{\bomega}^0 \} \bar{\bomega}^0 \|_1 + \| \bar{\bomega}^0 \|_1 | \{ (\bar{\bomega}^0)^\T \bS^0 \bar{\bomega}^0 \} -  \{ \tilde{\bomega}^\T \bS(\hbeta) \tilde{\bomega} \} | \\
		= & (\bar{\bomega}^0)^\T \bS^0 \bar{\bomega}^0 \| \tilde{\bomega} - \bar{\bomega}^0 \|_1 +  \| \bar{\bomega}^0 \|_1 | \{ (\bar{\bomega}^0)^\T \bS^0 \bar{\bomega}^0 \} -  \{ \tilde{\bomega}^\T \bS(\hbeta) \tilde{\bomega} \} | \\
		= & {O}(1) \OP(s^{\prime} \lambda^{\prime}) + \OP(s^{\prime} \lambda^{\prime}) \| \bar{\bomega}^0 \|_1 \\
		= & \OP(s^{\prime} \lambda^{\prime} \| \bar{\bomega}^0 \|_1 ).
	\end{array}
	\]
	Hence, $\| \homega - \bomega^0 \|_1 = \OP(s^{\prime} \lambda^{\prime} \| \baromega^0 \|_1)$.
\end{proof}


The following lemma determines $\| \barPsi(\bbeta^0) \|_{\infty} = \OP[ \{ \log(p)/n \}^{1/2} ]$ for the surrogate estimating function $\barPsi(\bbeta^0)$ using concentration inequalities, whose rate is useful in proving the asymptotic normality of the proposed de-biased estimator for $\theta^0$. Although the components of the standardized residuals $\beps_i (\bbeta^0) =  \bG_i^{-1/2}(\bbeta^0)(\bY_i - \bmu_i(\bbeta^0))$ are assumed as sub-exponential in Assumption \ref{assump:bound_covs}, sub-Gaussian error terms are a special case of sub-exponential errors \citep{vershynin2018high} and also widely used;  we will include a separate short discussion in the proof on this issue. 

\begin{lemma} \label{lemma:psi_infty}
	Let the standardized residuals $\beps_i (\bbeta^0) =  \bG_i^{-1/2}(\bbeta^0)(\bY_i - \bmu_i(\bbeta^0)) = (\varepsilon_{i1}(\bbeta^0), \ldots, \varepsilon_{im}(\bbeta^0))^\T $. Under Assumptions \ref{assump:bound_covs}, \ref{assump:var_fun} and \ref{assump:cor_mat},  $\| \barPsi(\bbeta^0) \|_{\infty} = \OP [ \{ \log(p)/n \}^{1/2} ]$.
\end{lemma}

\begin{proof}[Proof of Lemma \ref{lemma:psi_infty}]
	
	Recall that $\barPsi(\bbeta^0) = n^{-1} \sum_{i=1}^n \bX_i^\T \bG_i^{1/2}(\bbeta^0) \barR^{-1} \beps_i(\bbeta^0)$, and for $j= 1, \ldots, p$, the $j$th element of $\bar{\Psi}(\bbeta^0)$,  $\bar{\Psi}_j(\bbeta^0) = n^{-1} \sum_{i=1}^n \bx_{i[j]}^\T \bG_i^{1/2}(\bbeta^0) \barR^{-1} \beps_i(\bbeta^0)$. Note that the expectation \\
	$\E\{ \bx_{i[j]}^\T \bG_i^{1/2}(\bbeta^0) \barR^{-1} \beps_i(\bbeta^0) \} = 0$ due to the correct specification of the first-order moment of $Y_i$, and by Assumptions \ref{assump:bound_covs}, \ref{assump:var_fun} and \ref{assump:cor_mat}, $| \bx_{i[j]}^\T \bG_i^{1/2}(\bbeta^0) \barR^{-1} \beps_i(\bbeta^0) | \le C \| \beps_i(\bbeta^0) \|_2 $. By the definition of $\psi_q$ norm ($q = 1, 2$), 
	\[
	\| \bx_{i[j]}^\T \bG_i^{1/2}(\bbeta^0) \barR^{-1} \beps_i(\bbeta^0) \|_{\psi_q} \le C \| ~  \| \beps_i(\bbeta^0) \|_2 \|_{\psi_q} \le C \sum_{j=1}^m \| \varepsilon_{ij}(\bbeta^0) \|_{\psi_q} , ~ q=1, 2.
	\]
	So $\bx_{i[j]}^\T \bG_i^{1/2}(\bbeta^0) \barR^{-1} \beps_i(\bbeta^0)$ is sub-exponential (or sub-Gaussian) if $\varepsilon_{ij}(\bbeta^0)$'s are sub-exponential (or sub-Gaussian).
	
	{\textit{Case I. When $\varepsilon_{ij}(\bbeta^0)$'s are sub-exponential}}
	
	By Bernstein's inequality (see Theorem 2.8.2 of \citealt{vershynin2018high}), for any $t > 0$, 
	\[ 
	\begin{array}{rl}
		& \PP ( | \bar{\Psi}_j(\bbeta^0) | \ge t ) \\
		= & \PP ( | \sum_{i=1}^n \bx_{i[j]}^\T \bG_i^{1/2}(\bbeta^0) \barR^{-1} \beps_i(\bbeta^0) | \ge nt ) \\
		\le  &  2 \exp \left\{ -C \min \left( \displaystyle \frac{n^2 t^2}{\sum_{i=1}^n \| \bx_{i[j]}^\T \bG_i^{1/2}(\bbeta^0) \barR^{-1} \beps_i(\bbeta^0) \|_{\psi_1}^2 }, \frac{nt}{\max_i \| \bx_{i[j]}^\T \bG_i^{1/2}(\bbeta^0) \barR^{-1} \beps_i(\bbeta^0) \|_{\psi_1}} \right) \right\}. \\
	\end{array}
	\]
	Since $\| \bx_{i[j]}^\T \bG_i^{1/2}(\bbeta^0) \barR^{-1} \beps_i(\bbeta^0) \|_{\psi_1} \le M^{\prime\prime}$ for some $M^{\prime\prime} > 0$ and all $j= 1, \ldots, p$, when $t = C^{\prime} \{ \log(p)/n \}^{1/2}$ for some large  enough $C^{\prime}>0$, $\PP ( \| \barPsi(\bbeta^0) \|_{\infty} \ge t ) \le \sum_{j=1}^p \PP ( | \bar{\Psi}_j(\bbeta^0) | \ge t ) \le O(p^{-C}) $ as $n \to \infty$. Hence, we have $\| \barPsi(\bbeta^0) \|_{\infty} = \OP [ \{\log(p)/n\}^{1/2} ]$.
	
	{\textit{Case II. When $\varepsilon_{ij}(\bbeta^0)$'s are sub-Gaussian}}
	
	By Hoeffding's inequality (see Theorem 2.6.3 of \citealt{vershynin2018high}), for any $t > 0$,
	\[ 
	\begin{array}{rcl}
		\PP ( | \bar{\Psi}_j(\bbeta^0) | \ge t ) & = & \PP ( | \sum_{i=1}^n \bx_{i[j]}^\T \bG_i^{1/2}(\bbeta^0) \barR^{-1} \beps_i(\bbeta^0) | \ge nt ) \\
		& \le & 2 \exp\{ -C n t^2 / M_j^2 \} \\
		& \le & 2 \exp\{ -C n t^2 / (M^{\prime})^2 \},
	\end{array}
	\]
	where $M_j = \| \bx_{1[j]}^\T \bG_1^{1/2}(\bbeta^0) \barR^{-1} \beps_1(\bbeta^0) \|_{\psi_2} \le M^{\prime}$ for some $M^{\prime} > 0$. And then,
	\[
	\PP ( \| \barPsi(\bbeta^0) \|_{\infty} \ge t ) \le \sum_{j=1}^p \PP ( | \bar{\Psi}_j(\bbeta^0) | \ge t ) \le 2 p \exp\{ -C n t^2 / (M^{\prime})^2 \}.
	\]
	Let $t = C^{\prime} \{ \log(p)/n \}^{1/2}$ for some large enough constant $C^{\prime} > 0$, then $\PP ( \| \barPsi(\bbeta^0) \|_{\infty} \ge t ) \le O(p^{-C})$ for some $C>0$ and we have $\| \barPsi(\bbeta^0) \|_{\infty} = \OP [ \{ \log(p)/n \}^{1/2} ]$.
\end{proof}


\begin{lemma} \label{lemma:leading_normal}
	Assume $\lambda \asymp \{ \log(p) / n \}^{1/2}$, $\lambda^{\prime} \asymp \| \baromega^0 \|_1 (s_0 \lambda +  r_n)$, and further, $\max(s_0, s^{\prime} \| \baromega^0 \|_1) \lambda^{\prime} \{ \log(p) \}^{1/2} \to 0$. Under Assumptions \ref{assump:bound_covs}--\ref{assump:cor_mat}, 
	$
	{n}^{1/2} \Psi^P(\theta^0) / \{ (\bomega^0)^\T \bV^0 \bomega^0 \}^{1/2} 
	$
	converges to $N(0,1)$ in distribution as $n \to \infty$.
\end{lemma}


\begin{proof}[Proof of Lemma \ref{lemma:leading_normal}]
	
	For convenience, let $\tbeta = \hbeta + \omegahat (\theta^0 - \bxi^\T \hbeta) = \hbeta + \omegahat \bxi^\T (\bbeta^0 - \hbeta)$. Then ${n}^{1/2} \Psi^P(\theta^0) = {n}^{1/2} \omegahat^\T \bPsi(\tbeta)$.
	We decompose ${n}^{1/2} \omegahat^\T \bPsi(\tbeta)$ as 
	\[
	\begin{array}{rcl}
		{n}^{1/2} \Psi^P(\theta^0) & = & {n}^{1/2} \homega^\T \bPsi(\tbeta)  \\
		& = & {n}^{1/2} (\bomega^0)^\T \barPsi(\bbeta^0) + {n}^{1/2} (\bomega^0)^\T \{ \bPsi(\bbeta^0) - \barPsi(\bbeta^0) \} \\
		& & + \left[ {n}^{1/2} \{ \homega^\T \bPsi(\tbeta) - (\bomega^0)^\T \bPsi(\tbeta) \} + {n}^{1/2} (\bomega^0)^\T \{ \bPsi(\tbeta) - \bPsi(\bbeta^0) \} \right], 
	\end{array}
	\]
	and define $I_1 = {n}^{1/2} (\bomega^0)^\T \barPsi(\bbeta^0)$, $I_2 = {n}^{1/2} (\bomega^0)^\T \{ \bPsi(\bbeta^0) - \barPsi(\bbeta^0) \}$, and $I_3 =  {n}^{1/2} \{ \homega^\T \bPsi(\tbeta) - (\bomega^0)^\T \bPsi(\tbeta) \} + {n}^{1/2} (\bomega^0)^\T \{ \bPsi(\tbeta) - \bPsi(\bbeta^0) \}$.
	We next examine each of these three terms.\\[0.05cm]
	
	\noindent {\textit{Part (i). $I_1 = {n}^{1/2} (\bomega^0)^\T \barPsi(\bbeta^0)$} is asymptotically normal.}
	
	Noting that
	\[
	{n}^{1/2} (\bomega^0)^\T \barPsi(\bbeta^0) = \displaystyle n^{-1/2} \sum_{i=1}^n (\bomega^0)^\T \bX_i^\T \bG_i^{1/2}(\bbeta^0) \barR^{-1} \bG_i^{-1/2}(\bbeta^0) (\bY_i - \bmu_i(\bbeta^0)),
	\]
	where the standardized residuals $\beps_i(\bbeta^0)  = \bG_i^{-1/2}(\bbeta^0) (\bY_i - \bmu_i(\bbeta^0))$, we define $$\eta_i = n^{-1/2} (\bomega^0)^\T \bX_i^\T \bG_i^{1/2}(\bbeta^0) \barR^{-1} \beps_i(\bbeta^0),$$ $i = 1, \ldots, n$. It can be verified that
	\[
	\begin{array}{rcl}
		\E (\eta_i) & = & 0,  \\
		\Var(\eta_i) & = & \E \left\{ \displaystyle \frac{1}{n} (\bomega^0)^\T \bX_i^\T \bG_i^{1/2}(\bbeta^0) \barR^{-1} \bR_0 \barR^{-1} \bG_i^{1/2}(\bbeta^0) \bX_i \bomega^0 \right\} = \displaystyle \frac{1}{n} (\bomega^0)^\T \bV^0 \bomega^0.
	\end{array}
	\]
	Let the variance $\nu_n^2 = \Var(\sum_{i=1}^n \eta_i) = (\bomega^0)^\T \bV^0 \bomega^0$. By Assumptions \ref{assump:bound_covs}, \ref{assump:bdd_eigen_covs}, \ref{assump:var_fun} and \ref{assump:cor_mat}, the eigenvalues of $\bV^0$ and $\bS^0$ are all bounded from both below and above, and, given $\| \bxi \|_2=1$, we have
	\begin{equation} \label{eq:bdd_vn}
		\nu_n^2 \ge \displaystyle  \lambda_{\min}(\bV^0) \| \bomega^0 \|_2^2 =  \lambda_{\min}(\bV^0) \{ \bxi^\T (\bS^0)^{-1} \bxi \}^{-2} \| (\bS^0)^{-1} \bxi \|_2^2 \ge C > 0.
	\end{equation}
	We next verify that the Lindeberg condition holds, i.e.
	\begin{equation} \label{eq:lindeberg}
		\displaystyle \frac{1}{\nu_n^2} \sum_{i=1}^n \E [ \eta_i^2   1\{ |\eta_i| > \epsilon \nu_n \} ] = \frac{n}{\nu_n^2}  \E [ \eta_1^2   1\{ |\eta_1| > \epsilon \nu_n \} ] \to 0
	\end{equation}
	for any $\epsilon > 0$. 
	By Cauchy-Schwarz inequality and Assumptions \ref{assump:bdd_xomega}, \ref{assump:var_fun} and \ref{assump:cor_mat}, we have 
	\[ \begin{array}{l}
		\eta_1^2 \le \displaystyle n^{-1} \| \beps_1(\bbeta^0) \|_2^2 \cdot \| \barR^{-1} \bG_1^{1/2}(\bbeta^0) \bX_1 \bomega^0 \|_2^2 \\ 
		\qquad \qquad \le \displaystyle n^{-1} \| \beps_1(\bbeta^0) \|_2^2 \{ \|\barR^{-1}\| \times \|\bG_1^{1/2}(\bbeta^0) \| \times \| \bX_1 \bomega^0 \|_2 \}^2 \le C  \| \beps_1(\bbeta^0) \|_2^2 / n.
	\end{array}
	\]
	Since $\E \| \beps_1(\bbeta^0) \|_2^2 = \trace \E \{ \beps_1(\bbeta^0) \beps_1^\T(\bbeta^0) \} = {O}(1)$, $\| \beps_1(\bbeta^0) \|_2^2 = \OP(1)$.  Combining the above inequality and \eqref{eq:bdd_vn}, we have
	$
	| \eta_1 / \nu_n | = \OP(n^{-1/2}) ~ \mathrm{and} ~ 1\{ |\eta_1| > \epsilon \nu_n \} = \oP(1).
	$
	Also, $n \eta_1^2 / \nu_n^2  \le C^{\prime} \| \beps_1(\bbeta^0) \|_2^2$, which is integrable.
	By the dominated convergence theorem, the Lindeberg condition \eqref{eq:lindeberg} is satisfied, and the conclusion  holds that  the term $I_1$ is asymptotically normal in the sense that
	\[
	\displaystyle \frac{{n}^{1/2} (\bomega^0)^\T \barPsi(\bbeta^0)}{ \{ (\bomega^0)^\T \bV^0 \bomega^0 \}^{1/2} } = \frac{\sum_{i=1}^n \eta_i}{\nu_n}  
	\]
	converges to $N(0,1)$ in distribution as $n \to \infty$. \\[0.05cm]

	\noindent {\textit{Part (ii). $I_2 = {n}^{1/2} (\bomega^0)^\T \{ \bPsi(\bbeta^0) - \barPsi(\bbeta^0) \} = \oP(1)$.}}
	
	We write $\beps_i(\bbeta^0) = (\varepsilon_{i1}(\bbeta^0), \cdots, \varepsilon_{im}(\bbeta^0))^\T$ as in Lemma \ref{lemma:psi_infty} and $\bQ=(q_{j_1 j_2})_{m \times m}=\hatR^{-1} - \barR^{-1}$, where $q_{j_1 j_2}$ is the $(j_1, j_2)$th element of $\bQ = \hatR^{-1} - \barR^{-1}$. Then
	\[
	\begin{array}{rcl}
		I_2 & =  &  {n}^{1/2} (\bomega^0)^\T  \displaystyle \frac{1}{n} \sum_{i=1}^n \bX_i^\T \bG_i^{1/2}(\bbeta^0) \bQ \beps_i(\bbeta^0) \\
		& =  & \displaystyle n^{-1/2} \sum_{i=1}^n \sum_{j_1 = 1}^m \sum_{j_2 = 1}^m  q_{j_1 j_2} G_{i j_1}^{1/2}(\bbeta^0) \bx_{i j_1}^\T \bomega^0 \varepsilon_{i j_2}(\bbeta^0) \\
		& = & \displaystyle n^{-1/2}  \sum_{j_1 = 1}^m \sum_{j_2 = 1}^m  q_{j_1 j_2} \sum_{i=1}^n G_{i j_1}^{1/2}(\bbeta^0) \bx_{i j_1}^\T \bomega^0 \varepsilon_{i j_2}(\bbeta^0).
	\end{array}
	\]
	Since 
	$$\E \left\{ \left| \displaystyle \sum_{i=1}^n G_{i j_1}^{1/2}(\bbeta^0) \bx_{i j_1}^\T \bomega^0 \varepsilon_{i j_2}(\bbeta^0) \right|^2 \right\} = \sum_{i=1}^n \E \{ G_{i j_1}(\bbeta^0) (\bx_{i j_1}^\T \bomega^0)^2 \varepsilon_{i j_2}^2(\bbeta^0) \} = {O}(n),$$
	under Assumptions \ref{assump:bdd_xomega} and \ref{assump:var_fun}, we have $\sum_{i=1}^n G_{i j_1}^{1/2}(\bbeta^0) \bx_{i j_1}^\T \bomega^0 \varepsilon_{i j_2}(\bbeta^0) = \OP({n}^{1/2})$. Then, $$|I_2| \le \displaystyle \frac{1}{n^{1/2}} \sum_{j_1=1}^m \sum_{j_2=1}^m |q_{j_1 j_2}| \times \left| \sum_{i=1}^n G_{i j_1}^{1/2}(\bbeta^0) \bx_{i j_1}^\T \bomega^0 \varepsilon_{i j_2}(\bbeta^0) \right| = \OP(\| \hatR^{-1} - \barR^{-1} \|) = \oP(1).$$

	\noindent {\textit{Part (iii). $I_3 = {n}^{1/2} \{ \homega - (\bomega^0) \}^\T \bPsi(\tbeta) + {n}^{1/2} (\bomega^0)^\T \{ \bPsi(\tbeta) - \bPsi(\bbeta^0) \} = \oP(1)$}.}

	By Taylor expansion,
	\[
	\bPsi(\tbeta) =  \bPsi(\bbeta^0) + \displaystyle \left. \frac{\partial \bPsi}{\partial \bbeta^\T } \right|_{\bbeta^*} (\tbeta - \bbeta^0),
	\]
	where $\bbeta^{*}$ is an intermediate vector between $\tilde{\bbeta}$ and $\bbeta^0$. Then
	\[
	\begin{array}{rcl}
		I_3 & = & {n}^{1/2} \{ \homega - (\bomega^0) \}^\T \bPsi(\tbeta) + {n}^{1/2} (\bomega^0)^\T \{ \bPsi(\tbeta) - \bPsi(\bbeta^0) \}  \\
		& = & {n}^{1/2}  \{ \homega - (\bomega^0) \}^\T \left\{ \bPsi(\bbeta^0) + \displaystyle \left. \frac{\partial \bPsi}{\partial \bbeta^\T } \right|_{\bbeta^*} (\tbeta - \bbeta^0) \right\} + {n}^{1/2} (\bomega^0)^\T \displaystyle \left. \frac{\partial \bPsi}{\partial \bbeta^\T } \right|_{\bbeta^*} (\tbeta - \bbeta^0) \\
		& = & {n}^{1/2} \{ \homega - (\bomega^0) \}^\T \bPsi(\bbeta^0) + {n}^{1/2} \homega^\T  \displaystyle \left. \frac{\partial \bPsi}{\partial \bbeta^\T } \right|_{\bbeta^*} (\tbeta - \bbeta^0) \\
		& = & I_{31} + I_{32}.
	\end{array}
	\]
	To show $I_{31} = \oP(1)$, first note that
	\[
	\begin{array}{rl}
		& \| \bPsi(\bbeta^0) - \barPsi(\bbeta^0) \|_{\infty} \\
		= & \displaystyle \left\|   \frac{1}{n} \sum_{i=1}^n \bX_i^\T \bG_i^{1/2}(\bbeta^0) \bQ \beps_i(\bbeta^0) \right\|_{\infty} \\
		= & \displaystyle \max_{1 \le k \le p} \left|  \frac{1}{n} \sum_{i=1}^n \sum_{j_1=1}^m \sum_{j_2=1}^m q_{j_1 j_2} x_{i[k]j_1} G_{i j_1}^{1/2}(\bbeta^0) \varepsilon_{i j_2}(\bbeta^0) \right|  \\
		\le & \displaystyle \max_{1 \le k \le p} \sum_{j_1=1}^m \sum_{j_2=1}^m |q_{j_1 j_2}| \times \frac{1}{n} \left| \sum_{i=1}^n x_{i[k]j_1} G_{i j_1}^{1/2}(\bbeta^0) \varepsilon_{i j_2}(\bbeta^0) \right|,
	\end{array}
	\]
	where $x_{i[k]j_1}$ is the $j_1$th element of $\bx_{i[k]}$. Since $\E | \sum_{i=1}^n  x_{i[k]j_1} G_{i j_1}^{1/2}(\bbeta^0) \varepsilon_{i j_2}(\bbeta^0) |^2 = {O}(n)$ uniformly for all $k, j_1$ and $j_2$, $|  \sum_{i=1}^n x_{i[k]j_1} G_{i j_1}^{1/2}(\bbeta^0) \varepsilon_{i j_2}(\bbeta^0) | = \OP({n}^{1/2})$. Hence,
	\[
	\| \bPsi(\bbeta^0) - \barPsi(\bbeta^0) \|_{\infty} = \OP(n^{-1/2} \| \bQ \|) = \OP(n^{-1/2} \| \hatR^{-1} - \barR^{-1} \|) \le \OP(n^{-1/2} r_n),
	\]
	where the last inequality is due to the fact that $\| \hatR^{-1} - \barR^{-1} \| \le \| \hatR^{-1} \| \times \| \hatR - \barR \| \times \| \barR^{-1} \| \le \OP(1) \| \hatR - \barR \| = \OP(r_n)$.
	Then, since $\| \bPsi(\bbeta^0) - \barPsi(\bbeta^0) \|_{\infty} = \OP(n^{-1/2} r_n)$,
	\[
	\begin{array}{rcl}
		| I_{31} | & \le & {n}^{1/2} \| \homega - \bomega^0 \|_1  \| \bPsi(\bbeta^0) \|_{\infty}  \\
		& \le & {n}^{1/2} \| \homega - \bomega^0 \|_1 \cdot \{ \| \barPsi(\bbeta^0) \|_{\infty} + \| \bPsi(\bbeta^0) - \barPsi(\bbeta^0) \|_{\infty}  \} \\
		& = & {n}^{1/2} \OP(s^{\prime} \lambda^{\prime} \| \baromega^0 \|_1) \left\{ \OP [ \{ \log(p)/n \}^{1/2} ] + \OP(n^{-1/2} r_n) \right\} \\
		& = & \OP \left(s^{\prime} \lambda^{\prime} \| \baromega^0 \|_1 \{ \log(p) \}^{1/2} \right) =  \oP(1).
	\end{array}
	\]

	Next, according to Lemma \ref{lemma:deriv_psi}, $I_{32}$ can be decomposed as 
	\[ 
	\begin{array}{rcl}
		I_{32} & = &  - {n}^{1/2} \homega^\T \bS(\bbeta^*)(\tbeta - \bbeta^0) + {n}^{1/2} \homega^\T \bE_n(\bbeta^*)(\tbeta - \bbeta^0) + {n}^{1/2} \homega^\T \bF_n(\bbeta^*)(\tbeta - \bbeta^0) \\
		& = & I_{321} + I_{322} + I_{323}. 
	\end{array}
	\]
	By the definition of $\tbeta$, $\tbeta - \bbeta^0 = (\bI - \homega \bxi^\T) (\hbeta - \bbeta^0)$. Thus,
	\begin{align} 
		- I_{321}  & =  ~ {n}^{1/2} \homega^\T \bS(\bbeta^*)(\tbeta - \bbeta^0) \notag  \\
		&  =   ~ {n}^{1/2} \homega^\T \bS(\bbeta^*) (\bI - \homega \bxi^\T) (\hbeta - \bbeta^0) \notag \\
		& = ~ {n}^{1/2} \left\{ \bS(\bbeta^*) \homega - \homega^\T \bS(\bbeta^*) \homega \bxi \right\}^\T  (\hbeta - \bbeta^0) \notag \\
		& = ~ {n}^{1/2} \left[ \bS(\hbeta) \homega - \homega^\T \bS(\hbeta) \homega \bxi + \left\{ \bS(\bbeta^*) - \bS(\hbeta) \right\} \homega + \homega^\T \left\{ \bS(\hbeta) - \bS(\bbeta^*) \right\} \homega \bxi  \right]^\T (\hbeta - \bbeta^0) \notag \\
		& = ~ {n}^{1/2} \left\{ \bS(\hbeta) \homega - \homega^\T \bS(\hbeta) \homega \bxi \right\}^\T (\hbeta - \bbeta^0) + {n}^{1/2} \homega^\T \left\{ \bS(\bbeta^*) - \bS(\hbeta) \right\}^\T (\hbeta - \bbeta^0) \notag \\ 
		& \qquad + {n}^{1/2} \homega^\T \left\{ \bS(\hbeta) - \bS(\bbeta^*) \right\} \homega \bxi^\T (\hbeta - \bbeta^0). \label{eq:decomp_I321}
	\end{align}
	By the definition of $\homega$ in Lemma \ref{lemma:hat_omega_rate}, the first term in \eqref{eq:decomp_I321} satisfies
	\[
	\begin{array}{rcl}
		\left| {n}^{1/2} \{ \bS(\hbeta) \homega - \homega^\T \bS(\hbeta) \homega \bxi \}^\T (\hbeta - \bbeta^0) \right|  & = & \left|  \displaystyle \frac{{n}^{1/2}}{\tilde{\bomega}^\T \bS(\hbeta) \tilde{\bomega}} \{ \bS(\hbeta) \tilde{\bomega} - \bxi \}^\T (\hbeta - \bbeta^0) \right|  \\
		& \le   &   \displaystyle \frac{{n}^{1/2}}{\tilde{\bomega}^\T \bS(\hbeta) \tilde{\bomega}} \| \bS(\hbeta) \tilde{\bomega} - \bxi \|_{\infty} \| \hbeta - \bbeta^0 \|_{1}  \\
		& = & \OP({n}^{1/2} \lambda^{\prime} s_0 \lambda) = \oP(1), \\
	\end{array}
	\]
	where in the last equality, $\lambda \asymp \{ \log(p)/n \}^{1/2}$ and $\max(s_0, s^{\prime} \| \baromega^0 \|_1) \lambda^{\prime} \{\log(p)\}^{1/2} \to 0$ imply that $s_0 \lambda^{\prime} \{\log(p)\}^{1/2} \to 0$.
	For the second term in \eqref{eq:decomp_I321}, we have
	\begin{align} \label{eq:decomp_I321_second}
		& | {n}^{1/2} \homega^\T \{ \bS(\bbeta^*) - \bS(\hbeta) \}^\T (\hbeta - \bbeta^0) | \notag \\   
		\le & {n}^{1/2} \left| \displaystyle \frac{1}{n} \sum_{i=1}^n \homega^\T \bX_i^\T \{ \bG_i^{1/2}(\bbeta^*) - \bG_i^{1/2}(\hbeta) \} \hatR^{-1} \bG_i^{1/2}(\bbeta^*) \bX_i(\hbeta - \bbeta^0) \right|  \notag \\
		& + {n}^{1/2} \left| \displaystyle \frac{1}{n} \sum_{i=1}^n \homega^\T \bX_i^\T \bG_i^{1/2}(\hbeta) \hatR^{-1} \{ \bG_i^{1/2}(\bbeta^*) - \bG_i^{1/2}(\hbeta) \}  \bX_i(\hbeta - \bbeta^0) \right|.
	\end{align}
	By Lemma \ref{lemma:hat_omega_rate} and Assumption \ref{assump:bdd_xomega}, $\| \bX_i \homega \|_{\infty} \le \| \bX_i \|_{\infty} \| \homega - \bomega^0 \|_1 + \| \bX_i \bomega^0 \|_{\infty} = \OP(1)$.
	Since $\bbeta^*$ is an intermediate value between $\tbeta$ and $\bbeta^0$, $|\bx_{ij}^\T (\bbeta^* - \hbeta)| \le \max( |\bx_{ij}^\T(\tbeta - \hbeta)|, |\bx_{ij}^\T(\bbeta^0 - \hbeta)| )$. Because $|\bx_{ij}^\T(\tbeta - \hbeta)| = | \bx_{ij}^\T \homega \bxi^\T (\bbeta^0 - \hbeta) | \le | \bx_{ij}^\T \homega | \| \bxi \|_{2} \| \hbeta - \bbeta^0 \|_2 = \OP(\| \hbeta - \bbeta^0 \|_2)$, we have $| \bx_{ij}^\T (\bbeta^* - \hbeta) | = \OP[ \max \{ \| \hbeta - \bbeta^0 \|_2, |\bx_{ij}^\T (\bbeta^0 - \hbeta)| \} ]$.
	Moreover, since $G_{ij}(\bbeta) = v (\mu_{ij}(\bbeta)) = v(\mu(\bx_{ij}^\T \bbeta))$ is continuously differentiable in $\bx_{ij}^\T \bbeta$, by Assumption \ref{assump:var_fun}, $| G_{ij}^{1/2}(\bbeta^*) - G_{ij}^{1/2}(\hbeta) | = \OP( | \bx_{ij}^\T (\bbeta^* - \hbeta) |)$. Now, writing $\hatR^{-1}_{\cdot j}$ as the $j$th column of $\hatR^{-1}$, since  $\| \bX_i \homega \|_{\infty} = \OP(1)$ and $\| \hatR^{-1} \| = \OP(1)$,
	\begin{align}
		& \left[ G_{ij}^{1/2}(\bbeta^*) \homega^\T \bX_i^\T \{ \bG_i^{1/2}(\bbeta^*) - \bG_i^{1/2}(\hbeta) \} \hatR^{-1}_{\cdot j} \right]^2 \notag \\ 
		\le &  G_{ij}(\bbeta^*) \| \bX_i \homega \|_2^2 \| \{ \bG_i^{1/2}(\bbeta^*) - \bG_i^{1/2}(\hbeta) \} \hatR^{-1}_{\cdot j} \|_2^2 \notag \\
		\le & G_{ij}(\bbeta^*) \| \bX_i \homega \|_2^2 \sum_{k=1}^m \left\{ G_{ik}^{1/2}(\bbeta^*) - G_{ik}^{1/2}(\hbeta) \right\}^2 (\hatR^{-1})_{kj}^2 \notag \\
		= & \OP \left( \sum_{k=1}^m | \bx_{ik}^\T(\bbeta^* - \hbeta) |^2 \right) \notag \\
		\le & \OP \left[ \sum_{k=1}^m \left\{ \| \hbeta - \bbeta^0 \|_2^2 + | \bx_{ik}^\T (\bbeta^0 - \hbeta) |^2 \right\} \right]. \label{eq:tmp_linear_pred}   
	\end{align}
	Thus, by Cauchy-Schwarz inequality, for the first term in \eqref{eq:decomp_I321_second}, we have
	\[
	\begin{array}{rl}
		& {n}^{1/2} \left| \displaystyle \frac{1}{n} \sum_{i=1}^n \homega^\T \bX_i^\T \{ \bG_i^{1/2}(\bbeta^*) - \bG_i^{1/2}(\hbeta) \} \hatR^{-1} \bG_i^{1/2}(\bbeta^*) \bX_i(\hbeta - \bbeta^0) \right|  \\
		\le & \displaystyle {n}^{1/2}  \left[ \sum_{i=1}^n \sum_{j=1}^m  \frac{1}{n} \left\{ \homega^\T \bX_i^\T \{ \bG_i^{1/2}(\bbeta^*) - \bG_i^{1/2}(\hbeta) \} G_{ij}^{1/2}(\bbeta^*) \hatR^{-1}_{\cdot j} \right\}^2  \right]^{1/2} \times \\
		& \left[ \displaystyle \sum_{i=1}^n \sum_{j=1}^m \frac{1}{n} \{ \bx_{ij}^\T (\hbeta - \bbeta^0) \}^2  \right]^{1/2} \\
		= & {n}^{1/2} \OP \left( m^2 \| \hbeta - \bbeta^0 \|_2^2 + m \displaystyle \sum_{i=1}^n \sum_{j=1}^m \frac{1}{n} \{ \bx_{ij}^\T (\hbeta - \bbeta^0) \}^2  \right)^{1/2} \times \OP(s_0 \lambda^2)^{1/2} \\
		= & {n}^{1/2} \OP(s_0 \lambda^2)^{1/2} \OP(s_0 \lambda^2)^{1/2} \\
		= & \OP \left( s_0 \log(p) / {n}^{1/2}  \right) = \oP(1), \\
	\end{array}
	\]
	where in the last equality, $\max(s_0, s^{\prime} \| \baromega^0 \|_1) \lambda^{\prime} \{ \log(p) \}^{1/2} \to 0$ in the assumption of this lemma implies $s_0 \lambda \{\log(p)\}^{1/2} \asymp s_0 \log(p) / {n}^{1/2} \to 0$ with $\lambda \asymp \{ \log(p)/n \}^{1/2}$.
	Similarly, 
	\[  
	{n}^{1/2} \left| \displaystyle \frac{1}{n} \sum_{i=1}^n \homega^\T \bX_i^\T \bG_i^{1/2}(\hbeta) \hatR^{-1} \{ \bG_i^{1/2}(\bbeta^*) - \bG_i^{1/2}(\hbeta) \}  \bX_i(\hbeta - \bbeta^0) \right| = \oP(1).
	\]
	Hence, with the inequality of \eqref{eq:decomp_I321_second}, the second term in \eqref{eq:decomp_I321}, i.e. ${n}^{1/2} \homega^\T \{ \bS(\bbeta^*) - \bS(\hbeta) \} (\hbeta - \bbeta^0) = \oP(1)$.
	
	For the third term in \eqref{eq:decomp_I321}, we have
	\[
	\begin{array}{rl}
		& {n}^{1/2}  | \homega^\T \{ \bS(\hbeta) - \bS(\bbeta^*) \} \homega \bxi^\T (\hbeta - \bbeta^0) | \\
		= & {n}^{1/2} \displaystyle \left|  \frac{1}{n} \sum_{i=1}^n \homega^\T \bX_i^\T \{ \bG_{i}^{1/2}(\bbeta^*) - \bG_i^{1/2}(\hbeta) \} \hatR^{-1} \bG_i^{1/2}(\bbeta^*) \bX_i \homega \bxi^\T (\hbeta - \bbeta^0) \right|  \\
		& + {n}^{1/2} \displaystyle \left|  \frac{1}{n} \sum_{i=1}^n \homega^\T \bX_i^\T \bG_i^{1/2}(\hbeta) \hatR^{-1} \{ \bG_{i}^{1/2}(\bbeta^*) - \bG_i^{1/2}(\hbeta) \} \bX_i \homega \bxi^\T (\hbeta - \bbeta^0) \right|. \\
	\end{array}
	\]
	With an argument similar to \eqref{eq:tmp_linear_pred} for showing the second term in \eqref{eq:decomp_I321}, i.e.  ${n}^{1/2} \homega^\T \{ \bS(\bbeta^*) - \bS(\hbeta)  \}^\T (\hbeta - \bbeta^0)$, is $\oP(1)$, we have
	\begin{equation} \label{eq:third_part_I321}
		\begin{array}{rl}
			& {n}^{1/2} \displaystyle \left|  \frac{1}{n} \sum_{i=1}^n \homega^\T \bX_i^\T \{ \bG_{i}^{1/2}(\bbeta^*) - \bG_i^{1/2}(\hbeta) \} \hatR^{-1} \bG_i^{1/2}(\bbeta^*) \bX_i \homega \bxi^\T (\hbeta - \bbeta^0) \right|  \\
			\le & \displaystyle {n}^{1/2} \left[ \sum_{i=1}^n \sum_{j=1}^m  \frac{1}{n} \left\{ \homega^\T \bX_i^\T \{ \bG_i^{1/2}(\bbeta^*) - \bG_i^{1/2}(\hbeta) \} G_{ij}^{1/2}(\bbeta^*) \hatR^{-1}_{\cdot j} \right\}^2  \right]^{1/2} \times \\
			& \left[ \displaystyle \sum_{i=1}^n \sum_{j=1}^m \frac{1}{n} \{ \bx_{ij}^\T \homega \bxi^\T (\hbeta - \bbeta^0) \}^2 \right]^{1/2} \\
			\le & {n}^{1/2} \OP \left( m^2 \| \hbeta - \bbeta^0 \|_2^2 + m \displaystyle \sum_{i=1}^n \sum_{j=1}^m \frac{1}{n} \{ \bx_{ij}^\T (\hbeta - \bbeta^0) \}^2  \right)^{1/2} \cdot \OP \left(\| \hbeta - \bbeta^0 \|_2 \right) \\
			= & {n}^{1/2} \OP(s_0^{1/2} \lambda) \cdot \OP(s_0^{1/2} \lambda) \\
			= & \OP(s_0 \log(p) / {n}^{1/2}) = \oP(1),
		\end{array}
	\end{equation}
	where the second inequality above results from $| \bx_{ij}^\T \homega \bxi^\T (\hbeta - \bbeta^0) | \le | \bx_{ij}^\T \homega |  \| \bxi \|_2 \| \hbeta - \bbeta^0 \|_2 = \OP(\| \hbeta - \bbeta^0 \|_2)$.
	Also, we can use the same techniques to show that
	\[
	{n}^{1/2} \displaystyle \left|  \frac{1}{n} \sum_{i=1}^n \homega^\T \bX_i^\T \bG_i^{1/2}(\hbeta) \hatR^{-1} \{ \bG_{i}^{1/2}(\bbeta^*) - \bG_i^{1/2}(\hbeta) \} \bX_i \homega \bxi^\T (\hbeta - \bbeta^0) \right| = \oP(1).
	\]
	Hence, ${n}^{1/2}  | \homega^\T \{ \bS(\hbeta) - \bS(\bbeta^*) \} \homega \bxi^\T (\hbeta - \bbeta^0) | = \oP(1)$. This completes the proof for the statement that $I_{321} = \oP(1)$.
	
	{Next, we show that $I_{322} = {n}^{1/2} \homega^\T \bE_n(\bbeta^*)(\tbeta - \bbeta^0)  = \oP(1)$. For technical reasons, here we assume that $\hbeta$ and $\hat{\bR}$ are estimated using data that are independent of $\{\bY_i, \bX_i \}_{i=1}^n$. This may be achieved by data splitting, i.e. suppose we have $2n$ observations, and the half indexed by $i = n+1, \ldots, 2n$ are used to obtain the estimates $\hbeta$ and $\hat{\bR}^{-1}$, and the half indexed by $i = 1, \ldots, n$ are used as the observed data when constructing the projected estimating equation $\Psi^P(\theta)$ for inference.}
	
	We introduce the decomposition of $\bE_n(\bbeta)$ into four terms as follows. For any $\bbeta \in \mR^p$,
	\begin{align} \label{eq:decomp_En}
		\bE_n(\bbeta) & =  - \displaystyle \frac{1}{2n} \sum_{i=1}^n \sum_{j=1}^m \dot{v}(\mu_{ij}(\bbeta)) G_{ij}^{-1/2}(\bbeta) (y_{ij} - \mu_{ij}(\bbeta)) \bX_i^\T \bG_i^{1/2}(\bbeta) \hatR^{-1} \bare_j \bare_j^\T \bX_i  \nonumber \\
		& =  - \displaystyle \frac{1}{2n} \sum_{i=1}^n \sum_{j=1}^m \dot{v}(\mu_{ij}(\bbeta^0)) G_{ij}^{-1/2}(\bbeta^0) (y_{ij} - \mu_{ij}(\bbeta^0)) \bX_i^\T \bG_i^{1/2}(\bbeta^0) \hatR^{-1} \bare_j \bare_j^\T \bX_i  \nonumber \\
		&  - \displaystyle \frac{1}{2n} \sum_{i=1}^n \sum_{j=1}^m \dot{v}(\mu_{ij}(\bbeta^0)) G_{ij}^{-1/2}(\bbeta^0) (y_{ij} - \mu_{ij}(\bbeta^0)) \bX_i^\T \{ \bG_i^{1/2}(\bbeta) - \bG_i^{1/2}(\bbeta^0) \} \hatR^{-1} \bare_j \bare_j^\T \bX_i  \nonumber \\
		&  - \displaystyle \frac{1}{2n} \sum_{i=1}^n \sum_{j=1}^m \{ \dot{v}(\mu_{ij}(\bbeta)) G_{ij}^{-1/2}(\bbeta) - \dot{v}(\mu_{ij}(\bbeta^0)) G_{ij}^{-1/2}(\bbeta^0) \} (y_{ij} - \mu_{ij}(\bbeta^0)) \bX_i^\T \bG_i^{1/2}(\bbeta) \nonumber \\
		&  \qquad \qquad \times \hatR^{-1} \bare_j \bare_j^\T \bX_i \nonumber \\
		&  - \displaystyle \frac{1}{2n} \sum_{i=1}^n \sum_{j=1}^m \dot{v}(\mu_{ij}(\bbeta)) G_{ij}^{-1/2}(\bbeta) \{ \mu_{ij}(\bbeta^0) - \mu_{ij}(\bbeta) \} \bX_j^\T \bG_i^{1/2}(\bbeta) \hatR^{-1} \bare_j \bare_j^\T \bX_i \nonumber \\
		& =  \bE_{n1}(\bbeta^0) + \bE_{n2}(\bbeta) + \bE_{n3}(\bbeta) + \bE_{n4}(\bbeta). 
	\end{align}
	Then, in order to show $I_{322} = \oP(1)$, we will continue to show that ${n}^{1/2} \homega^\T \bE_{n1}(\bbeta^0)(\tilde{\bbeta} - \bbeta^0) = \oP(1)$ in Goal (a) below, and ${n}^{1/2} \homega^\T \bE_{nk}(\bbeta^*)(\tilde{\bbeta} - \bbeta^0) = \oP(1), ~ k=2,3,4$ in Goal (b) below.
	Recall that $\bbeta^*$ lies between $\tilde{\bbeta}$ and $\bbeta^0$, and  $\tilde{\bbeta} - \bbeta^0 = (\bI - \homega \bxi^\T) (\hbeta - \bbeta^0)$, and then
	\begin{equation} \label{eq:diff_xbeta_bar}
		\begin{array}{rcl}
			| \bx_{ij}^\T (\tilde{\bbeta} - \bbeta^0) | & = & | (\bx_{ij}^\T - \bx_{ij}^\T \homega \bxi^\T) (\hbeta - \bbeta^0) |  \\
			& \le & \| \bx_{ij} - \bx_{ij}^\T \homega \bxi \|_{\infty} \| \hbeta - \bbeta^0 \|_1 \\
			& \le & ( \| \bx_{ij} \|_{\infty} + |\bx_{ij}^\T \homega| \| \bxi \|_{\infty} ) \| \hbeta - \bbeta^0 \|_1 \\
			& \le & (K + \OP(1) \| \bxi \|_{\infty}) \| \hbeta - \bbeta^0 \|_1 \\
			& = & \OP(\| \hbeta - \bbeta^0 \|_1).
		\end{array}
	\end{equation}

	\medskip
	{
		\noindent {\textit{Goal (a): To show ${n}^{1/2} \homega^\T \bE_{n1}(\bbeta^0) (\tilde{\bbeta} - \bbeta^0) = \oP(1)$.}}
	}
	
	We first consider ${n}^{1/2} \homega^\T \bE_{n1}(\bbeta^0)(\tilde{\bbeta} - \bbeta^0)$ and rewrite
	\begin{equation} \label{eq:normal_decomp_En1_op1}
		\begin{array}{rl}
			& {n}^{1/2} \homega^\T \bE_{n1}(\bbeta^0) (\tilde{\bbeta} - \bbeta^0) \\
			= & - \displaystyle \frac{1}{2} \sum_{j=1}^m \sum_{i=1}^n \displaystyle n^{-1/2} \dot{v}(\mu_{ij}(\bbeta^0)) \homega^\T \bX_i^\T \bG_i^{1/2}(\bbeta^0) \hatR^{-1} \bare_j \bare_j^\T \bX_i (\tilde{\bbeta} - \bbeta^0) \varepsilon_{ij}(\bbeta^0),
		\end{array}
	\end{equation}
	where $\varepsilon_{ij}(\bbeta^0) = G_{ij}^{-1/2}(\bbeta^0) (y_{ij} - \mu_{ij}(\bbeta^0))$. Let $$h_{ij} = h_{ij} (\hbeta, \homega, \hatR, \bX_i) = \dot{v}(\mu_{ij}(\bbeta^0)) \homega^\T \bX_i^\T \bG_i^{1/2}(\bbeta^0) \hatR^{-1} \bare_j \bare_j^\T \bX_i (\tilde{\bbeta} - \bbeta^0),$$ and note that conditional on $(\hbeta, \homega, \hatR, \bX_i)$, $h_{ij}$ can be viewed as a constant. We observe that  the estimation for $\homega$ only depends on $\{\bX_i\}_{i=1}^n$ and $\hbeta$ through $\bS(\hbeta)$, so the conditional statements in the following Case I and Case II given $\hbeta, \homega, \hatR$ and $\{ \bX_i \}_{i=1}^n$ are equivalent to being conditional on $\hbeta, \hatR$ and $\{ \bX_i \}_{i=1}^n$. For explicit expression, we keep $\homega$ in the conditional statements.

	{\textit{Case I. When $\varepsilon_{ij}(\bbeta^0)$'s are sub-exponential}}
	
	Given $(\hbeta, \homega, \hatR, \{  \bX_i \}_{i=1}^n)$, by Bernstein's inequality (see Theorem 2.8.2 of \citet{vershynin2018high}), we have the following concentration inequality for $j= 1, \cdots, m$,
	\begin{align*}
		& \PP \left( \left| \sum_{i=1}^n \displaystyle n^{-1/2} h_{ij} \varepsilon_{ij}(\bbeta^0) \right| 
		\ge t  \mid  \hbeta, \homega, \hatR, \{  \bX_i \}_{i=1}^n  \right) \\
		\le & 2 \exp \left\{ - \displaystyle C \cdot \min \left( \frac{n t^2}{(L_j^{\prime})^2 \sum_{i=1}^n h_{ij}^2}, \frac{{n}^{1/2} t}{L_j^{\prime} \max_i |h_{ij}|} \right) \right\},
	\end{align*}
	where $C>0$ is some constant and $L_j^{\prime} = \| \varepsilon_{1j}(\bbeta^0) \|_{\psi_1}$ is finite. Because $| \bx_{ij}^\T \homega |  = \OP(1), \| \hatR^{-1} \| = \OP(1)$, and $\| \hbeta - \bbeta^0 \|_1 = \OP(s_0 \lambda)$. For any $\delta > 0$, there exists a constant $M > 0$ such that
	\[
	\PP \left[ \left\{ \| \bX_i \homega \|_2 \le M \right\} \cap \left\{ \| \hatR^{-1} \| \le M \right\} \cap \left\{ \| \hbeta - \bbeta^0 \|_1 / (s_0 \lambda) \le M \right\} \right] \ge 1 - \delta/4.
	\] 
	On the event $\mathcal{J} = \{ \| \bX_i \homega \|_2 \le M \} \cap \{ \| \hatR^{-1} \| \le M \} \cap \{ \| \hbeta - \bbeta^0 \|_1 / (s_0 \lambda) \le M \}$, due to Assumption \ref{assump:var_fun} and the fact that
	\begin{equation} \label{eq:ineq_h_ij}
		\begin{array}{l}
			|h_{ij}| \le C \| \bX_i \homega \|_2 \| \hatR^{-1} \| | \bx_{ij}^\T  (\tilde{\bbeta} - \bbeta^0) | \\
			\qquad \qquad \le C \| \bX_i \homega \|_2 \| \hatR^{-1} \| \cdot \left( \| \bx_{ij} \|_{\infty} + |\bx_{ij}^\T \homega|  \| \bxi \|_{\infty} \right) \| \hbeta - \bbeta^0 \|_1,
		\end{array}
	\end{equation}
	the two terms within the minimum function satisfy
	\[
	\begin{array}{rcl}
		\displaystyle \frac{n t^2}{(L_j^{\prime})^2 \sum_{i=1}^n h_{ij}^2} & \gtrsim & \displaystyle \frac{t^2}{ (L_j^{\prime})^2 M^6 (K + M \| \bxi \|_{\infty})^2 (s_0 \lambda)^2 }, \\
		& & \\
		\displaystyle \frac{{n}^{1/2} t}{L_j^{\prime} \max_i |h_{ij}|} & \gtrsim  & \displaystyle \frac{{n}^{1/2} t}{L_j^{\prime} M^3 (K + M \| \bxi \|_{\infty}) s_0 \lambda}.
	\end{array}
	\]
	Taking $t = C^{\prime} s_0 \lambda$ with a large enough constant $C^{\prime} > 0$, we have that on the event $\mathcal{J}$, when $n$ is large enough,
	\[
	\exp \left\{ - \displaystyle \frac{C n t^2}{(L_j^{\prime})^2 \sum_{i=1}^n h_{ij}^2}  \right\}  \le  \delta/4, \quad \mathrm{and} \quad
	\exp \left\{  - \displaystyle \frac{C {n}^{1/2} t}{L_j^{\prime} \max_i |h_{ij}|} \right\}  \le  \delta/4.
	\]
	Then, 
	\[
	\begin{array}{rl}
		& \PP \left( \left| \displaystyle \sum_{i=1}^n  n^{-1/2} h_{ij} \varepsilon_{ij}(\bbeta^0) \right| \ge t  \right) \\ 
		\le & 2 \E \exp \left\{ - \displaystyle C \cdot \min \left( \frac{n t^2}{(L_j^{\prime})^2 \sum_{i=1}^n h_{ij}^2}, \frac{{n}^{1/2} t}{L_j^{\prime} \max_i |h_{ij}|} \right) \right\}  \\
		= & 2 \E \left[ \exp \left\{ - \displaystyle C \cdot \min \left( \frac{n t^2}{(L_j^{\prime})^2 \sum_{i=1}^n h_{ij}^2}, \frac{{n}^{1/2} t}{L_j^{\prime} \max_i |h_{ij}|} \right) \right\} \cdot 1_{\mathcal{J}} \right] \\
		& \qquad + 2 \E \left[ \exp \left\{ - \displaystyle C \cdot \min \left( \frac{n t^2}{(L_j^{\prime})^2 \sum_{i=1}^n h_{ij}^2}, \frac{{n}^{1/2} t}{L_j^{\prime} \max_i |h_{ij}|} \right) \right\} \cdot 1_{\mathcal{J}^c} \right] \\
		\le & 2 \times \delta/4 \cdot \PP(\mathcal{J}) + 2 \PP(\mathcal{J}^c) \\
		\le & \delta.
	\end{array}
	\]
	Hence, $\sum_{i=1}^n \displaystyle n^{-1/2} h_{ij} \varepsilon_{ij}(\bbeta^0) = \OP(s_0 \lambda) = \oP(1)$. And since $m$ is finite, by \eqref{eq:normal_decomp_En1_op1}, we have ${n}^{1/2} \homega^\T \bE_{n1}(\bbeta^0)(\tilde{\bbeta} - \bbeta^0) = \oP(1)$. \\[0.05cm]

	{\textit{Case II. When $\varepsilon_{ij}(\bbeta^0)$'s are sub-Gaussian}}
	
	This is a special case of $\varepsilon_{ij}(\bbeta^0)$'s being sub-exponential, and since sub-Gaussianality is also very common, we briefly discuss the proof under this scenario. Given $(\hbeta, \homega, \hatR, \{  \bX_i \}_{i=1}^n)$, by the general Hoeffding's inequality (see Theorem 2.6.3 of \citet{vershynin2018high}), we have the following concentration inequality for $j= 1, \ldots, m$,
	\[
	\PP \left( \left| \sum_{i=1}^n \displaystyle n^{-1/2} h_{ij} \varepsilon_{ij}(\bbeta^0) \right| \ge t \mid \hbeta, \homega, \hatR, \{  \bX_i \}_{i=1}^n \right) \le 2 \exp\left( - \displaystyle \frac{Cnt^2}{L_j^2 \sum_{i=1}^n h_{ij}^2}  \right),
	\]
	where $C>0$ is some constant and $L_j = \| \varepsilon_{1j}(\bbeta^0) \|_{\psi_2}$ is finite. 
	On the event $\mathcal{J} = \{ \| \bX_i \homega \|_2 \le M \} \cap \{ \| \hatR^{-1} \| \le M \} \cap \{ \| \hbeta - \bbeta^0 \|_1 / (s_0 \lambda) \le M \}$, according to \eqref{eq:ineq_h_ij} and Assumption \ref{assump:var_fun},
	\[
	2 \exp\left( - \displaystyle \frac{Cnt^2}{L_j^2 \sum_{i=1}^n h_{ij}^2}  \right) \le 2 \exp \left( - \displaystyle \frac{C'nt^2}{L_j^2 n M^6 (K + M \| \bxi \|_{\infty})^2 (s_0 \lambda)^2} \right).
	\]
	Then, taking $t = C^{\prime\prime} s_0 \lambda$ with a large enough constant $C^{\prime\prime} > 0$ such that 
	\[
	\exp \left( - \displaystyle \frac{C^{\prime} t^2}{L_j^2  M^6 (K + M \| \bxi \|_{\infty})^2 (s_0 \lambda)^2} \right) \le \delta /4,
	\]
	we have
	\begin{equation*}
		\begin{array}{rl}
			& \PP \left( \left| \displaystyle \sum_{i=1}^n  n^{-1/2} h_{ij} \varepsilon_{ij}(\bbeta^0) \right| \ge t \right) \\
			\le &  \E \left\{ 2 \exp\left( - \displaystyle \frac{Cnt^2}{L_j^2 \sum_{i=1}^n h_{ij}^2}  \right) \right\}  \\
			= & \E \left\{ 2 \exp\left( - \displaystyle \frac{Cnt^2}{L_j^2 \sum_{i=1}^n h_{ij}^2}  \right)  1_{\mathcal{J}} \right\} + \E \left\{ 2 \exp\left( - \displaystyle \frac{Cnt^2}{L_j^2 \sum_{i=1}^n h_{ij}^2}  \right)  1_{\mathcal{J}^c} \right\} \\
			\le & 2 \exp \left( - \displaystyle \frac{C^{\prime} t^2}{L_j^2  M^6 (K + M \| \bxi \|_{\infty})^2 (s_0 \lambda)^2} \right) \PP(\mathcal{J}) + 2 \PP(\mathcal{J}^c) \\
			\le & 2 \times \delta/4 + 2 \times \delta/4 = \delta.
		\end{array}
	\end{equation*}
	Hence, $\sum_{i=1}^n \displaystyle n^{-1/2} h_{ij} \varepsilon_{ij}(\bbeta^0) = \OP(s_0 \lambda) = \oP(1)$. And since $m$ is finite, by \eqref{eq:normal_decomp_En1_op1}, we have ${n}^{1/2} \homega^\T \bE_{n1}(\bbeta^0)(\tilde{\bbeta} - \bbeta^0) = \oP(1)$.

	\medskip
	{
		\noindent {\textit{Goal (b): To show ${n}^{1/2} \homega^\T \bE_{nk}(\bbeta^*) (\tilde{\bbeta} - \bbeta^0) = \oP(1)$, for $k= 2,3, 4$.}}
	}
	
	We next consider ${n}^{1/2} \homega^\T \bE_{n2}(\bbeta^*)(\tilde{\bbeta} - \bbeta^0)$ and note that
	\begin{align} \label{eq:normal_decomp_En2_op1}
		& | {n}^{1/2} \homega^\T \bE_{n2}(\bbeta^*) (\tilde{\bbeta} - \bbeta^0) | \notag \\
		= & \left| \displaystyle \frac{1}{2} \sum_{i=1}^n \sum_{j=1}^m  \displaystyle n^{-1/2} \dot{v}(\mu_{ij}(\bbeta^0)) \homega^\T \bX_i^\T \{ \bG_i^{1/2}(\bbeta^*) - \bG_i^{1/2}(\bbeta^0) \} \hatR^{-1} \bare_j \bare_j^\T \bX_i (\tilde{\bbeta} - \bbeta^0) \varepsilon_{ij}(\bbeta^0) \right| \notag \\
		\le & \displaystyle \frac{1}{2 {n}^{1/2}} \sum_{i=1}^n \sum_{j=1}^m |\dot{v}(\mu_{ij}(\bbeta^0))|  \| \bX_i \homega \|_2 \| \bG_i^{1/2}(\bbeta^*) - \bG_i^{1/2}(\bbeta^0) \|  \times \notag \\ 
		& \qquad  \qquad  \qquad \| \hatR^{-1} \| \times | \bx_{ij}^\T (\tilde{\bbeta} - \bbeta^0) | \times | \varepsilon_{ij}(\bbeta^0) |.
	\end{align}
	By Assumption \ref{assump:var_fun}, $|\dot{v}(\mu_{ij}(\bbeta^0))|$ is bounded. We also have $\| \bX_i \homega \|_2 = \OP(1),  \| \hatR^{-1} \| = \OP(1)$ and $|\bx_{ij}^\T (\tilde{\bbeta} - \bbeta^0)| = \OP(\| \hbeta - \bbeta^0 \|_1) = \OP(s_0 \lambda)$. Since $\bG_i(\bbeta)$ is a diagonal matrix and $\bx_{ij}^\T\bbeta^*$ lies between $\bx_{ij}^\T\tilde{\bbeta}$ and $\bx_{ij}^\T\bbeta^0$, $\| \bG_i^{1/2}(\bbeta^*) - \bG_i^{1/2}(\bbeta^0) \| = \max_{1 \le j \le m} | G_{ij}^{1/2} (\bbeta^*) - G_{ij}^{1/2}(\bbeta^0) | = \OP( \max_{j} |\bx_{ij}^\T(\bbeta^* - \bbeta^0)|) \le \OP( \max_{j} |\bx_{ij}^\T(\tilde{\bbeta} - \bbeta^0)|) = \OP(\| \hbeta - \bbeta^0 \|_1)$, where the last equality is due to \eqref{eq:diff_xbeta_bar}.
	
	{\textit{Case I. When $\varepsilon_{ij}(\bbeta^0)$'s are sub-exponential}}
	
	Since $\E | \varepsilon_{ij}(\bbeta^0) | \le \| \varepsilon_{1j}(\bbeta^0) \|_{\psi_1} \le \max_{1 \le j \le m} \| \varepsilon_{1j}(\bbeta^0) \|_{\psi_1} = \mathcal{O}(1)$, $| \varepsilon_{ij}(\bbeta^0) | = \OP(1)$. 
	
	{\textit{Case II. When $\varepsilon_{ij}(\bbeta^0)$'s are sub-Gaussian}}
	
	Since $\E | \varepsilon_{ij}(\bbeta^0) | \le \| \varepsilon_{1j}(\bbeta^0) \|_{\psi_2} \le \max_{1 \le j \le m} \| \varepsilon_{1j}(\bbeta^0) \|_{\psi_2} = \mathcal{O}(1)$, $| \varepsilon_{ij}(\bbeta^0) | = \OP(1)$.

	Therefore, in both Case I and Case II, for \eqref{eq:normal_decomp_En2_op1}, we have
	\[  
	\left| {n}^{1/2} \homega^\T \bE_{n2}(\bbeta^*) (\tilde{\bbeta} - \bbeta^0) \right| \le \OP \left( n^{-1/2} \times nm \times \| \hbeta - \bbeta^0 \|_1^2 \right) = \OP \left( {n}^{1/2} s_0^2 \lambda^2 \right).
	\]
	By the assumption that $\max(s_0, s^{\prime} \| \baromega^0 \|_1) \lambda^{\prime} \{ \log(p) \}^{1/2} \to 0$ with $\lambda^{\prime} \asymp \| \baromega^0 \|_1 (s_0 \lambda + \| \hatR^{-1} - \barR^{-1} \|)$, we have $s_0^2 \lambda \{ \log(p) \}^{1/2} \to 0$. Since $\lambda \asymp \{ \log(p)/n \}^{1/2}$ for the penalized estimator $\hbeta$, ${n}^{1/2} s_0^2 \lambda^2 \asymp s_0^2 \lambda \{ \log(p) \}^{1/2}$. Hence, ${n}^{1/2} \homega^\T \bE_{n2}(\bbeta^*) (\tilde{\bbeta} - \bbeta^0) = \oP(1)$.
	
	Similar to showing ${n}^{1/2} \homega^\T \bE_{n2}(\bbeta^*) (\tilde{\bbeta} - \bbeta^0) = \oP(1)$, we can show that ${n}^{1/2} \homega^\T \bE_{nk}(\bbeta^*) (\tilde{\bbeta} - \bbeta^0) = \oP(1)$ for $k=3,4$. Thus, $I_{322} = \oP(1)$.
	
	For $I_{323} = {n}^{1/2} \homega^\T \bF_n(\bbeta^*)(\tilde{\bbeta} - \bbeta^0)$, we can show $I_{323} = \oP(1)$ using a decomposition of $\bF_n(\bbeta^*)$ similar to that of $\bE_n(\bbeta^*)$, so we omit the detailed proof here. This completes the proof that $I_{3} = I_{31} + I_{321} + I_{322} + I_{323} = \oP(1)$. Therefore, by Slutsky's Theorem,
	\[
	\displaystyle \frac{{n}^{1/2} \Psi^P(\theta^0)}{\{ (\bomega^0)^\T \bV^0 \bomega^0 \}^{1/2}} = \frac{{n}^{1/2} (\bomega^0)^\T \bar{\bPsi}(\bbeta^0)}{\{ (\bomega^0)^\T \bV^0 \bomega^0 \}^{1/2}} + \oP(1) 
	\]
	converges to $N(0,1)$ in distribution, as $n \to \infty$.
\end{proof}


\subsection{Proof of Theorem 1}

Finally, the detailed proof of the main result, Theorem 1, is provided below. 


\begin{proof}[Proof of Theorem 1]
	
	By the first-order Taylor expansion of $\Psi^P (\hat{\theta})$ around $\theta^0$, 
	\begin{align} \label{eq:taylor_onestep}
		{n}^{1/2} ( \tilde{\theta} - \theta^0 ) & = {n}^{1/2} \left\{  ( \hat{\theta} - \theta^0 ) + \displaystyle \frac{\Psi^P(\hat{\theta})}{\homega^\T \bS (\hbeta) \homega}  \right\} \notag \\
		& = \displaystyle \frac{{n}^{1/2} \Psi^P(\theta^0) }{\homega^\T \bS (\hbeta) \homega} + {n}^{1/2} \left\{    1 + \frac{\dot{\Psi}^P(\theta^*)}{\homega^\T \bS (\hbeta) \homega} \right\} ( \hat{\theta} - \theta^0 ),
	\end{align}
	where $\theta^*$ lies between $\hat{\theta}$ and $\theta^0$. In Lemma \ref{lemma:hat_omega_rate}, we have shown that $\tilde{\bomega}^\T \bS(\hbeta) \tilde{\bomega} = (\baromega^0)^\T \bS^0 \baromega^0 + \oP(1)$. By the definition of $\homega$, $\homega^\T \bS(\hbeta) \homega = \{ \tilde{\bomega}^\T \bS(\hbeta) \tilde{\bomega} \}^{-1} = \{ (\baromega^0)^\T \bS^0 \baromega^0 \}^{-1} + \oP(1) = (\bomega^0)^\T \bS^0 \bomega^0 + \oP(1)$. By Lemma \ref{lemma:leading_normal}, ${n}^{1/2} \Psi^P(\theta^0) / \{ (\bomega^0)^\T \bV^0 \bomega^0 \}^{1/2}$ converges to $N(0, 1)$ in distribution as $n \to \infty$. Hence, 
	\[
	{n}^{1/2} \displaystyle \frac{\Psi^P(\theta^0)}{\homega^\T \bS(\hbeta) \homega} \times \frac{(\bomega^0)^\T \bS^0 \bomega^0}{ \{ (\bomega^0)^\T \bV^0 \bomega^0 \}^{1/2} }
	\]
	should converge to $N(0, 1)$ in distribution as $n \to \infty$.

	Next, we show that ${n}^{1/2} [ 1 + \{ \homega^\T \bS (\hbeta) \homega \}^{-1} \dot{\Psi}^P(\theta^*) ] (\hat{\theta} - \theta^0) = \oP(1)$. Recall that $\theta^*$ here lies between $\hat{\theta} = \bxi^\T \hbeta$ and $\theta^0 = \bxi^\T \bbeta^0$. Then, we rewrite
	\[
	\begin{array}{rcl}
		\left| {n}^{1/2} \left[ 1 + \{ \homega^\T \bS (\hbeta) \homega \}^{-1} \dot{\Psi}^P(\theta^*) \right] (\hat{\theta} - \theta^0) \right| & = & \left| {n}^{1/2} \displaystyle \frac{\homega^\T \bS (\hbeta) \homega + \dot{\Psi}^P(\theta^*)}{\homega^\T \bS (\hbeta) \homega} \bxi^\T (\hbeta - \bbeta^0) \right|.   \\
	\end{array}
	\]
	For convenience, denote $\bbeta^{**} = \hbeta + \homega (\theta^* - \bxi^\T \hbeta)$, and then
	\[
	\begin{array}{rl}
		& \homega^\T \bS (\hbeta) \homega + \dot{\Psi}^P(\theta^*)  \\
		= & \displaystyle \homega^\T \left\{ \bS(\hbeta) +  \left. \frac{\partial \bPsi}{\partial \bbeta^\T} \right|_{\bbeta^{**}} \right\} \homega \\
		= & \homega^\T \left\{ \bS(\hbeta) - \bS(\bbeta^{**}) + \bE_n(\bbeta^{**}) +  \bF_n(\bbeta^{**}) \right\} \homega. 
	\end{array}
	\]
	Since $\theta^*$ lies between $\hat{\theta} = \bxi^\T \hbeta$ and $\theta^0 = \bxi^\T \bbeta^0$, we write $\theta^* = t \bxi^\T \hbeta + (1 - t) \bxi^\T \bbeta^0$ for some $0 \le t \le 1$. Then $\bbeta^{**} - \hbeta = - (1-t) \homega \bxi^\T (\hbeta - \bbeta^0)$, and $| \bx_{ij}^\T ( \bbeta^{**} - \hbeta ) | \le | \bx_{ij}^\T \homega | \times  | \bxi^\T (\hbeta - \bbeta^0) | \le | \bx_{ij}^\T \homega | \cdot \| \bxi \|_2 \| \hbeta - \bbeta^0 \|_2 = \OP(\| \hbeta - \bbeta^0 \|_2)$.
	Similar to showing \eqref{eq:third_part_I321} in the proof of Lemma \ref{lemma:leading_normal}, we can show that
	\[
	\begin{array}{rl}
		& {n}^{1/2}  | \homega^\T \{ \bS(\bbeta^{**}) - \bS(\hbeta)   \} \homega \bxi^\T (\hbeta - \bbeta^0) | \\
		= & {n}^{1/2} \displaystyle \left|  \frac{1}{n} \sum_{i=1}^n \homega^\T \bX_i^\T \{ \bG_{i}^{1/2}(\bbeta^{**}) - \bG_i^{1/2}(\hbeta) \} \hatR^{-1} \bG_i^{1/2}(\bbeta^{**}) \bX_i \homega \bxi^\T (\hbeta - \bbeta^0) \right|  \\
		& + {n}^{1/2} \displaystyle \left|  \frac{1}{n} \sum_{i=1}^n \homega^\T \bX_i^\T \bG_i^{1/2}(\hbeta) \hatR^{-1} \{ \bG_{i}^{1/2}(\bbeta^{**}) - \bG_i^{1/2}(\hbeta) \} \bX_i \homega \bxi^\T (\hbeta - \bbeta^0) \right| \\
		= &  \oP(1).
	\end{array}
	\]
	We continue to prove that ${n}^{1/2} \homega^\T \bE_{n}(\bbeta^{**}) \homega \bxi^\T ( \hbeta - \bbeta^0 ) = \oP(1)$. Given the decomposition of $\bE_{n}(\bbeta) = \bE_{n1}(\bbeta^0) + \sum_{k=2}^4 \bE_{nk}(\bbeta)$ in \eqref{eq:decomp_En}, we first prove ${n}^{1/2} \homega^\T \bE_{n1}(\bbeta^{0}) \homega \bxi^\T ( \hbeta - \bbeta^0 ) = \oP(1)$, followed by the arguments for ${n}^{1/2} \homega^\T \bE_{nk}(\bbeta^{**}) \homega \bxi^\T ( \hbeta - \bbeta^0 ) = \oP(1), ~ k=2,3,4$. By Assumption \ref{assump:var_fun}, $| \dot{v}(\mu_{ij}(\bbeta^0)) |$ and $\| \bG_i^{1/2}(\bbeta^0) \|$ are both bounded. By Assumption \ref{assump:bdd_xomega} and Lemma \ref{lemma:hat_omega_rate}, $\| \bX_i \homega \|_{\infty} = \OP(1)$. The standardized residual $\varepsilon_{ij}(\bbeta^0) = \OP(1)$. By Assumption \ref{assump:cor_mat}, $\| \hatR^{-1} \| = \OP(1)$. In fact, since $|\bxi^\T (\hbeta - \bbeta)| = \oP(1)$, it suffices to show $ \homega^\T \bE_{n1}(\bbeta^0) \homega = \oP(n^{-1/2}) $. The same data splitting argument, where $(\hbeta, \hatR)$ are estimated independently as shown in the proof of Lemma \ref{lemma:leading_normal}, is again invoked here. For $\homega^\T \bE_{n1}(\bbeta^0) \homega$, we can rewrite $\homega^\T  \bE_{n1}(\bbeta^0) \homega = - (2n)^{-1}  \sum_{j=1}^m \sum_{i=1}^n b_{ij} \varepsilon_{ij}(\bbeta^0)$, where $$ b_{ij} = b_{ij}(\hbeta, \hatR, \homega, \bX_i) = \dot{v}(\mu_{ij}(\bbeta^0)) \homega^\T \bX_i^\T \bG_i^{1/2}(\bbeta^0) \hatR^{-1} \bare_j \bare_j^\T \bX_i \homega.$$
	And $ | b_{ij}|  \le  | \dot{v}(\mu_{ij}(\bbeta^0)) | ~ \|  \bX_i \homega \|_2 \| \bG_i^{1/2}(\bbeta^0) \| \times \| \hatR^{-1} \| \times  | \bx_{ij}^\T \homega |.$
	
	{\textit{Case I. When $\varepsilon_{ij}(\bbeta^0)$'s are sub-exponential}}
	
	Given $(\hbeta, \homega, \hatR, \{ \bX_i \}_{i=1}^n)$, by Bernstein's inequality (see Theorem 2.8.2 of \citet{vershynin2018high}),we have the following concentration inequality for $j = 1, \ldots, m$,
	\[
	\PP \left( \left| \displaystyle \sum_{i=1}^n   b_{ij} \varepsilon_{ij}(\bbeta^0) \right| \ge t \mid \hbeta, \homega, \hatR, \{  \bX_i \}_{i=1}^n \right) \le 2 \exp \left\{ - \displaystyle C \min \left( \frac{ t^2}{(L_j^{\prime})^2 \sum_{i=1}^n b_{ij}^2}, \frac{ t}{L_j^{\prime} \max_i |b_{ij}|} \right) \right\},
	\]
	where $C>0$ is some constant and $L_j^{\prime} = \| \varepsilon_{1j}(\bbeta^0) \|_{\psi_1} = {O}(1)$. For any $\delta >0$, there exists a constant $M>0$ such that $\PP(\mathcal{J}^{\prime}) \ge 1-\delta/4$, where $\mathcal{J}^{\prime} = \{ \| \bX_i \homega \|_2 \le M, i=1, \ldots, n \} \cap \{ \| \hatR^{-1} \| \le M \}$. On the event $\mathcal{J}^{\prime} = \{ \| \bX_i \homega \|_2 \le M, i=1, \ldots, n \} \cap \{ \| \hatR^{-1} \| \le M \}$,
	\[
	\begin{array}{rcl}
		\displaystyle \frac{ t^2}{(L_j^{\prime})^2 \sum_{i=1}^n b_{ij}^2} & \gtrsim & \displaystyle \frac{t^2}{(L_j^{\prime})^2 n M^6},  \\
		\displaystyle  \frac{ t}{L_j^{\prime} \max_i |b_{ij}|} & \gtrsim & \displaystyle \frac{t}{L_j^{\prime} M^3}.
	\end{array}
	\]
	Let $t = C^{\prime} {n}^{1/2}$ with a large enough constant $C^{\prime}>0$ such that, on the event $\mathcal{J}^{\prime}$, we have 
	\[
	\exp \left\{ - \displaystyle C  \min \left( \frac{ t^2}{(L_j^{\prime})^2 \sum_{i=1}^n b_{ij}^2}, \frac{ t}{L_j^{\prime} \max_i |b_{ij}|} \right) \right\} \le \delta/4.
	\]
	Then, 
	\[
	\begin{array}{rl}
		\PP \left( \left| \displaystyle \sum_{i=1}^n   b_{ij} \varepsilon_{ij}(\bbeta^0) \right| \ge t \right)   &
		\le  2 \E \left[ 1_{\mathcal{J}^{\prime}} \exp \left\{ - \displaystyle C  \min \left( \frac{ t^2}{(L_j^{\prime})^2 \sum_{i=1}^n b_{ij}^2}, \frac{ t}{L_j^{\prime} \max_i |b_{ij}|} \right) \right\} \right] \\
		& \quad + 2 \E \left[ 1_{(\mathcal{J}^{\prime})^c} \exp \left\{ - \displaystyle C  \min \left( \frac{ t^2}{(L_j^{\prime})^2 \sum_{i=1}^n b_{ij}^2}, \frac{ t}{L_j^{\prime} \max_i |b_{ij}|} \right) \right\} \right] \\
		& \le 2 \times \delta/4 \times \PP(\mathcal{J}^{\prime}) + 2 \times \PP \{ (\mathcal{J}^{\prime})^c \} \\
		& \le \delta.
	\end{array}
	\]
	Hence, $\sum_{i=1}^n b_{ij} \varepsilon_{ij}(\bbeta^0) = \OP({n}^{1/2})$ and $\homega^\T \bE_{n1} (\bbeta^0) \homega = \OP(n^{-1/2}) = \oP(1)$.
	
	{\textit{Case II. When $\varepsilon_{ij}(\bbeta^0)$'s are sub-Gaussian}}
	
	Given $(\hbeta, \homega, \hatR, \{ \bX_i \}_{i=1}^n)$, which again is equivalent to given $(\hbeta,  \hatR, \{ \bX_i \}_{i=1}^n)$, by the general Hoeffding’s inequality (see Theorem 2.6.3 of \citet{vershynin2018high}),we have the following concentration inequality for $j = 1, \cdots, m$,
	\[
	\PP \left( \left| \sum_{i=1}^n b_{ij} \varepsilon_{ij}(\bbeta^0) \right| \ge t ~ | ~ \hbeta, \homega, \hatR, \{ \bX_i \}_{i=1}^n \right) \le 2 \exp \left( - \displaystyle \frac{C t^2}{L_j^2 \sum_{i=1}^n b_{ij}^2} \right),
	\]
	where $L_j = \| \varepsilon_{1j}(\bbeta^0) \|_{\psi_2} = {O}(1)$. For any $\delta > 0$, there exists a constant $M>0$ such that $\PP(\{ \| \bX_i \homega \|_2 \le M, i=1, \ldots, n \} \cap \{ \| \hatR^{-1} \| \le M \} ) \ge 1- \delta/4$. On the event $\mathcal{J}^{\prime} = \{ \| \bX_i \homega \|_2 \le M, i=1, \cdots, n \} \cap \{ \| \hatR^{-1} \| \le M \}$,
	\[
	\PP \left( \left| \sum_{i=1}^n b_{ij} \varepsilon_{ij}(\bbeta^0) \right| \ge t ~ | ~ \hbeta, \homega, \hatR, \{ \bX_i \}_{i=1}^n \right)  \le 2 \exp \left( - \displaystyle \frac{C t^2}{L_j^2 n M^6 } \right).
	\]
	Let $t = C^{\prime} {n}^{1/2}$ for a large enough constant $C^{\prime} > 0$ such that $\exp \{ - C t^2 / (L_j^2 n M^6) \} \le \delta/4$. Then
	\[
	\begin{array}{rl}
		& \PP \left( \left| \sum_{i=1}^n b_{ij} \varepsilon_{ij}(\bbeta^0) \right| \ge t \right)  \\
		\le &  2 \E \left\{ 1_{\mathcal{J}^{\prime}} \exp \left( \displaystyle - \frac{Ct^2}{L_j^2 \sum_{i=1}^n b_{ij}^2} \right) \right\} + 2 \E \left\{ 1_{(\mathcal{J}^{\prime})^c} \exp \left( \displaystyle - \frac{Ct^2}{L_j^2 \sum_{i=1}^n b_{ij}^2} \right) \right\} \\
		\le & 2 \times \delta/4 + 2 \PP\{(\mathcal{J}^{\prime})^c\} \\
		\le & \delta.
	\end{array}
	\]
	So $\sum_{i=1}^n b_{ij} \varepsilon_{ij}(\bbeta^0) = \OP({n}^{1/2})$ and $\homega^\T \bE_{n1} (\bbeta^0) \homega = \OP(n^{-1/2}) = \oP(1)$. 
	
	Now we prove that ${n}^{1/2} \homega^\T \bE_{n2}(\bbeta^{**}) \homega \bxi^\T (\hbeta - \bbeta^0) = \oP(1)$. By the definition of $\bbeta^{**}$, $\bbeta^{**} - \bbeta^0 = \{ \bI - (1-t) \homega \bxi^\T \} (\hbeta - \bbeta^0)$ for some $0 \le t \le 1$. By Assumption \ref{assump:var_fun}, $| G_{ij}^{1/2}(\bbeta^{**}) - G_{ij}^{1/2}(\bbeta^0) | = \OP(|\bx_{ij}^\T (\hbeta - \bbeta^0)|) \le \OP (\| \hbeta - \bbeta^0 \|_1)$. Hence,
	\[
	\begin{array}{rl}
		& | {n}^{1/2} \homega^\T \bE_{n2}(\bbeta^{**}) \homega \bxi^\T (\hbeta - \bbeta^0) | \\
		= & \left|   \displaystyle  \frac{1}{2 {n}^{1/2} } \sum_{i=1}^n \sum_{j=1}^m \dot{v}(\mu_{ij}(\bbeta^0)) \varepsilon_{ij}(\bbeta^0) \homega^\T \bX_i^\T \{ \bG_i^{1/2}(\bbeta^{**}) - \bG_i^{1/2}(\bbeta^0) \} \hatR^{-1} \bare_j \bare_j^\T \bX_i  \homega \bxi^\T (\hbeta - \bbeta^0) \right|  \\
		\le & \displaystyle \frac{1}{2 {n}^{1/2}} \sum_{i=1}^n \sum_{j=1}^m | \dot{v}(\mu_{ij}(\bbeta^0)) | ~ | \varepsilon_{ij}(\bbeta^0) | ~ \| \bX_i \homega \|_2 \| \bG_i^{1/2}(\bbeta^{**}) - \bG_i^{1/2}(\bbeta^0) \| ~ \times \\
		& \qquad \| \hatR^{-1} \| ~ | \bx_{ij}^\T \homega | ~ \| \bxi \|_{2} \| \hbeta - \bbeta^0 \|_2 \\
		\le & C {n}^{1/2} \OP( \| \hbeta - \bbeta^0 \|_1 \| \hbeta - \bbeta^0 \|_2   ) \\
		= & \OP(s_0^{3/2} \log(p) / {n}^{1/2}) = \oP(1).
	\end{array}
	\]
	Due to the assumption that $\max(s_0, s^{\prime} \| \baromega^0 \|_1) \lambda^{\prime} \{ \log(p) \}^{1/2} \to 0$ with $\lambda^{\prime} \asymp \| \baromega^0 \|_1 (s_0 \lambda + r_n \|)$,  it is implied that $s_0^{3/2} \log(p) / {n}^{1/2} \to 0$. And hence, ${n}^{1/2} \homega^\T \bE_{n2}(\bbeta^{**}) \homega \bxi^\T (\hbeta - \bbeta^0) = \oP(1).$ 
	
	Similarly, we can show ${n}^{1/2} \homega^\T \bE_{nk}(\bbeta^{**}) \homega \bxi^\T (\hbeta - \bbeta^0) = \oP(1)$ for $k= 3, 4$. So ${n}^{1/2} \homega^\T \bE_{n}(\bbeta^{**})  \homega \bxi^\T (\hbeta - \bbeta^0) = \oP(1)$. With similar arguments, it is easy to see that ${n}^{1/2} \homega^\T  \bF_{n}(\bbeta^{**})  \homega \bxi^\T (\hbeta - \bbeta^0) = \oP(1)$ and we will omit the proof here. This shows that ${n}^{1/2} [ 1 + \{ \homega^\T \bS (\hbeta) \homega \}^{-1} \dot{\Psi}^P(\theta^*) ] (\hat{\theta} - \theta^0) = \oP(1)$.
	
	By Taylor expansion \eqref{eq:taylor_onestep} and Slutsky's theorem, the conclusion of this theorem holds.
\end{proof}


\section{Further discussion on some technical conditions}
\label{supp-supp:sec:cond}

\subsection{Restricted eigenvalue condition on $\ddot{\ell}_n (\bbeta^0)$}

{Establishing the rates of convergence for the penalized quasi log-likelihood estimator $\hbeta$ is pivotal for our theory. One commonly adopted requirement for this task is the restricted eigenvalue condition, the earlier discussion on which can be found in \citet{raskutti2010restricted} and \citet{buhlmann2011statistics} for high-dimensional linear regression with random designs. In our setting of random designs, Assumptions~\ref{assump:bound_covs}, \ref{assump:var_fun}, \ref{assump:bdd_eigen_covs} and \ref{assump:cor_mat} imply that the following restricted eigenvalue condition on $\ddot{\ell}_n(\bbeta^0)  = \partial \ell^2(\bbeta) / \partial \bbeta \partial \bbeta^\T$, which is identical to Assumption~2.2 in \citet{fang2020test}, holds with probability going to one:
	For any set $\mathcal{S} \subseteq \{ 1, \ldots, p \}$ with $|\mathcal{S}| = s_0$, 
	\[
	\mathrm{RE}(\tau, \mathcal{S}, \ddot{\ell}_n(\bbeta^0)) = \inf \left\{ \displaystyle \frac{\bnu^\T \ddot{\ell}_n(\bbeta^0) \bnu}{\| \bnu_{\mathcal{S}} \|_2^2} : ~ \bnu \in \mR^p, \bnu \ne 0, \| \bnu_{\mathcal{S}^c} \|_1 \le \tau \| \bnu_{\mathcal{S}} \|_1 \right\} \ge \tau_0
	\]
	for some constants $\tau \ge 1$ and $\tau_0 > 0$. Detailed verification can be found below. Hence, the estimation and prediction error rates of the penalized quasi log-likelihood estimator $\hbeta$ can be established, as stated in Lemma 1 in the main text. We omit the proof of Lemma 1, and interested readers may refer to  \citet{van2012quasi} and \citet{fang2020test}.}

The following proposition states that under the current assumptions in Theorem 1, the restricted eigenvalue condition holds for $\ddot{\ell}_n (\bbeta^0)$ with large probability. Its proof is somewhat similar to that of Lemma 6.17 in \citet{buhlmann2011statistics}. Hence, the estimation and the prediction error rates for $\hbeta$ are legitimately established as in Lemma 1. Note that with a canonical link function $g(\cdot)$, as considered in Section 2 of the main text, $\dot{\mu}(\eta) = v(\eta)$ for $\eta \in \mR$, and then 
\[
\displaystyle \frac{\partial^2 \ell_n (\bbeta)}{\partial \bbeta \partial \bbeta^\T} = \displaystyle \frac{1}{n} \sum_{i=1}^n \sum_{j=1}^m v(\bx_{ij}^\T \bbeta) \bx_{ij} \bx_{ij}^\T.
\]

\begin{proposition} \label{prop:re_cond}
	Under Assumptions \ref{assump:bound_covs}, \ref{assump:var_fun}, \ref{assump:bdd_eigen_covs} and \ref{assump:cor_mat}, the following restricted eigenvalue condition for $\ddot{\ell}_n (\bbeta^0)$ holds with probability going to one:
	for any set $\mathcal{S} \subseteq \{ 1, \ldots, p \}$ with $|\mathcal{S}| = s_0$, 
	\[
	\mathrm{RE}(\tau, \mathcal{S}, \ddot{\ell}_n(\bbeta^0)) = \inf \left\{ \displaystyle \frac{\bnu^\T \ddot{\ell}_n(\bbeta^0)\bnu}{\| \bnu_{\mathcal{S}} \|_2^2} : ~ \bnu \in \mR^p, \bnu \ne \mathbf{0}, \| \bnu_{\mathcal{S}^c} \|_1 \le \tau \| \bnu_{\mathcal{S}} \|_1 \right\} \ge \tau_0,
	\]
	for some constant $\tau_0 > 0$, where $\ddot{\ell}_n(\bbeta) = \partial \ell^2(\bbeta) / \partial \bbeta \partial \bbeta^\T$.
\end{proposition}

\begin{proof}[Proof of Proposition \ref{prop:re_cond}]
	
	First, under Assumptions \ref{assump:bound_covs} and \ref{assump:var_fun}, it is straightforward to show that the difference between $\ddot{\ell}_n (\bbeta^0)$ and $\E \{ \sum_{j=1}^m v (\bx_{1j}^\T \bbeta^0) \bx_{1j} \bx_{1j}^\T \}$ is small. To be specific, there exists a constant $C > 0$ such that $\| \sum_{j=1}^m v ( \bx_{ij}^\T \bbeta^0 ) \bx_{ij} \bx_{ij}^\T \|_{\infty} \le m C$. Applying Hoeffding's inequality, we have that for every $k, l \in \{ 1, \ldots, p \}$ and $t > 0$,
	\begin{align*}
		& \PP \left( \left| \displaystyle \frac{1}{n} \sum_{i=1}^n \sum_{j=1}^m v (\bx_{ij}^\T \bbeta^0) x_{ijk} x_{ijl} - \E \left\{ \sum_{j=1}^m v (\bx_{1j}^\T \bbeta^0) x_{1jk} x_{1jl} \right\} \right| \ge t \right) \\
		\le & 2\exp\{ - n t^2 / (2 m^2 C^2) \}.
	\end{align*}
	Then
	\begin{align*}
		& \PP \left( \left\| \displaystyle \frac{1}{n} \sum_{i=1}^n \sum_{j=1}^m v (\bx_{ij}^\T \bbeta^0) \bx_{ij} \bx_{ij}^\T - \E \left\{ \sum_{j=1}^m v (\bx_{1j}^\T \bbeta^0) \bx_{1j} \bx_{1j}^\T \right\} \right\|_{\infty} \ge t \right)  \\
		\le & 2 p^2 \exp\{ - n t^2 / (2 m^2 C^2) \}.
	\end{align*}
	Taking $t \asymp \{ \log(p)/n \}^{1/2}$, we see that $\| \ddot{\ell}_n (\bbeta^0) - \E \{ \sum_{j=1}^m v (\bx_{1j}^\T \bbeta^0) \bx_{1j} \bx_{1j}^\T \} \|_{\infty} = \OP [ \{ \log(p) / n \}^{1/2} ]$.
	
	Second, the expectation $\E \{ \sum_{j=1}^m v (\bx_{1j}^\T \bbeta^0) \bx_{1j} \bx_{1j}^\T \}$ itself satisfies the restricted eigenvalue condition. Without loss of generality, for a positive definite matrix $\bSigma_1$, if its smallest eigenvalue $\lambda_{\mathrm{min}} (\bSigma_1) \ge C$ for some constant $C > 0$, then the restricted eigenvalue condition holds for $\bSigma_1$. To see this, note that since $\| \bnu_{\calS} \|_2 \le \| \bnu \|_2$, 
	\begin{align*}
		& \inf \left\{ \displaystyle \frac{\bnu^\T \bSigma_1 \bnu}{ \| \bnu_{\calS} \|_2^2}: \bnu \in \mR^p, \bnu \ne 0, \| \bnu_{\calS^c} \|_1 \le \tau \| \bnu_{\calS} \|_1 \right\} \\
		\ge & \inf \left\{ \displaystyle \frac{\bnu^\T \bSigma_1 \bnu}{ \| \bnu \|_2^2}: \bnu \in \mR^p, \bnu \ne 0, \| \bnu_{\calS^c} \|_1 \le \tau \| \bnu_{\calS} \|_1 \right\} \\
		\ge & \inf \left\{ \displaystyle \frac{\bnu^\T \bSigma_1 \bnu}{ \| \bnu \|_2^2}: \bnu \in \mR^p, \bnu \ne 0 \right\} \\
		= & \lambda_{\mathrm{min}} (\bSigma_1) \ge C.
	\end{align*}
	By Assumptions \ref{assump:var_fun} and \ref{assump:bdd_eigen_covs}, $$\lambdamin ( \E \{ \sum_{j=1}^m v (\bx_{1j}^\T \bbeta^0) \bx_{1j} \bx_{1j}^\T \} ) \ge K_2^{-1} \lambdamin ( \E \sum_{j=1}^m \bx_{1j} \bx_{1j}^\T ) = K_2^{-1} \lambdamin(\bX_1 \bX_1^\T ) \ge c/K_2, $$ where the constant $c>0$ is the same as in Assumption \ref{assump:bdd_eigen_covs}.  And hence, the restricted eigenvalue condition holds for the expected matrix $\E \{ \sum_{j=1}^m v (\bx_{1j}^\T \bbeta^0) \bx_{1j} \bx_{1j}^\T \}$.
	
	Finally, since $\E \{ \sum_{j=1}^m v (\bx_{1j}^\T \bbeta^0) \bx_{1j} \bx_{1j}^\T \}$ satisfies the restricted eigenvalue condition with \\
	$\mathrm{RE} (\tau, \calS, \E \{ \sum_{j=1}^m v (\bx_{1j}^\T \bbeta^0) \bx_{1j} \bx_{1j}^\T \}) \ge \tau_0 = c/K_2,$ and that $\| \ddot{\ell}_n (\bbeta^0) - \E \{ \sum_{j=1}^m v (\bx_{1j}^\T \bbeta^0) \bx_{1j} \bx_{1j}^\T \} \|_{\infty} = \OP [ \{ \log(p) / n \}^{1/2} ]$, we conclude that the restricted eigenvalue condition also holds for $\ddot{\ell}_n (\bbeta^0)$ with \\ $\mathrm{RE} (\tau, \calS, \ddot{\ell}_n (\bbeta^0)) \ge \tau_0 /2 = c/(2 K_2)$ with probability going to one. Without loss of generality, suppose a positive definite matrix $\bSigma_1$ satisfied the restricted eigenvalue condition with $\tau_0$ and another positive definite matrix $\bSigma_2$ is close to $\bSigma_1$ such that $\| \bSigma_2 - \bSigma_1 \|_{\infty} \le \tilde{\lambda}$ for some small $\tilde{\lambda} > 0$. Then for $|\calS| = s_0$ and $\bnu \in \mR^p$ such that $ \| \bnu_{\calS^c} \|_1 \le \tau \| \bnu_{\calS} \|_1$,
	\begin{equation} \label{eq:tmp_re1}
		| \bnu^\T \bSigma_2 \bnu - \bnu^\T \bSigma_1 \bnu | \le \| \bSigma_2 - \bSigma_1 \|_{\infty} \| \bnu \|_1^2 \le \tilde{\lambda} \| \bnu \|_1^2.
	\end{equation}
	Since $\| \bnu \|_1 = \| \bnu_{\calS} \|_1 + \| \bnu_{\calS^c} \|_1 \le (1 + \tau) \| \bnu_{\calS} \|_1 \le (1+ \tau) {s_0}^{1/2} \| \bnu_{\calS} \|_2$, by the definition of $\mathrm{RE}(\tau, \calS, \bSigma_1)$,
	\begin{equation} \label{eq:tmp_re2}
		\| \bnu \|_1 \le (1 + \tau) {s_0}^{1/2} \| \bnu_{\calS} \|_2 \le (1+\tau) {s_0}^{1/2} \{ \bnu^\T \bSigma_1 \bnu / \tau_0 \}^{1/2}.
	\end{equation}
	Combining \eqref{eq:tmp_re1} and \eqref{eq:tmp_re2}, we have
	\[
	\left| \displaystyle \frac{\bnu^\T \bSigma_2 \bnu}{\| \bnu_{\calS} \|_2} - \frac{\bnu^\T \bSigma_1 \bnu}{\| \bnu_{\calS} \|_2} \right| \le \frac{\tilde{\lambda} (1+\tau)^2 s_0}{\tau_0} \displaystyle \frac{\bnu^\T \bSigma_1 \bnu}{\| \bnu_{\calS} \|_2}, 
	\]
	which implies that 
	\[
	\displaystyle \frac{\bnu^\T \bSigma_2 \bnu}{\| \bnu_{\calS} \|_2^2} \ge \frac{1}{2} \frac{\bnu^\T \bSigma_1 \bnu}{\| \bnu_{\calS} \|_2^2}
	\]
	when $\tilde{\lambda} (1+\tau)^2 s_0 / \tau_0 \le 1/2$. Now with $\tilde{\lambda} \asymp \{ \log(p)/n \}^{1/2}$, $\bSigma_1 = \E \{ \sum_{j=1}^m v (\bx_{1j}^\T \bbeta^0) \bx_{1j} \bx_{1j}^\T \}$ and $\bSigma_2 = \ddot{\ell}_n (\bbeta^0)$, the final conclusion stands with probability going to one.
\end{proof}

\subsection{Estimation of the working correlation matrix $\hatR$}

In Assumption \ref{assump:cor_mat}, it is required that the estimated working correlation matrix $\hatR$ satisfies
\begin{equation} \label{eq:rate_hatR_assump5}
	\| \hatR - \barR \| = \oP[ \{ \max(s_0, s^{\prime}\| \baromega^0 \|_1 ) \{ \log(p) \}^{1/2} \| \baromega^0 \|_1 \}^{-1} ],
\end{equation}
where $\barR \in \mR^{m \times m}$ is some fixed positive definite matrix. 

One example to consider is the moment estimator for $\hatR$ under unstructured working correlation, which is similar to  \citet{balan2005asymptotic}, i.e.,
\[
\hatR = \displaystyle \frac{1}{n} \sum_{i=1}^{n} \bG_i^{-1/2} (\hbeta) (\bY_i - \bmu_i(\hbeta)) (\bY_i - \bmu_i(\hbeta))^\T \bG_i^{-1/2} (\hbeta),
\]
and in this case $\barR = \bR_0$. 

\begin{proposition} \label{prop:hatR}
	For the moment estimator $\hatR$ under unstructured working correlation as defined above, under Assumptions \ref{assump:bound_covs}, \ref{assump:var_fun}, \ref{assump:bdd_eigen_covs} and \ref{assump:cor_mat} (except the requirement \eqref{eq:rate_hatR_assump5} in Assumption \ref{assump:cor_mat}), we have
	\[
	\| \hatR - \bR_0 \| = \OP [ \{ s_0 \log(p)/n \}^{1/2} ].
	\]
\end{proposition}

By Proposition \ref{prop:hatR}, we need
\begin{align} \label{eq:rate_needed}
	& {s_0 \log(p)/n}^{1/2}  \{ \max(s_0, s^{\prime}\| \baromega^0 \|_1 ) \{ \log(p) \}^{1/2} \| \baromega^0 \|_1 \} \notag \\
	= &  \| \baromega^0 \|_1 \max (s_0^{3/2}, s_0^{1/2} s^{\prime} \| \baromega^0 \|_1)  \log(p) / {n}^{1/2} \to 0
\end{align}
in order to satisfy \eqref{eq:rate_hatR_assump5}, the rate requirement in Assumption \ref{assump:cor_mat}. Suppose $\| \baromega^0 \|_1 = {O}(1)$ in \eqref{eq:rate_needed}, then a necessary condition for \eqref{eq:rate_needed} is that $s_0^3 \{ \log(p) \}^2 /n \to 0$. This requirement is stronger than the model sparsity requirement in \citet{van2014asymptotically}, i.e. $s_0^2 \{ \log(p) \}^2 /n \to 0$, for de-sparsified lasso in generalized linear models without repeated measurements. The requirement \eqref{eq:rate_needed} resembles but is slightly weaker than Assumption 4(c) in \citet{yu2021confidence}, which studies confidence intervals for high-dimensional Cox models.

\begin{proof}[Proof of Proposition \ref{prop:hatR}]
	
	This proof follows similar arguments in Example 2 of \citet{wang2011gee}.
	First, we define an intermediate matrix
	\[
	\bR^* = \displaystyle \frac{1}{n} \sum_{i=1}^{n} \bG_i^{-1/2} (\bbeta^0) (\bY_i - \bmu_i(\bbeta^0)) (\bY_i - \bmu_i(\bbeta^0))^\T \bG_i^{-1/2} (\bbeta^0).
	\]
	From the central limit theorem, $\| \bR^* - \bR_0 \| = \OP(n^{-1/2})$. Then, the rest of the task is to characterize the rate of $\| \hatR - \bR^* \|$. We introduce an additional definition of Frobenius norm for a matrix $\bA$, $\| \bA \|_{\rmF} = (\sum_i \sum_j A_{ij}^2)^{1/2}$.
	
	Note that the element difference
	\begin{align*}
		| \hat{R}_{kj} - R^*_{kj} | & \le \left|  \displaystyle \frac{1}{n} \sum_{i=1}^n \displaystyle \frac{(Y_{ik} - \mu_{ik}(\hbeta))(Y_{ij} - \mu_{ij}(\hbeta)) - (Y_{ik} - \mu_{ik}(\bbeta^0))(Y_{ij} - \mu_{ij}(\bbeta^0))}{G_{ik}(\bbeta^0) G_{ij}(\bbeta^0)}  \right|  \\
		& \quad + \left| \displaystyle \frac{1}{n} \sum_{i=1}^n \frac{(Y_{ik} - \mu_{ik}(\hbeta))(Y_{ij} - \mu_{ij}(\hbeta))}{ \{ G_{ik}(\bbeta^0) G_{ij}(\bbeta^0) \}^{1/2} } \cdot \delta_{ijk} \right| \\
		& = I_{kj,1} + I_{kj,2},
	\end{align*}
	where $\delta_{ijk} = { \{ G_{ik}(\bbeta^0) G_{ij}(\bbeta^0) \}^{1/2} / \{ G_{ik}(\hbeta) G_{ij}(\hbeta) \}^{1/2} } - 1$. Then, by Cauchy inequality,
	\[
	\| \hatR - \bR^* \|_{\rmF}^2 \le \sum_{k=1}^m \sum_{j=1}^m ( I_{kj,1} + I_{kj,2} )^2 \le 2 \sum_{k=1}^m \sum_{j=1}^m I_{kj,1}^2 + 2 \sum_{k=1}^m \sum_{j=1}^m I_{kj,2}^2 = I_{n1} + I_{n2}.
	\]
	Now we consider $I_{n1}$. For $I_{kj,1}$, using triangular inequality, we have
	\begin{align*}
		I_{kj,1} & \le \displaystyle \frac{1}{n} \sum_{i=1}^n \frac{| ( \mu_{ik}(\bbeta^0) - \mu_{ik}(\hbeta) ) ( \mu_{ij}(\bbeta^0) - \mu_{ij}(\hbeta) ) |}{ \{ G_{ik}(\bbeta^0) G_{ij}(\bbeta^0) \}^{1/2} }
		+ \frac{1}{n} \sum_{i=1}^n \frac{| ( \mu_{ik}(\bbeta^0) - \mu_{ik}(\hbeta) ) ( Y_{ij} - \mu_{ij}(\bbeta^0) ) |}{ \{ G_{ik}(\bbeta^0) G_{ij}(\bbeta^0) \}^{1/2} } \\
		& \quad + \displaystyle \frac{1}{n} \sum_{i=1}^n \frac{| ( \mu_{ij}(\bbeta^0) - \mu_{ij}(\hbeta) ) ( Y_{ik} - \mu_{ik}(\bbeta^0) ) |}{ \{ G_{ik}(\bbeta^0) G_{ij}(\bbeta^0) \}^{1/2} }  = I_{kj,11} + I_{kj,12} + I_{kj,13}. \\
	\end{align*}
	Then $I_{n1} \le 6 \sum_{k,j} I_{kj,11}^2 + 6 \sum_{k,j} I_{kj,12}^2 + 6 \sum_{k,j} I_{kj,13}^2 = I_{n11} + I_{n12} + I_{n13}$. By Assumption \ref{assump:var_fun}, $|\mu_{ij}(\bbeta^0) - \mu_{ij}(\hbeta)| \le K_1 | \bx_{ij}^\T (\bbeta^0 - \hbeta) |$. And by Cauchy's inequality, and Assumption \ref{assump:var_fun} that $G_{ik}(\bbeta^0)$'s are lower bounded,
	\begin{align*}
		I_{kj,11}^2 & \le \left\{ \displaystyle \frac{1}{n} \sum_{i=1}^n \frac{ (\mu_{ik}(\bbeta^0) - \mu_{ik}(\hbeta))^2 }{G_{ik}(\bbeta^0)}  \right\}  \left\{ \displaystyle \frac{1}{n} \sum_{i=1}^n \frac{(\mu_{ij}(\bbeta^0) - \mu_{ij}(\hbeta))^2}{G_{ij}(\bbeta^0)}  \right\} \\
		& \le C \left\{ \displaystyle \frac{1}{n} \sum_{i=1}^n  |  \bx_{ik}^\T (\bbeta^0 - \hbeta)|^2  \right\} \left\{ \displaystyle \frac{1}{n} \sum_{i=1}^n  |  \bx_{ij}^\T (\bbeta^0 - \hbeta)|^2  \right\},
	\end{align*}
	which gives $$I_{n11} = \displaystyle  6 \sum_{k,j} I_{kj,11}^2 \le 6 C \left\{ \frac{1}{n} \sum_{i=1}^n \sum_{k=1}^m  |  \bx_{ik}^\T (\bbeta^0 - \hbeta)|^2  \right\} \left\{ \frac{1}{n} \sum_{i=1}^n \sum_{j=1}^m  |  \bx_{ij}^\T (\bbeta^0 - \hbeta)|^2  \right\}.$$ By the results in Lemma 1, $I_{n11} \le \OP [ \{ s_0 \log(p) /n \}^2 ]$. Because $\varepsilon_{ij}(\bbeta^0) = Y_{ij} - \mu_{ij}(\bbeta^0) = \OP(1)$, similarly, one can show that $I_{n12} = \OP(s_0 \log(p) / n)$ and $I_{n13} = \OP(s_0 \log(p) / n)$. And hence, $I_{n1} \le I_{n11} + I_{n12} + I_{n13} = \OP(s_0 \log(p) / n)$. 
	
	As for $I_{n2}$, by Cauchy's inequality,
	\begin{align*}
		I_{n2} \le 2 \sum_{k=1}^m \sum_{j=1}^m \left\{ \displaystyle \frac{1}{n} \sum_{i=1}^n \frac{ (Y_{ik} - \mu_{ik}(\hbeta))^2 (Y_{ij} - \mu_{ij}(\hbeta))^2 }{ G_{ik}(\bbeta^0) G_{ij}(\bbeta^0) } \right\}  \left\{ \displaystyle \frac{1}{n} \sum_{i=1}^n \delta_{ijk}^2 \right\}.
	\end{align*}
	Again, since, by the mean value theorem and Assumption \ref{assump:var_fun},  $| \mu_{ij}(\bbeta^0) - \mu_{ij}(\hbeta) | \le K_1 | \bx_{ij}^\T (\bbeta^0 - \hbeta) |$,
	\[
	| Y_{ij} - \mu_{ij}(\hbeta) |^2 \le \{ | Y_{ij} - \mu_{ij}(\bbeta^0) | + K_1 | \bx_{ij}^\T (\hbeta - \bbeta^0) | \}^2 \le 2 \varepsilon_{ij}^2(\bbeta^0) + 2 K_1^2 | \bx_{ij}^\T (\hbeta - \bbeta^0)  |^2.
	\]
	Then in the upper bound for $I_{n2}$, 
	\begin{align*}
		& \displaystyle \frac{1}{n} \sum_{i=1}^n (Y_{ik} - \mu_{ik}(\hbeta))^2 (Y_{ij} - \mu_{ij}(\hbeta))^2 \\
		\le & \displaystyle \frac{2}{n} \sum_{i=1}^n \{ \varepsilon_{ik}^2(\bbeta^0) + K_1^2 | \bx_{ik}^\T (\bbeta^0 - \hbeta) |^2 \}   \{ \varepsilon_{ij}^2(\bbeta^0) + K_1^2 | \bx_{ij}^\T (\bbeta^0 - \hbeta) |^2 \} \\
		\le  &  \displaystyle \frac{2}{n} \sum_{i=1}^n \left\{ \varepsilon_{ik}^2(\bbeta^0) \varepsilon_{ij}^2(\bbeta^0) + K_1^2 K^2 \| \hbeta - \bbeta^0 \|_1^2 \varepsilon_{ij}^2(\bbeta^0) + \right. \\
		& \quad \left. K_1^2 K^2 \| \hbeta - \bbeta^0 \|_1^2 \varepsilon_{ik}^2(\bbeta^0) + K_1^4 K^4 \| \hbeta - \bbeta^0 \|_1^4 \right\} \\
		= &  \displaystyle \frac{2}{n} \sum_{i=1}^n \left\{ \varepsilon_{ik}^2(\bbeta^0) \varepsilon_{ij}^2(\bbeta^0) \right\} + \OP(\| \hbeta - \bbeta^0 \|_1^2) = \OP(1),
	\end{align*}
	where the second inequality is due to Assumption \ref{assump:bound_covs} that $\| \bx_{ij} \|_{\infty} \le K$ and $| \bx_{ij}^\T (\hbeta - \bbeta^0 ) | \le \| \bx_{ij} \|_{\infty} \| \hbeta - \bbeta^0 \|_1 \le K \| \hbeta - \bbeta^0 \|_1$. Note that 
	\[
	\delta_{ijk} = \displaystyle \frac{ \{ G_{ik}(\bbeta^0) G_{ij}(\bbeta^0) \}^{1/2} - \{ G_{ik}(\hbeta) G_{ij}(\hbeta) \}^{1/2} }{ \{ G_{ik}(\hbeta) G_{ij} (\hbeta) \}^{1/2} }.
	\]
	By Taylor expansion, for some $\tilde{\bbeta}^{(ijk)}$ between $\hbeta$ and $\bbeta^0$, 
	\begin{align*}
		\{G_{ik}(\hbeta) G_{ij}(\hbeta) \}^{1/2} & = \{G_{ik}(\bbeta^0) G_{ij}(\bbeta^0) \}^{1/2} + \\ 
		& \quad \displaystyle \frac{  \dot{v}(\bx_{ik}^\T \tilde{\bbeta}^{(ijk)}) \bx_{ik} G_{ij}(\tilde{\bbeta}^{(ijk)}) + \dot{v}(\bx_{ij}^\T\tilde{\bbeta}^{(ijk)}) \bx_{ij} G_{ik}(\tilde{\bbeta}^{(ijk)})  }{ 2 \{G_{ik}(\tilde{\bbeta}^{(ijk)}) G_{ij}(\tilde{\bbeta}^{(ijk)}) \}^{1/2} } \sigma^2 (\hbeta - \bbeta^0).
	\end{align*}
	In the Taylor expansion above, the fact that $G_{ij}(\bbeta) = \sigma^2 v(\bx_{ij}^\T \bbeta)$ is used.
	By Assumption \ref{assump:var_fun} and Lemma 1, $\{ G_{ik}(\hbeta) G_{ij} (\hbeta)  \}^{-1/2} = \OP(1)$, and in the numerator of $\delta_{ijk}$,
	\[
	\left|  \{ G_{ik}(\bbeta^0) G_{ij}(\bbeta^0) \}^{1/2} - \{ G_{ik}(\hbeta) G_{ij}(\hbeta) \}^{1/2} \right| \le \OP(1)  \left\{ | \bx_{ik}^\T (\hbeta - \bbeta^0) | + | \bx_{ij}^\T (\hbeta - \bbeta^0) | \right\}.
	\]
	Hence, 
	\begin{align*}
		I_{n2} & \le \OP(1) n^{-1} \sum_{i=1}^n \sum_{k=1}^m \sum_{j=1}^m \delta_{ijk}^2 \\
		& \le  \OP(1) n^{-1} \sum_{i=1}^n \sum_{k=1}^m \sum_{j=1}^m \left\{ | \bx_{ik}^\T (\hbeta - \bbeta^0) |^2 + | \bx_{ij}^\T (\hbeta - \bbeta^0) |^2 \right\} \\
		& =  \OP(s_0 \log(p) / n).
	\end{align*}
	Since $m$ is a fixed integer not growing with the sample size $n$, 
	\[
	\| \hatR - \bR^* \|_{\rmF}^2 \le I_{n1} + I_{n2} = \OP(s_0 \log(p) / n) ~ ~ \mathrm{and} ~ ~ \| \hatR - \bR^* \|_{\rmF} = \OP [ \{ s_0 \log(p) / n \}^{1/2} ].
	\]
	As $\| \bR^* - \bR_0 \| = \OP(n^{-1/2})$, 
	\[
	\| \hatR - \bR_0 \| \le \| \hatR - \bR^* \| + \| \bR^* - \bR_0 \|
	\le \| \hatR - \bR^* \|_{\rmF} + \| \bR^* - \bR_0 \| = \OP [ \{ s_0 \log(p) / n \}^{1/2} ].
	\]
\end{proof}

\section{Outline of the modified inference procedure for high-dimensional estimating equations}
\label{supp-supp:sec:hdee}

In the simulation studies, the proposed method HDIGEE is compared to a modified version of the inferential method originally designed for high-dimensional estimating equations (HDEE), developed by \citet{neykov2018unified}. The original framework described by \citet{neykov2018unified} entails linear regression via Dantzig selector, instrumental variables regression, graphical models, discriminant analysis and vector autoregressive models; however, it does not include correlated data via estimating equations as an application.  In this section, we outline the modification we have made to accommodate correlated data for HDEE. 

Suppose we are interested in inference on the $j$th element of $\bbeta^0$, i.e. $\beta^0_j$. After obtaining the initial estimator $\hbeta$, a projected estimating function of $\bbeta$ along the direction $\hat{\bnu}$ is constructed as
\[
\phi(\bbeta) = \hat{\bnu}^\T \Psi (\bbeta),
\]
where the direction $\hat{\bnu}$ is defined as the solution to the following optimization problem
\[
\hat{\bnu} = \argmin \{  \| \bnu \|_1:  \| - \bS(\hbeta) \bnu - \be_j \|_{\infty} \le \eta^{\prime} \}
\] 
for some tuning parameter $\eta^{\prime} > 0$, and $\be_j$ is the unit vector with the $j$th element being one. Here, $\Psi (\bbeta)$ is identical to the definition in the main text and $- \bS(\hbeta)$ is approximately the derivative of $\Psi(\bbeta)$ evaluated at $\bbeta = \hbeta$. 

Without loss of generality, we rewrite $\hbeta = (\hat{\beta}_j, \hat{\bbeta}_{-j})$, where $\hat{\bbeta}_{-j}$ is the estimated coefficients without the $j$th element. Then, the final estimator for the parameter of interest $\beta_j^0$ is the Z-estimator $\tilde{\beta}_j$ that is the root of the equation $\phi(\beta_j, \hat{\bbeta}_{-j}) = 0$. Since \citet{neykov2018unified} used a fixed tuning parameter $\eta^{\prime}$ without discussing more realistic data-driven choices, we implemented cross-validation procedures analogous to one used for our proposed method.

\end{document}